\documentclass[11pt]{article}
\usepackage{graphicx}
\usepackage[utf8]{inputenc}
\usepackage{fullpage}
\usepackage{amsmath,amssymb,amsfonts,amsthm,stmaryrd}
\usepackage{thmtools,thm-restate}
\usepackage{enumitem}
\usepackage[colorlinks=true,linkcolor=blue,allcolors=blue]{hyperref}
\usepackage{todonotes}
\usepackage{physics}
\usepackage[ruled,vlined,linesnumbered]{algorithm2e}
\usepackage[capitalize,nameinlink]{cleveref}

\newcommand{\opn}[1]{\operatorname{#1}}

\newcommand{\spn}{\opn{span}}
\newcommand{\Enc}{\opn{Enc}}
\newcommand{\Dec}{\opn{Dec}}
\newcommand{\supp}{\opn{supp}}
\newcommand{\evl}{\opn{ev}}
\newcommand{\poly}{\opn{poly}}
\newcommand{\gInit}{\opn{Init}}
\newcommand{\gTerm}{\opn{Term}}
\newcommand{\gX}{\opn{X}}
\newcommand{\gCX}{\opn{CX}}
\newcommand{\gCCX}{\opn{CCX}}
\newcommand{\gSwap}{\opn{Swap}}
\newcommand{\tran}{\opn{tran}}
\newcommand{\RS}{\opn{RS}}
\newcommand{\NP}{\mathsf{NP}}

\newcommand{\cE}{\mathcal{E}}\newcommand{\cF}{\mathcal{F}}
\newcommand{\cG}{\mathcal{G}}

\newcommand{\cO}{\mathcal{O}}\newcommand{\cP}{\mathcal{P}}
\newcommand{\cR}{\mathcal{R}}

\newcommand{\bF}{\mathbb{F}}

\newcommand{\bN}{\mathbb{N}}

\DeclareMathOperator*{\argmin}{arg\,min}

\newtheorem{theorem}{Theorem}[section]
\newtheorem{lemma}[theorem]{Lemma}
\newtheorem{claim}[theorem]{Claim}
\newtheorem{definition}[theorem]{Definition}

\newtheorem{fact}[theorem]{Fact}
\theoremstyle{remark} 
\newtheorem{remark}[theorem]{Remark}

\newcommand{\npb}[1]{}
\newcommand{\qn}[1]{}
\newcommand{\anu}[1]{}
\newcommand{\lzg}[1]{}

\title{Classical Adversarial Fault-Tolerance and PCPs\thanks{To appear in 67th IEEE Symposium On Foundations Of Computer Science (FOCS) 2026.}}
\author{
  Anurag Anshu \\
  Harvard University \\
  \href{mailto:anuraganshu@fas.harvard.edu}{\texttt{anuraganshu@fas.harvard.edu}}
  \and
  Nikolas P.~Breuckmann \\
  University of Bristol \\
  \href{mailto:niko.breuckmann@bristol.ac.uk}{\texttt{niko.breuckmann@bristol.ac.uk}}
  \and
  Louis Golowich \\
  UC Berkeley \\
  \href{mailto:lgolowich@berkeley.edu}{\texttt{lgolowich@berkeley.edu}}
  \and
  Quynh T.~Nguyen \\
  Harvard University \\
\href{mailto:qnguyen@g.harvard.edu}{\texttt{qnguyen@g.harvard.edu}}
  \and
  Umesh Vazirani \\
  UC Berkeley \\
  \href{mailto:vazirani@cs.berkeley.edu}{\texttt{vazirani@cs.berkeley.edu}}
}
\date{}

\begin{document}

\maketitle

\begin{abstract}
  We show how to compile an arbitrary classical circuit into a fault-tolerant circuit, which performs the desired computation even when an almost-linear number of bits are adversarially chosen and corrupted in each timestep. Using a variant of this fault-tolerance scheme that only detects (rather than corrects) corruptions, we give a new construction of probabilistically checkable proofs (PCPs) for $\NP$ with polylogarithmic query complexity. This PCP construction from fault-tolerance presents a promising candidate for quantization by the work of Anshu, Breuckmann, and Nguyen (STOC'24), who provided a roadmap for constructing quantum PCPs via fault-tolerance.
\end{abstract}

\newpage

\tableofcontents

\newpage


\section{Introduction}
The celebrated PCP theorem~\cite{arora_probabilistic_1998,arora_proof_1998} proves that large computations can be verified extremely quickly, with queries to just a few random bits of a suitably constructed proof. This verification is inherently \emph{robust}, in that it will succeed even if a constant fraction of the proof's bits are corrupted, because the sparse queries are unlikely to land on corrupted bits. It is then natural to ask whether we can \emph{perform}, rather than simply verify, a computation in a manner that is similarly fault-tolerant, i.e.~robust to large errors. In this work, we give an affirmative answer: we show how to compile any classical circuit into a fault-tolerant version, which simulates the original circuit even when an almost-linear number of bits are corrupted in every timestep. Our fault-tolerance holds even when the set of bits, as well as the errors applied, are chosen adversarially. We apply this fault-tolerance scheme to give a new proof of a ``polylog-weak'' PCP theorem, which has polylogarithmic query complexity instead of the (optimal) constant query complexity of the PCP theorem in~\cite{arora_probabilistic_1998,arora_proof_1998}.
\npb{Discuss this. It reads a bit as if the existence of PCPs raised the question of fault-tolerant computation and we're the first to do computational fault-tolerance.} \lzg{Yeah I guess that was kind of the intention? At least in the sense of, we're the first to do computational fault-tolerance that's comparable to PCPs (or more formally, strong enough to recover PCPs).}
\npb{maybe let's discuss in a call to sharpen this a bit.}
\lzg{sounds good}

Multiple aspects of our construction of PCPs via fault-tolerance may be of independent interest. Our construction is elementary in nature, as it primarily consists of manipulating tensor products of Reed-Solomon codes.
The main prior result we rely on is that tensor product codes are locally testable, meaning that the weight of a corruption on a codeword can be estimated using a small number of random queries. 
This testability result itself has a clean combinatorial proof given by~\cite{viderman_combination_2015}. 
We construct our PCPs from such locally testable codes without relying on stronger properties present in many prior PCP constructions, such as local decodability or local correctability.
While some other PCP constructions share many of these aspects, we present these ingredients through the lens of fault-tolerance; see \Cref{sec:priorpcp} for a more detailed comparison to prior PCPs. Furthermore, and perhaps most notably, our PCP construction is an appealing candidate for quantization, as we discuss below.

The connection between fault-tolerance and PCPs was previously explored by \cite{gal_fault_1995}.
They show that if arbitrary logical\footnote{We refer to a (likely non-fault-tolerant) circuit that we want to run as a \emph{logical} circuit, and a fault-tolerant version of it as a \emph{physical} circuit.} circuits can be turned into fault-tolerant circuits with polynomial size overhead and depth $O(\log s)$, where $s$ is the size of the original circuit, then a polylog-weak PCP theorem follows.
They also demonstrate the existence of such a scheme.
However, there are a couple of important caveats:
(a)~their construction is already based on the (stronger) PCP theorem itself, so it does not yield a new proof, and
(b)~they assume the input is provided in an encoding that depends on the logical circuit, and hence can directly include information about the output value.
Our work directly addresses both of these issues, which were left as open questions in \cite{gal_fault_1995}.

Other classical fault-tolerance schemes have been also developed without using PCPs, e.g.~\cite{neumann_probabilistic_1956,dobrushin_upper_1977,pippenger_networks_1985,spielman_highly_1996}, see \Cref{sec:priorft} for a more comprehensive overview.
However, these schemes were only resilient against adversarial errors on at most a polynomially vanishing fraction of bits, and hence are too weak to construct PCPs with subpolynomial query complexity.

In contrast, we emphasize the truly global nature of the corruptions that our fault-tolerance scheme protects against. Specifically, we map a logical circuit acting on $\bar{N}$ logical bits to a fault-tolerant circuit acting on $N\leq\bar{N}^{1+o(1)}$ physical bits, which can protect against adversarial errors acting on $N^{-o(1)}$-fraction of the physical bits in each timestep. This protection provides a super-polynomial improvement over prior schemes, which could only handle an inverse polynomial $N^{-\Omega(1)}$-fraction errors, assuming polynomial space overhead $N\leq\bar{N}^{O(1)}$. Indeed, the traditional approach of repeating each logical bit many times is always vulnerable to errors on $\bar{N}^{-1}$-fraction of the physical bits, which can wipe out all copies of a given logical bit \cite{neumann_probabilistic_1956,dobrushin_upper_1977,pippenger_networks_1985}. We avoid this limitation by performing the entire computation on a joint encoding of all the logical bits, so that no logical bit is localized to a small set of physical bits that can easily be corrupted.


Recently, \cite{anshu_circuit--hamiltonian_2024} suggested that fault-tolerance schemes could provide a route towards proving the \emph{quantum PCP conjecture}~\cite{aharonov2002quantum,aharonov_quantum_2013}, one of the major open questions in quantum complexity theory.
A quantum analogue of the classical PCP theorem, the conjecture would imply that quantum computations can be efficiently verified by querying only a constant number of qubits in an appropriate proof.

While there exist significant barriers to translating existing classical PCP constructions to the quantum setting \cite{aharonov_quantum_2013}, the proposed approach of \cite{anshu_circuit--hamiltonian_2024} seems to avoid these barriers. 
Yet this approach requires fault-tolerance schemes protecting against an almost-linear number of errors with respect to the total number $N$ of qudits.
Achieving such strong error protection was previously an open problem; prior schemes could only protect against a smaller polynomial number~$N^{1-\Omega(1)}$ of errors (see above).
Our work resolves this question in the classical case. A related fault-tolerance scheme protecting against an almost-linear number of errors in the quantum case is given in the follow-up work \cite{breuckmann_fault-tolerant_2026}. There remain open questions from \cite{anshu_circuit--hamiltonian_2024} regarding the mapping from fault-tolerant circuits to PCPs, whose resolution is needed to construct a quantum PCP with nontrivial (i.e.~subpolynomial) query complexity. However, our classical PCPs provide strong evidence of the viability of this approach towards quantum PCPs.


The remainder of this section provides more details on our results and techniques. In \Cref{sec:circuitsinf} below, we begin by informally defining fault-tolerance schemes, as well as a slightly weaker notion of fault-detecting schemes, which only detect corruptions instead of correcting them. The fully rigorous and general definitions can be found in \Cref{sec:prelim}. In \Cref{sec:maininf}, we state our main result providing fault-tolerance/detecting schemes against adversarial errors. We outline the construction and proof of this result in \Cref{sec:schemeinf}. We then present our main result on PCPs in \Cref{sec:pcpinf}, and in \Cref{sec:pcpconstructinf} we describe how we prove this result using our fault-detecting scheme. We compare our results and techniques to prior works in \Cref{sec:priorft,sec:priorpcp}.

\subsection{Background on Circuits and Fault-Tolerance/Fault-Detection}
\label{sec:circuitsinf}
In this paper, we consider circuits with gates of constant fan-in and fan-out. Each circuit acts on some set $N$ of ``dits,'' i.e.~symbols taking a value in some finite field $\bF_q$. The circuit runs over some number $T$ of timesteps. 
Each dit is acted on by a gate (which may be the identity) at every timestep.
We specifically use a gate set that can add or multiply any two dits (see \Cref{def:gates} for details).

Defining fault-tolerance requires some care; we follow the general approach of \cite{nguyen_quantum_2025,he_composable_2025}, which was developed for the (even more delicate) quantum setting. Specifically, our goal is to compile a \emph{logical} circuit $\bar{\cR}$ into a \emph{physical} circuit $\cR$ such that $\cR$ simulates $\bar{\cR}$ even in the presence of noise. However, $\bar{\cR}$ may have input and output dits taking on arbitrary values not known during the compilation. It then seems that even a single corruption on the input or output of $\cR$ may render the output incorrect. To resolve this issue, we specify encoding maps $\Enc_{\mathrm{in}}$ (resp.~$\Enc_{\mathrm{out}}$), which map inputs (resp.~outputs) of $\bar{\cR}$ to inputs (resp.~outputs) of $\cR$. The images of these encoding maps should be \emph{error-correcting codes}, meaning that the encoding map can be inverted even after some dits have been corrupted.

Then loosely speaking, we define a \emph{fault-tolerance} scheme with error thresholds $\lambda_{\mathrm{in}},\lambda_{\mathrm{run}},\lambda_{\mathrm{out}}$ to be a compiler from logical circuits $\bar{\cR}$ to physical circuits $\cR$ satisfying the following: 
for every input $x$ to $\bar{\cR}$, if we run $\cR$ on some input $y$ that differs from $\Enc_{\mathrm{in}}(x)$ on an arbitrary set of $<\lambda_{\mathrm{in}}$ dits, such that an arbitrary set of $<\lambda_{\mathrm{run}}$ dits are corrupted after each timestep, then the resulting output differs from $\Enc_{\mathrm{out}}\circ\bar{\cR}(x)$ on $<\lambda_{\mathrm{out}}$ dits.
We will want to construct schemes with $\lambda_{\mathrm{in}}\geq\lambda_{\mathrm{out}}$, so that we can sequentially compose fault-tolerant circuits while maintaining low error weight.

In our applications to PCPs, it will in fact be sufficient to construct a \emph{fault-detecting} scheme, which weakens the notion of a fault-tolerance scheme to only detect errors, rather than correct them.
Specifically, we introduce an additional ``detection threshold'' parameter $\lambda_{\mathrm{det}}$, and we specify a set of ``detector dits'' in each timestep of our circuit.
A fault-detecting scheme is then like a fault-tolerance scheme, except that the output is allowed to be wrong if $\geq\lambda_{\mathrm{det}}$ detector dits were nonzero in some timestep.
That is, a fault-detecting scheme is a compiler from logical circuits~$\bar{\cR}$ to physical circuits~$\cR$ satisfying the following:
for every input $x$ to $\bar{\cR}$, if we run $\cR$ on some input~$y$ that differs from $\Enc_{\mathrm{in}}(x)$ in an arbitrary set of $<\lambda_{\mathrm{in}}$ dits, such that an arbitrary set of $<\lambda_{\mathrm{run}}$ dits are corrupted after each timestep, then either the resulting output differs from $\Enc_{\mathrm{out}}\circ\bar{\cR}(x)$ in $<\lambda_{\mathrm{out}}$ dits, or else $\geq\lambda_{\mathrm{det}}$ detector dits were nonzero in some timestep.

\subsection{Main Result on Fault-Tolerance/Fault-Detection}
\label{sec:maininf}
Our main result constructing fault-tolerance and fault-detecting schemes is stated below.

\begin{theorem}[Informal statement of \Cref{thm:main}]
  \label{thm:maininf}
  The following statements hold over every constant alphabet $\bF_r$:
  \begin{enumerate}
  \item For every set $\bar{N}$ of logical dits and every $\epsilon>0$, there exists a fault-detecting scheme that maps an arbitrary logical circuit $\bar{\cR}$ acting on dits $\bar{N}$ using time $\bar{T}$ to a fault-detecting physical circuit $\cR$ acting on $|N|\leq|\bar{N}|^{1+\epsilon}$ dits using time $T\leq\bar{T}\cdot(\log|N|)^{O(1/\epsilon)}$, with error thresholds
    \begin{equation*}
      \lambda_{\mathrm{in}} \geq \lambda_{\mathrm{out}} \geq \lambda_{\mathrm{run}} = \lambda_{\mathrm{det}} \geq \frac{|N|}{(\log|N|)^{O(1/\epsilon)}}.
    \end{equation*}
  \item For every set $\bar{N}$ of logical dits, there exists a fault-tolerance scheme that maps an arbitrary logical circuit $\bar{\cR}$ acting on dits $\bar{N}$ using time $\bar{T}$ to a fault-tolerant physical circuit $\cR$ acting on $|N|\leq |\bar{N}|\cdot 2^{O(\log|N|)^{2/3}}$ dits using time $T\leq\bar{T}\cdot 2^{O(\log|N|)^{2/3}}$, with error thresholds
    \begin{equation*}
      \lambda_{\mathrm{in}} \geq \lambda_{\mathrm{out}} \geq \lambda_{\mathrm{run}} \geq \frac{|N|}{2^{O(\log|N|)^{2/3}}}.
    \end{equation*}
  \end{enumerate}
\end{theorem}

Because fault-detection is a weaker requirement than fault-tolerance, in \Cref{thm:maininf} we are able to construct fault-detecting schemes with error thresholds $|N|/\poly\log|N|$, whereas our fault-tolerance schemes have slightly lower error thresholds of $|N|/2^{O(\log|N|)^{2/3}}=|N|^{1-o(1)}$. As a result, the PCPs we obtain from our fault-detecting scheme have query complexity $\poly\log|N|$, whereas our fault-tolerance scheme would only give PCPs with query complexity $|N|^{o(1)}$.

\subsection{Scheme via Tensor Code Switching}
\label{sec:schemeinf}
We now outline our fault-tolerance and fault-detecting schemes in \Cref{thm:maininf}. The two schemes use similar codes and protocols; the main difference is that our fault-tolerance scheme runs a decoder to correct errors, whereas our fault-detecting scheme simply applies local testability to detect errors.

For both schemes, we encode the logical dits in $\bar{N}$ into a tensor product of Reed-Solomon codes, as defined below. As described below, the key insight is that we can efficiently switch between different such tensor codes in a fault-tolerant (or fault-detecting) manner. This code switching gives us the necessary flexibility to perform a general computation with good fault-tolerance/detection.

Recall that a classical (linear error-correcting) code of length $n$, dimension $k$, distance $d$, and alphabet size $q$ is a $k$-dimensional linear subspace $C\subseteq\bF_q^n$ such that every nonzero codeword $c\in C\setminus\{0\}$ has Hamming weight $|c|\geq d$. We summarize these parameters by saying $C$ is a $[n,k,d]_q$ code.

\begin{definition}
  For a subset $E\subseteq\bF_q$ of size $|E|=n$, the the \emph{Reed-Solomon code} $\RS(q,k,E)\subseteq\bF_q^E=\bF_q^n$ is the $[n,k,n-k+1]_q$ code whose codewords are evaluations $(f(x))_{x\in E}$ of polynomials $f\in\bF_q[X]^{<k}$ of degree $<k$. 
\end{definition}

Reed-Solomon codes are useful for fault-tolerant computation because they support ``transversal multiplication'' (also called a ``multiplication property''): for codewords $c,c'\in\RS(q,k,E)$, then the component-wise (i.e.~\emph{transversal}) product $c*c':=(c_xc'_x)_{x\in E}$ is a codeword in the higher-dimensional Reed-Solomon code $\RS(q,2k-1,E)$. All linear codes also support component-wise (i.e.~transversal) addition, as for $c,c'\in C$, then $c+c'\in C$. Because addition and multiplication form a universal gate set over a finite field,
we can perform universal computation with transversal operations on Reed-Solomon codes. Because each transversal operation performs a single gate to each dit in a codeword, it cannot propagate errors between dits in a codeword, and hence is naturally well-suited for achieving fault-tolerance.

We specifically use \emph{systematic} encoding maps for all our codes, meaning that the message equals the restriction of the codeword to an appropriate $k^u$ dits. Hence transversal addition and multiplication on codewords induces addition and multiplication on the underlying message dits.


There are three main issues with trying to achieve a full fault-tolerance scheme (in the sense of \Cref{thm:maininf}) using Reed-Solomon codes:
\begin{enumerate}
\item After too many multiplications, the code dimension $k$ will surpass the block length $n$, at which point the distance vanishes.
\item While transversal gates limit error propagation, we ultimately want to reduce the size of a corruption by correcting, or at least detecting, the errors.
\item It is unclear how to perform gates between different message dits encoded in the same codeword.
\end{enumerate}
We resolve all of these issues by performing our fault-tolerant circuit not on Reed-Solomon codes, but rather on tensor products of such codes, as defined below.

\begin{definition}
  The \emph{tensor product} $C^1\otimes C^2$ of codes $C^1,C^2$ of respective length $n_1,n_2$ consists of all matrices $c\in\bF_q^{n_1\times n_2}$ such that every column lies in $C^1$ and every row lies in $C^2$.
  
  More generally for $u\in\bN$, for $i\in[u]$ let $C^i$ be a $[n_i,k_i,d_i]_q$ code for $i\in[u]$. For $i\in[u]$, we say a \emph{direction-$i$ column} in $\prod_{i\in[u]}[n_i]$ is a subset $[n_i]\times\{j_{-i}\}\subseteq\prod_{i'\in[u]}[n_{i'}]$ (where we push $[n_i]$ to the $i$th factor) for some $j_{-i}\in\prod_{i'\in[u]\setminus\{i\}}[n_{i'}]$. Then the $[\prod_{i\in[u]}n_i,\;\prod_{i\in[u]}k_i,\;\prod_{i\in[u]}d_i]_q$ \emph{tensor product code} $C=\bigotimes_{i\in[u]}C^i$ is the code containing every $c\in\bF_q^{\prod_{i\in[u]}[n_i]}$ whose restriction to every direction-$i$ column is a codeword in $C^i$ for every $i\in[u]$.
\end{definition}

Our fault-detecting and fault-tolerance schemes in \Cref{thm:maininf} use input and output codes $\RS(q,k,E)^{\otimes u}$ for an arbitrary size-$n$ subset $E\subseteq\bF_q$, where $n,k,u$ are appropriately chosen functions of $|\bar{N}|$ to ensure that the code has large distance, as well as large enough dimension $k^u\geq|\bar{N}|$ to store the logical dits. In each timestep of our physical circuit, the logical dits are also encoded in such a tensor product of Reed-Solomon codes.

These tensor products preserve transversal multiplication (and of course, transversal addition) of the underlying Reed-Solomon codes. That is, for $c,c'\in\RS(q,k,E)^{\otimes u}$, then $c*c'\in\RS(q,2k-1,E)^{\otimes u}$. 
However,
we now also have a way to reduce the dimension of the multiplied code from $(2k-1)^u$ back down to the original dimension $k^u$: we sequentially loop through $i=1,\dots,u$, and for each $i$, we unencode the $\RS(q,2k-1,E)$ code in each direction-$i$ column, and then re-encode the resulting message into the code $\RS(q,k,E)$.
This procedure effectively switches the factors $\RS(q,2k-1,E)$ in the tensor product back to $\RS(q,k,E)$ one at a time.


This code switching procedure has one major problem: it is not fault-tolerant, as a single error may propagate to cause many more. Fortunately, we can address this issue by repeatedly performing error-correction on $\bigotimes_{i\in[u]}C^i$ with a similar technique: we can sequentially loop through $i=1,\dots,u$, and for each $i$ we run a decoder for $C^i$ to correct errors within every direction-$i$ column in parallel. In \Cref{lem:errcorr}, we show that this error-correction procedure is fault-tolerant, meaning that it corrects more errors than it introduces, even in the presence of an almost-linear number of new errors at every timestep, as long as $u$ grows as an appropriate function of $n$. Our key observation to prove \Cref{lem:errcorr} is that the new errors occurring during error-correction can only propagate within some appropriately defined \emph{closure} set, which is not too large by a bound of \cite{kalachev_maximally_2025}. Outside of this closure set, the error-correction proceeds noiselessly as if no new errors had occurred, and hence will successfully correct any input error that is not too high-weight. Hence the error-correction procedure only leaves a small residual error bounded by the closure set.

If we only need to detect errors rather than correct them (as in the fault-detection scheme in \Cref{thm:maininf}, we can instead simply apply the result of \cite{viderman_combination_2015} showing that all tensor product codes of sufficiently large distance are locally testable. Hence to detect errors on $\bigotimes_{i\in[u]}C^i$, \cite{viderman_combination_2015} shows it is sufficient to loop through $i=1,\dots,u$, and check if every direction-$i$ column lies in $C^i$; if most checks pass, any corruption on our codeword must have low weight. This procedure can detect more errors than our error-correction procedure, leading to the higher error thresholds for fault-detection than fault-tolerance in \Cref{thm:maininf}.

Related ideas involving code switching in one direction at a time has previously been used in the quantum fault-tolerance literature (e.g.~\cite{bombin_dimensional_2016,jochym-oconnor_fault-tolerant_2019,tan_single-shot_2025,golowich_constant-overhead_2025,xu_batched_2025}), and is also reminiscent of the classical sum-check protocol (see \Cref{sec:priorpcp}).

We have now seen how to perform addition, multiplication, and error-correction on code words in a fault-tolerant or fault-detecting manner. However, these addition and multiplication operations act transversally, meaning if we add or multiply codewords $c,c'$ that (systematically) encode messages $m,m'$, then we can only add or multiply $m_j$ with $m'_{j'}$ if $j=j'$.

To implement logical circuits with general connectivity, we again turn to code switching. Specifically, given a codeword in $\bigotimes_{i\in[u]}C^i$, we may choose some $i\in[u]$ and then unencode every direction-$i$ column to obtain $k=\dim(C^i)$ codewords of $\bigotimes_{i'\in[u]\setminus\{i\}}C^{i'}$. We may then perform transversal gates between these $k$ codewords, before re-encoding in every direction-$i$ column to return to our original tensor code $\bigotimes_{i\in[u]}C^i$. By repeating this procedure in all $u$ directions $i\in[u]$, interspersed with rounds of error correction/detection as described above, we obtain fault-tolerant/detecting circuits with hypercubic connectivity. That is, our logical dits are arranged in a $u$-dimensional hypercube~$[k]^u$, and we can perform gates between dits labeled $j,j'\in[k]^u$ if $j$ and $j'$ only differ in a single component $i$ (by unencoding in direction $i$ and then performing transversal gates).

To compile a general logical circuit to respect this connectivity, we apply the sorting networks of \cite{batcher_sorting_1968} (see \Cref{lem:route}), which show how to arbitrarily rearrange dits in the hypercube $[k]^u$ using a $\poly\log(k^u)$-depth circuit, in which each timestep performs transversal gates in some direction $i\in[u]$. Hence for every gate that we want to perform in some timestep of the logical circuit, we permute the logical dits acted upon by the gate to be neighbors in the hypercube $[k]^u$, and then we implement the gate transversally. Recall that every transversal gate here in turn is performed by unencoding the code $C^i$ in every direction-$i$ column, performing transversal addition or multiplication, re-encoding in direction-$i$, and then performing direction-by-direction dimensional reduction (if we performed a multiplication) and error-correction or error-detection as appropriate.

Combining all of the components described above, we obtain \Cref{thm:maininf}. Note that while Reed-Solomon codes, and hence their tensor products, have growing alphabet size $\bF_q$, we stated \Cref{thm:maininf} over a constant-sized alphabet $\bF_r$. Therefore to prove \Cref{thm:maininf}, we construct our fault-tolerance/detecting scheme over an extension field $\bF_q=\bF_{r^\kappa}\cong\bF_r^\kappa$, and then we replace each $q$-ary dit with a size-$\kappa$ block of $r$-ary dits. This alphabet reduction slightly degrades the error thresholds relative to the number of dits, though the loss is insignificant in our parameter regimes.

\subsection{Comparison to Prior Fault-Tolerance Schemes}
\label{sec:priorft}
The study of fault-tolerant computation was initiated by von Neumann \cite{neumann_probabilistic_1956}, who proposed repeating each bit many times (i.e.~encoding in a repetition code), and correcting errors by comparing the different copies of a bit. Follow-up works later made this scheme rigorous and explicit \cite{dobrushin_upper_1977,pippenger_networks_1985}. This scheme can with high probability correct random errors using only logarithmic redundancy. However, if an adversary is allowed to choose the error locations, they can flip all copies of a single bit to corrupt the output. Hence the computation remains vulnerable to corruptions on localized sets of bits. More explicitly, if there are $|\bar{N}|$ logical bits and $|N|$ physical bits, the scheme is robust to at most $|N|/|\bar{N}|$ adversarial errors.

The repetition code scheme can only achieve exponential error suppression against random errors by introducing polynomial space overhead $|N|/|\bar{N}|=\poly(|\bar{N}|)$. Spielman \cite{spielman_highly_1996} showed how to reduce this space overhead to $\poly\log(|\bar{N}|)$ using a scheme based on a tensor product of Reed-Solomon codes. Spielman's scheme shares many characteristics with ours: it uses a degree reduction technique similar to our code switching following transversal multiplications, and it also transforms the logical circuit into one with hypercubic connectivity. However, \cite{spielman_highly_1996} only takes the product of $u=2$ Reed-Solomon codes, and uses a different error-correction scheme, which leads to a tolerance of at most $\tilde{O}(|N|^{1/4})$ adversarial errors. In contrast, by taking the product of a growing number $u$ of Reed-Solomon codes, and decoding in one direction at a time, we are able to protect against an almost-linear number $|N|^{1-o(1)}$ adversarial errors. \cite{spielman_highly_1996} also requires the Reed-Solomon evaluation points to respect a particular symmetry in order to obtain appropriate logical permutations. Our scheme instead works for arbitrary evaluation points, as we unencode out of a factor code before permuting along the associated axis in the hypercube.

As described above, G\'{a}l and Szegedy \cite{gal_fault_1995} constructed a sort of fault-tolerance scheme against an almost-linear number of adversarial errors. However, this scheme relied on PCPs, and also does not match our definition of a fault-tolerance scheme, as it allowed the input encoding to contain information about the logical output.

Other lines of work have also studied different error models, such the short-circuit model, where every corrupted gate must output one of its input bits \cite{kleitman_design_1997,kalai_formulas_2012,braverman_optimal_2019,efremenko_circuits_2022}. In particular, in this model \cite{efremenko_circuits_2022} obtained fault-tolerance against adversarial errors acting on a constant fraction of gates along every path through the circuit. However, short-circuit errors are much more restrictive than the general errors that we permit. In particular, the works above obtain fault-tolerance against short-circuit errors without needing to encode the logical input, whereas fault-tolerance in our model is impossible if the input is unencoded (see \Cref{sec:circuitsinf}).

The study of robust multi-party computation (MPC), which dates back to \cite{ben-or_completeness_1988}, provides a form of fault-tolerance in which several parties communicate to collectively compute a function, even when some parties may become corrupted. Our fault-tolerance scheme uses certain related techniques, including degree reduction on Reed-Solomon codes and routing networks, as some such MPC protocols (e.g.~\cite{damgard_perfectly_2010}). However, we emphasize that our notion of fault-tolerance requires robustness against faults that corrupt different sets of dits (i.e.~parties) in different timesteps, whereas MPC protocols typically assume all parties outside of a ``bad'' set remain uncorrupted throughout the entire computation. Nevertheless, the MPC notion of fault-tolerance is sufficient for various applications, including for constructing certain proof systems (see e.g.~\cite{boneh_quasi-optimal_2018}). It is an interesting direction of future work to further develop the connection between our results and robust MPC protocols.

Our fault-tolerance scheme also shares some characteristics with certain quantum fault-tolerance schemes. In particular, the idea of code switching in one direction at a time to obtain universal computation or to implement logical permutations has been used in various schemes (e.g.~\cite{bombin_dimensional_2016,jochym-oconnor_fault-tolerant_2019,tan_single-shot_2025,golowich_constant-overhead_2025,xu_batched_2025}). However, these schemes relied on quantum codes of distance at most $O(\sqrt{|N|})$, and hence could not protect against more than $O(\sqrt{|N|})$ adversarial errors.

\subsection{Main Result on PCPs}
\label{sec:pcpinf}
We now state our main result on PCPs; we will subsequently describe how we prove it using our fault-detecting scheme in \Cref{thm:maininf}. We briefly review the necessary definitions in this section; see \Cref{sec:pcp} for the details.

We will describe PCPs in terms of \emph{constraint satisfaction problems} (CSPs). For the purpose of this paper, a CSP instance $\cP$ with $n$ boolean variables and $m$ constraints consists of a set of functions $\cP=\{\cP_1,\dots,\cP_m\}$, where each function $\cP_j:\bF_2^{B_j}\rightarrow\bF_2$ maps some subset $B_j\subseteq[n]$ of the variables to a boolean value. We restrict attention to constant-locality CSPs, meaning that every $|B_j|=O(1)$, and evey variable lies in $O(1)$ sets $B_j$. We say a constraint $\cP_j$ is \emph{satisfied} by an assignment $x\in\bF_q^{[n]}$ of the variables if $\cP_j(x|_{B_j})=0$. We say $\cP$ is \emph{satisfiable} if there exists an assignment that satisfies every constraint.


We now state our main result on PCPs.

\begin{theorem}[PCPs for circuit-satisfiability with inverse-polylogarithmic soundness; informal statement of \Cref{thm:pcp}]
  \label{thm:pcpinf}
  For every $\epsilon>0$, there exists a polynomial-time algorithm that takes as input a circuit $\cR$ of size $s$ that has a single output bit, and outputs a constant-locality CSP instance $\cP$ with $n\leq s^{1+\epsilon}$ variables and $m=\Theta(n)$ constraints, which exhibits:
  \begin{enumerate}
  \item (Completeness) If there exists an input $x$ such that $\cR(x)=0$, then $\cP$ is satisfiable.
  \item (Soundness) If every input $x$ has $\cR(x)=1$, then every assignment for $\cP$ has $\geq 1/(\log m)^{O(1/\epsilon)}$ unsatisfied constraints.
  \end{enumerate}
\end{theorem}

To interpret \Cref{thm:pcpinf} as providing a ``probabilistically checkable proof'' (PCP) that a given circuit $\cR$ has a an input $x$ with $\cR(x)=0$, we view a satisfying assignment to the associated CSP instance $\cP$ as the ``proof.'' \Cref{thm:pcpinf} implies that if we sample and check a single random constraint from $\cP$, we will successfully detect if $\cP$ is unsatisfiable with probability $1/(\log m)^{O(1/\epsilon)}$ while querying just a constant number of variables. If we repeat this test $(\log m)^{O(1/\epsilon)}$ many times, we boost the test's success probability to $1-1/2^{(\log m)^{O(1/\epsilon)}}$, while increasing the query complexity to polylogarithmic. Hence we can view \Cref{thm:pcpinf} as providing polylog-query PCPs for the $\NP$-complete problem of circuit-satisfiability.

The PCP proof (i.e.~list of variable assignments) in \Cref{thm:pcpinf} has length $n\leq s^{1+\epsilon}$ for arbitrarily small constant $\epsilon>0$. In contrast, state-of-the-art PCP constructions have length growing as $s\poly\log s$, while only requiring a constant number of queries to verify the proof with constant success probability \cite{dinur_pcp_2007}.

\subsection{PCP Construction}
\label{sec:pcpconstructinf}
We now describe how we construct PCPs to prove \Cref{thm:pcpinf} from our fault-detecting scheme in \Cref{thm:maininf}.
The main idea behind mapping a fault-tolerance scheme to a PCP was outlined in \cite[Appendix~A]{anshu_circuit--hamiltonian_2024}. We adapt this technique to fault-detecting schemes to obtain the $|N|/\poly\log|N|$ error thresholds for fault-detection in \Cref{thm:maininf}.


We begin by mapping the circuit $\cR$ in \Cref{thm:pcpinf} to a circuit $\bar{\cR}$ using constant time (i.e.~constant depth) such that $\cR$ has an input $x$ with $\cR(x)=0$ iff $\bar{\cR}$ has an input $z$ with all output bits of~$\bar{\cR}(z)$ equal to $0$. We will then apply \Cref{thm:maininf} to the logical circuit $\bar{\cR}$ to obtain a fault-detecting physical circuit $\cR'$. Our final CSP instance $\cP$ in \Cref{thm:pcpinf} will then have variables corresponding to bits at every point in space and time in a transcript of the execution of~$\cR'$. 
The constraints of~$\cP$ will enforce that the transcript describes a valid execution of~$\cR'$ with all-$0$s output. 
If a small number of constraints are violated, the transcript describes an execution of~$\cR'$ with a small number of errors, which by the fault-detecting property still has the correct output (up to a low-weight corruption). An incorrect output can only be induced by a large number of errors in the execution of $\cR'$, which would result in many violated constraints of $\cP$. Hence completeness and soundness follow.

We now provide a more detailed description of each of these steps. To construct $\bar{\cR}$, we first define the \emph{transcript} $\tran(\cR;x)$ of $\cR$ on input $x$ to be a list of the values of every bit in $\cR$ during every timestep of the execution of $\cR$ on input $x$. Therefore $\tran(\cR;x)\in\bF_2^S$, where $S$ is a list of all space-time locations in $\cR$.\footnote{In \Cref{sec:pcp}, we actually define $S$ to omit idling locations where identity gates are applied, so that $s$ measures the number of non-identity gates. For simplicity in this informal presentation we ignore this detail.} In particular, $|S|=\Theta(s)$ grows linearly in the size of $\cR$.

We then let $\bar{\cR}$ take as input some $z\in\bF_2^S$. If $\cR$ has $g$ gates, then $\bar{\cR}$ has $g+1$ outputs. Viewing~$z$ as a transcript of the execution of $\cR$ under the presence of noise, then the first output of $\bar{\cR}$ equals the bit outputted by this execution of $\cR$ (i.e.~the unique bit in the final timestep of the transcript). The remaining $g$ outputs of $\bar{\cR}$ each equal $0$ iff the $g$ respective gates in the transcript were executed correctly, meaning that the gate's outputs equal the gate's image under its inputs.

Therefore if $z$ equals a transcript of a noiseless execution of $\cR$ on some input $x$ with output $\cR(x)=0$, then all outputs of $\bar{\cR}(z)$ are $0$; otherwise, $\bar{\cR}(z)$ has a nonzero output. Because every gate in $\cR$ acts on a constant number of bits (as we assume constant fan-in and fan-out), each input to $\bar{\cR}$ only affects a constant number of outputs, and each output only relies on a constant number of inputs. Thus we can implement $\bar{\cR}$ using $\Theta(s)$ bits in constant time (i.e.~constant depth).

The key step in our PCP construction is now to apply our fault-detecting scheme in \Cref{thm:maininf} to the logical circuit $\bar{\cR}$, in order to obtain a fault-detecting physical circuit $\cR'$. By \Cref{thm:maininf}, the set $S'$ of space-time locations in $\cR'$ will have size $|\cR'|=s^{1+O(\epsilon)}$, and $\cR'$ will use time $T'\leq(\log s)^{O(1/\epsilon)}$. Our final CSP instance $\cP$ will have variables $y\in\bF_2^{S'}$ labeled by the set $S'$. Viewing~$y$ as the transcript of an execution of $\cR'$ under the presence of noise, then we define $\cP$ to have constraints that check:
\begin{enumerate}
\item (Fault constraints) That every gate in $\cR'$ was executed correctly,
\item (Detector constraints) That every detector bit had value $0$, and
\item (Output constraints) That the output of $\cR'$ was all-$0$s.
\end{enumerate}
Each fault constraint acts on a constant number of bits because we assume all gates have constant fan-in and fan-out. Meanwhile, each detector and output constraint just checks that a single bit is nonzero. Therefore $\cP$ has constant locality.

To see that the completeness in \Cref{thm:pcpinf} holds, if there exists $x$ with $\cR(x)=0$, then we can set $z=\tran(\cR;x)$, and then we set $y=\tran(\cR';\Enc_{\mathrm{in}}(z))$ to be the transcript of $\cR'$ on input given by the encoding (associated to the tensor Reed-Solomon codes from our fault-detecting scheme; see \Cref{sec:circuitsinf}) of $z$. By construction $\bar{\cR}(z)=0$, and hence all constraints in $\cP$ will be satisfied by assignment $y$.

We now turn to showing the soundness in \Cref{thm:pcpinf}. Suppose that $y$ is an assignment for $\cP$ with few (i.e.~$<1/(\log m)^{O(1/\epsilon)}$-faction) unsatisfied constraints. Viewing $y$ as the transcript for a noisy execution of $\cR'$, then it follows that most gates in $\cR'$ executed noiselessly, most detector bits did not detect an error, and most output bits were $0$. Therefore if the restriction of $y$ to the first timestep (corresponding to the input to $\cR'$) is close to some codeword $\Enc_{\mathrm{in}}(z)$ of the input code of our fault-detecting scheme, then the fault-detecting property implies that the restriction of $y$ to the last timestep (corresponding to the output of $\cR'$) must be close to the encoding $\Enc_{\mathrm{out}}\circ\bar{\cR}(z)$ of $\bar{\cR}(z)$. Assuming our output code has high distance, then $\Enc_{\mathrm{out}}\circ\bar{\cR}(z)$ has large Hamming weight unless $\bar{\cR}(z)=0$. Hence by the assumption that most output constraints were satisfied, we must have $\bar{\cR}(z)=0$, which by the definition of $\bar{\cR}$ means that $\cR(x)=0$ for some $x$, as desired for soundness.

There is only one remaining case to consider, namely, if the restriction of $y$ to the first timestep is far from all codewords in the image of $\Enc_{\mathrm{in}}$. That is, $y$ could be the transcript of an execution of $\cR'$ with a large input error, but then with few additional errors during the execution. To guard against such large input errors, we use the local testability of our tensor Reed-Solomon \cite{viderman_combination_2015} codes to test the weight of the input error, and add detector constraints to enforce a low-weight error. Note that we formalize this idea using the ``mending'' property in \Cref{def:faultdet} below.

\subsection{Comparison to Prior PCP Constructions}
\label{sec:priorpcp}
Many PCP constructions are fundamentally based on similar ingredients as ours, namely, locally testable codes with sufficient algebraic structure to encode a computation. Indeed, multiple PCP constructions specifically rely on the local testability of tensor products of Reed-Solomon (or related) codes, e.g.~\cite{ben-sasson_short_2008,meir_combinatorial_2012,ben-sasson_constant_2016}. Our fault-tolerance/detecting schemes that operate on tensor codes in one direction (i.e.~one factor in the product) at a time are reminiscent of the sum-check protocol \cite{lund_algebraic_1992}, which was generalized to operate on arbitrary tensor codes in \cite{meir_ip_2013}, and is used in PCP constructions such as \cite{arora_probabilistic_1998,bensasson_robust_2006,ben-sasson_constant_2016,gur_perfect_2024,gur_zero-knowledge_2025}. Indeed, the sum-check protocol also operates on one factor of the tensor product at a time. However, whereas sum-check is only used to check the sum of logical dits, we use code switching to apply general logical circuits. This key idea was largely inspired by applications of code switching in the quantum fault-tolerance literature (see \Cref{sec:priorft}).

A recent work~\cite{bafna2025quasi} also constructed PCPs using fault-tolerance ideas. However, their objective and techniques are substantially different from ours. There, the authors constructed a fault-tolerant communication protocol on a high-dimensional expander that gives rise to small soundness and 2-query PCPs.


\section{Preliminaries}
\label{sec:prelim}
This section presents preliminary notions and results.

\subsection{Notation}
\label{sec:notation}
In this section we describe basic notation that we will use throughout the paper. For $n\in\bN$, we write $[n]=\{1,2,\dots,n\}$. For a prime power $q$, we let $\bF_q$ denote the finite field of order $q$. For a vector $x\in\bF^n$, we let $\supp(x)=\{i\in[n]:x_i\neq 0\}$ denote the support, and we let $|x|=|\supp(x)|$ denote the Hamming weight (or simply the ``weight'') of $x$. 

For vectors $x,y\in\bF^n$, we let $x\cdot y=\sum_{i\in[n]}x_iy_i$ denote the standard bilinear form, and we let $x*y=(x_iy_i)_{i\in[n]}\in\bF^n$ denote the component-wise product. We extend this notation to subspaces $A,B\subseteq\bF^n$ by letting $A*B=\spn\{a*b:a\in A,b\in B\}$. We let $I_n:\bF_q^n\rightarrow\bF_q^n$ denote the identity map $I_n(x)=x$.

For sets $A,B$, we let $A\sqcup B$ denote the disjoint union, so that for instance $A\sqcup A$ is isomorphic to $A\times[2]$. We extend this notation to powers, so that $A^{\sqcup n}=A\sqcup\cdots\sqcup A\cong A\times[n]$ denotes the disjoint union of $n$ copies of $A$. We also use the notation `$\sqcup$' for direct products of functions, so that if $f_i:A_i\rightarrow B_i$ is a function for $i\in[2]$, then $f_1\sqcup f_2:A_1\times A_2\rightarrow B_1\times B_2$ is the function $(f_1\sqcup f_2)(a_1,a_2)=(f_1(a_1),f_2(a_2))$.

\subsection{Error-Correcting Codes}
In this section we present definitions and basic results regarding error-correcting codes, which we use to construct our fault-tolerance scheme.

\begin{definition}
  A \emph{classical linear code} (or simply \emph{code}) $C$ over the alphabet $\bF_q$ of \emph{length} $n$ and \emph{dimension} $k$ is a $k$-dimensional subspace $C\subseteq\bF_q^n$. The \emph{distance} of $C$ is the minimum weight $d=\min_{c\in C\setminus\{0\}}|c|$ of a nonzero element of $C$. We summarize these parameters by saying that $C$ is a $[n,k]_q$ or a $[n,k,d]_q$ code.

  An \emph{encoding map} for $C$ is a linear isomorphism $\Enc:\bF_q^k\xrightarrow{\sim}C$.
  An encoding map is \emph{systematic} if there exists a subset $K\subseteq[n]$ of size $|K|=k$ such that the restriction of the encoding to dits in $K$ equals the message, i.e.~for every $x\in\bF_q^k$ we have $\Enc(x)|_K=x$ (under some fixed isomorphism $K\cong[k]$).
\end{definition}

The following fact follows from a basic Gaussian elimination argument.

\begin{fact}
  \label{fact:systematic}
  For every $[n,k]_q$ code $C$, there exists a set $K\subseteq[n]$ such that there is a systematic encoding map $\Enc:\bF_q^k\cong\bF_q^K\rightarrow\bF_q^n$ for which $\Enc(x)|_K=x$. Furthermore, for every set $K_0\subseteq[n]$ such that $C|_{K_0}=\bF_q^{K_0}$, then $K$ can be chosen to contain $K_0$.
\end{fact}

The primary codes we consider are tensor products of Reed-Solomon codes, defined below.

\subsubsection{Tensor Product Codes}
Here we define tensor product codes.

\begin{definition}
  Let $q$ be a prime power and let $u\in\bN$. For $i\in[u]$, let $C_i$ be a $[n_i,k_i,d_i]_q$ code with encoding map $\Enc_i$. The \emph{tensor product code} (or simply \emph{tensor code}) $C=\bigotimes_{i\in[u]}C_i$ is the $[\prod_{i\in[u]}n_i,\; \prod_{i\in[u]}k_i,\; \prod_{i\in[u]}d_i]_q$ code
  \begin{equation*}
    C = \spn\left\{\bigotimes_{i\in[u]}c_i:c_i\in C_i\;\forall i\in[u]\right\},
  \end{equation*}
  with \emph{tensor product encoding map} $\Enc=\bigotimes_{i\in[u]}\Enc_i$ given by
  \begin{equation*}
    \Enc\left(\bigotimes_{i\in[u]}c_i\right) = \bigotimes_{i\in[u]}\Enc_i(c_i)
  \end{equation*}
  for $c_i\in C_i\;\forall i\in[u]$, and extended by linearity to all of $\bF_q^{\prod_{i\in[u]}k_i}=\bigotimes_{i\in[u]}\bF_q^{k_i}$.
\end{definition}

Tensor codes can be equivalently defined as follows. For $i\in[u]$ and $j\in\prod_{i'\in[u]\setminus\{i\}}[n_{i'}]$, we say the associated \emph{direction-$i$ column} is the set of all $j'\in\prod_{i'\in[u]}[n_{i'}]$ such that $j'_{i'}=j_{i'}$ for every $i'\in[u]\setminus\{i\}$. Therefore every direction-$i$ column has $n_i$ elements. Then the tensor product of $C_1,\dots,C_u$ is simply the set of all vectors $c\in\bF_q^{\prod_{i\in[u]}[n_i]}$ such that for every $i\in[u]$, the restriction of $c$ to every direction-$i$ column is a codeword in $C_i$. In other words, the tensor code is the kernel of the parity-check matrix defined below.


\begin{definition}
  \label{def:pcpc}
  For $u\in\bN$, for $i\in[u]$ let $H^i\in\bF_q^{m_i\times n_i}$ be a matrix. Define sets
  \begin{equation*}
    N = [n_1]\times\cdots\times[n_u]
  \end{equation*}
   and $S=\bigsqcup_{i\in[u]}S^i$ with each
  \begin{equation*}
    S^i = [n_1]\times\cdots\times[n_{i-1}]\times[m_i]\times[n_{i+1}]\times\cdots\times[n_u].
  \end{equation*}
  The \emph{tensor-code parity-check matrix} $H\in\bF_q^{S\times N}$ associated to $H^1,\dots,H^u$ is given for $x\in\bF_q^N$ by
  \begin{equation*}
    Hx = ((I^{\otimes i-1}\otimes H^i\otimes I^{\otimes u-i})x)_{i\in[u]} \in \bigoplus_{i\in[u]}\bF_q^{S^i} = \bF_q^S.
  \end{equation*}
  This matrix is by definition a parity-check matrix for the tensor product of the codes $\ker(H^i)$ for $i\in[u]$, that is,
  \begin{equation*}
    \ker(H) = \bigotimes_{i\in[u]}\ker(H^i).
  \end{equation*}
\end{definition}

We will use the fact shown by \cite{viderman_combination_2015} (which strengthened a result of \cite{ben-sasson_robust_2004}) that tensor product codes are \emph{locally testable}. Loosely speaking, this result says that for every $y\in\bF_q^{\prod_{i\in[u]}[n_i]}$, the number of direction-$i$ columns in which the restriction of $y$ is not a codeword of $C_i$ grows with the distance $\min_{c\in C}|y-c|$ from $y$ to $C$. In other words, defining $H$ as in \Cref{def:pcpc}, then $|Hy|$ grows with $\min_{c\in C}|y-c|$. We provide the formal statement in \Cref{lem:loctest} below.

\begin{lemma}[Local testability of tensor product codes \cite{viderman_combination_2015}]
  \label{lem:loctest}
  For $u,n,d\in\bN$, for $i\in[u]$ let $C_i=\ker(H_i)$ be an $[n,\; k_i,\; d_i\geq d]_q$ code with full-rank parity-check matrix $H_i\in\bF_q^{(n-k_i)\times n}$. Let $H\in\bF_q^{S\times[N]}$ be the associated tensor-code parity-check matrix, and let $C=\ker(H)$ be the tensor code. Then it holds for every $y\in\bF_q^N$ that $|Hy|\geq\rho(u,n,d)\cdot\min_{c\in C}|y-c|$, where
  \begin{equation*}
    \rho(u,n,d) = \frac{1}{(un\cdot(n/d)^u)^\eta}
  \end{equation*}
  for an absolute constant $\eta>1$.
\end{lemma}

The quantity $\rho(u,n,d)$ in \Cref{lem:loctest} is called the \emph{soundness} of the tensor code.


\begin{remark}
  \label{remark:viderman}
  We have stated \Cref{lem:loctest} slightly differently from \cite[Theorem 3.1]{viderman_combination_2015}, but our statement still follows directly from \cite{viderman_combination_2015}. Specifically:
  \begin{enumerate}
  \item \cite[Theorem 3.1]{viderman_combination_2015} is stated assuming all codes $C_i$ are equal, though as noted in \cite[Remark 4.5]{viderman_combination_2015}, the techniques extend flawlessly to different $C_i$.
  \item The tester in \cite[Theorem 3.1]{viderman_combination_2015} actually checks if any of the $n-k_i$ entires of $Hy$ in a given direction-$i$ column are nonzero, meaning that $|Hy|$ in \Cref{lem:loctest} would be replaced by the number of such direction-$i$ columns with a nonzero entry. Yet this quantity is always within a factor of $n$ of $|Hy|$, so we can absorb the factor of $n$ into the $n^\eta=\poly(n)$ term in $\mu(u,n,d)$.
  \item When $u$ is not a power of $3$, the local tester in \cite{viderman_combination_2015} samples different syndrome components with different probabilities, meaning that it may seem that we should add different weights for different components when computing the Hamming norm $|Hy|$ in \Cref{lem:loctest}. However, the result of \cite{viderman_combination_2015} naturally holds under all different permutations of the codes $C_1,\dots,C_u$, and all permutations of the $n$ columns and the $n-k_i$ rows in each $H_i$. Averaging over all of these permutations, we conclude that the result holds with uniform weights on the components of the syndrome $Hy$.
  \end{enumerate}
\end{remark}

\subsubsection{Reed-Solomon Codes}
\label{sec:reedsol}
Here we define Reed-Solomon codes.

\begin{definition}
  \label{def:RS}
  For a finite field $\bF_q$, we let $\bF_q[X]$ denote the space of univariate polynomials over $\bF_q$. For $k\in\bN$, we let $\bF_q[X]^{<k}\subseteq\bF_q[X]$ denote the subspace of polynomials of degree $<k$. For $E\subseteq\bF_q$, We let $\evl_E:\bF_q[X]\rightarrow\bF_q^E$ denote the evaluation map, that is, $\evl_E(f)=(f(x))_{x\in E}$.

  The \emph{Reed-Solomon code} $\RS(q,k,E)\subseteq\bF_q^E$ is the code $\RS(q,k,E)=\evl_E(\bF_q[X]^{<k})$. This code has parameters $[n=|E|,\; k,\; d=n-k+1]_q$.
\end{definition}

We will crucially leverage the \emph{multiplication property} of Reed-Solomon codes, stated below, which simply reflects the fact that the product of two degree $<k$ polynomials is a degree $<2k-1$ polynomial. Here recall from \Cref{sec:notation} that we use `$*$' to denote component-wise multiplication of vectors.

\begin{fact}
  \label{fact:RSmult}
  $\RS(q,k,E)*\RS(q,k,E)\subseteq\RS(q,2k-1,E)$.
\end{fact}

\subsection{Gates and Circuits}
\label{sec:circuits}
In this section, we describe the definitions of circuits and fault-tolerance that we use in this paper. As we work over general prime power alphabets $\bF_q$ instead of just binary alphabets, we use the term ``$q$-ary dit'' to refer to a symbol over such an alphabet, or simply ``dit'' if the alphabet size $q$ is clear from context.

\begin{definition}
  Fix a prime power $q$. A \emph{gate} acting on $q$-ary dits with $m_{\mathrm{in}}$ input dits and $m_{\mathrm{out}}$ output dits is a function $G:\bF_q^{m_{\mathrm{in}}}\rightarrow\bF_q^{m_{\mathrm{out}}}$.
\end{definition}

We now present the gate set that we will use.

\begin{definition}
  \label{def:gates}
  For a prime power $q$, we define the following $q$-ary gates:
  \begin{itemize}
  \item The \emph{initialization gate} $\gInit:\{0\}\rightarrow\bF_q$ given by
    \begin{equation*}
      \gInit(0) = 0.
    \end{equation*}
  \item The \emph{termination gate} $\gTerm:\bF_q\rightarrow\{0\}$ given by
    \begin{equation*}
      \gTerm(x) = 0.
    \end{equation*}
  \item For $a\in\bF_q$, the gate $\gX^a:\bF_q\rightarrow\bF_q$ given by
    \begin{equation*}
      \gX^a(x) = x+a.
    \end{equation*}
  \item For $a\in\bF_q$, the gate $\gCX^a:\bF_q^2\rightarrow\bF_q^2$ given by
    \begin{equation*}
      \gCX^a(x_1,x_2) = (x_1,ax_1+x_2).
    \end{equation*}
  \item For $a\in\bF_q$, the gate $\gCCX^a:\bF_q^3\rightarrow\bF_q^3$ given by
    \begin{equation*}
      \gCCX^a(x_1,x_2,x_3) = (x_1,x_2,ax_1x_2+x_3).
    \end{equation*}
  \end{itemize}
  For the gates above, we sometimes replace the subscript with `$*$' to denote the set of all possible values. For instance, $\{\gInit,\gCCX^*\}=\{\gInit\}\cup\{\gCCX^a:a\in\bF_q\}$.
\end{definition}

Note that the identity gate is given by $\gX^0$. Also, the gate $\gCX^a$ can be implemented using the gates $\gInit,\gX^1,\gCCX^a$. We nevertheless include $\gCX^a$ in our gate set for convenience.

We next define our notion of a circuit. Below, recall that a \emph{directed acyclic graph} (dag) is a directed graph with no cycles. Also recall from \Cref{sec:notation} that we use `$\sqcup$' to denote direct products of functions.

\begin{definition}
  \label{def:circuit}
  A \emph{circuit} $\cR=(R_1,\dots,R_T)$ acting on a set $N$ of $q$-ary dits is given by a sequence $R_1,\dots,R_T$ of functions $R_t:\bF_q^{N_{t-1}}\rightarrow\bF_q^{N_t}$ for some subsets $N_0,\dots,N_T\subseteq N$. We call $N_t$ the \emph{active dits at time $t$}, and we call $N_0$ and $N_T$ the sets of \emph{input} and \emph{output dits}, respectively. We let $\cR(\cdot):\bF_q^{N_0}\rightarrow\bF_q^{N_T}$ denote the function
  \begin{equation*}
    \cR(x) = R_T\circ\cdots\circ R_1(x).
  \end{equation*}
  We define the \emph{transcript} $\tran(\cR;x)\in\bF_q^{N_0\sqcup\cdots\sqcup N_T}$ of $\cR$ on input $x$ by
  \begin{equation*}
    \tran(\cR;x) = (x,\; R_1(x),\; R_2\circ R_1(x),\dots,R_T\circ\cdots\circ R_1(x)).
  \end{equation*}
  
  We say $\cR$ uses time $T$ and space $N$ (or sometimes, space $|N|$). Note that each $q$-ary dit implicitly requires $\log_2q$ bits to represent; $q$ will be clear from context when discussing space usage.
  
  We say $\cR$ uses a gate set $\cG$ if each $R_t$ can be decomposed as a direct product of gates $G_{t,i}\in\cG$ acting on disjoint sets of $q$-ary dits. More formally, $\cR$ uses gate set $\cG$ if for each $t\in[T]$, there exists $m_t\in\bN$ and gates $G_{t,1},\dots,G_{t,m_t}\in\cG$ such that $R_t$ equals the direct product $\bigsqcup_{i\in[m_t]}G_{t,i}$ for some assignment of the dits in $N_{t-1}$ (resp.~$N_t$) to the set of all input (resp.~output) dits of $G_{t,i}$. Furthermore, letting $m_{\mathrm{in}}$ and $m_{\mathrm{out}}$ denote the number of input and output dits of $G_{t,i}$, respectively, then we require that the dits in $N_{t-1}\subseteq N$ assigned to the first $\min\{m_{\mathrm{in}},m_{\mathrm{out}}\}$ inputs of $G_{t,i}$ equal the dits in $N_t\subseteq N$ assigned to the first $\min\{m_{\mathrm{in}},m_{\mathrm{out}}\}$ outputs of $G_{t,i}$.

  The \emph{size} of a circuit $\cR$ using gate set $\cG$ is the number of input dits $|N_0|$ plus the total number of non-identity gates $G_{t,i}$ in the circuit, across all $t\in[T]$ and $i\in[m_t]$.
\end{definition}

We will typically consider circuits using gate set $\cG=\{\gInit,\gTerm,\gX^*,\gCX^*,\gCCX^*\}$. This gate set is universal because it can perform addition and multiplication, and every function over $\bF_q$ can be expressed as a sufficiently high-degree polynomial.

\Cref{def:circuit} captures the standard notion of a circuit with distinct timesteps, such that every dit is acted upon by some gate in every timestep. Dits can also be created and terminated, if the sets $N_t$ are not all the same.

Perhaps the least standard part of \Cref{def:circuit} is the requirement that for each $G_{t,i}$, the first $\min\{m_{\mathrm{in}},m_{\mathrm{out}}\}$ input and output dits are equal. Intuitively, this condition says that we do not allow ``free'' permutations of the dits between timesteps; gates must output the same dits in $N$ they were fed as inputs, up to creation/termination of dits. Because a pair of dits labeled by $1,2$ can be swapped using an ancilla dit labeled $3$ with the circuit
\begin{align*}
  (a,b)
  &\xrightarrow{\gInit_3} (a,b,0) \xrightarrow{\gCX_{1,3}^{+1}} (a,b,a) \xrightarrow{\gCX_{3,1}^{-1}} (0,b,a) \xrightarrow{\gCX_{2,1}^{+1}} (b,b,a) \\
  &\xrightarrow{\gCX_{1,2}^{-1}} (b,0,a) \xrightarrow{\gCX_{3,2}^{+1}} (b,a,a) \xrightarrow{\gCX_{2,3}^{-1}} (b,a,0) \xrightarrow{\gTerm_3} (b,a),
\end{align*}
circuits that allow $\{\gInit,\gCX^*,\gTerm\}$ gates can perform arbitrary permutations with constant space and time overhead. Therefore our prohibition of ``free'' permutations of dits only affects the time and space usage of the circuits we consider by a constant factor. We nevertheless choose to prohibit free permutations for consistency with a standard model of circuits common in the quantum literature, in which gates act on a fixed set $N$ of (qu)dits, and each gate has the same input and output (qu)dits.

The following basic method for applying linear transformations using $\gCX^*$ gate will be useful.

\begin{definition}
  \label{def:matcirc}
  For a matrix $A\in\bF_q^{m\times n}$, we define an \emph{associated circuit $\cR_A$} as follows. $\cR_A$ acts on dits $N=[m]\sqcup[n]$, uses time $T=mn$, and uses gate set $\{\gCX^*\}$. The input and output dits are all of $N$. We define $\cR_A$ to sequentially loop through every $(i,j)\in[m]\times[n]$, and for each such pair apply $\gCX^{A_{i,j}}$ with control dit $j\in[n]$ and target dit $i\in[m]$. Therefore on input $(x,y)\in\bF_q^{[n]\sqcup[m]}$, then $\cR_A$ outputs $\cR_A(x,y)=(x,y+Ax)$.
\end{definition}

\subsection{Fault-Tolerance and Fault-Detection}
In this section, we define the notions of fault-tolerance and fault-detection we will use in this paper. The definitions here are inspired by those for quantum circuits in \cite{nguyen_quantum_2025,he_composable_2025}, though of course we only consider classical circuits.

We begin by defining faults, which capture the errors against which we will want to protect.

\begin{definition}
  Let $\cR=(R_1,\dots,R_T)$ be a circuit with each $R_t:\bF_q^{N_{t-1}}\rightarrow\bF_q^{N_t}$. A \emph{fault} on the circuit $\cR$ is a sequence $\cF=(F_1,\dots,F_T)$ of functions $F_t:\bF_q^{N_t}\rightarrow\bF_q^{N_t}$. Letting $\supp(F_t)\subseteq N_t$ denote the set of dits in $N_t$ on which $F_t$ does not act as the identity, we define $\supp(\cF)\subseteq N^{\sqcup T}$ by $\supp(\cF)=\supp(F_1)\sqcup\cdots\sqcup\supp(F_T)$.

  The \emph{$\cF$-corrupted circuit $\cR[\cF]$} is defined by
  \begin{equation*}
    \cR[\cF] = (R_1,F_1,R_2,F_2,\dots,R_T,F_T).
  \end{equation*}
  In a slight abuse of notation, we define the \emph{transcript} $\tran(\cR[\cF];x)\in\bF_q^{N_0\sqcup\cdots\sqcup N_T}$ of an $\cF$-corrupted circuit $\cR[\cF]$ on input $x$ to exclude every other timestep corresponding to a fault layer, so that
  \begin{equation*}
    \tran(\cR[\cF];x) = (x,\; F_1\circ R_1(x),\; F_2\circ R_2\circ F_1\circ R_1(x),\dots,F_T\circ R_T\circ\cdots\circ F_1\circ R_1(x)).
  \end{equation*}
\end{definition}

We next define the notion of a bad set, which is a set of dits that, if entirely corrupted, will overwhelm our fault-tolerance scheme.

\begin{definition}
  For a set $N$ and a family of subsets $\cE\subseteq 2^N$ called \emph{bad sets}, we say a set $S\subseteq N$ is \emph{$\cE$-avoiding} if for every $E\in\cE$ we have $E\not\subseteq S$.

  We extend this definition to faults in the natural way. That is, given a family of bad sets $\cE\subseteq 2^{N^{\sqcup T}}=2^{N\times T}$ for a circuit using space $N$ and time $T$, a fault $\cF$ is \emph{$\cE$-avoiding} if $\supp(\cF)$ is $\cE$-avoiding.
\end{definition}

We will often construct families of bad sets for a circuit using space $N$ and time $T$ by specifying a family $\cE\subseteq N$ of bad sets for a single timestep, and then letting $\cE^{\sqcup T}\subseteq N^{\sqcup T}$ be our family of bad sets for the circuit. We in turn typically choose $\cE$ to contain all sets of size above some given threshold, as defined below.

\begin{definition}
  For a set $N$, a family $\cE\subseteq 2^N$, and a real number $\lambda$, we define
  \begin{equation*}
    \cE|_{\geq\lambda} = \{E\in\cE:|E|\geq\lambda\}.
  \end{equation*}
  In particular, $2^N|_{\geq\lambda}$ denotes the set of all subsets of $N$ of size $\geq\lambda$.
\end{definition}

It will also be helpful to associate families of bad sets to codes, to represent errors on codewords that we cannot detect and/or correct.

\begin{definition}
  A \emph{decorated code} $D=(C,\Enc,\cE)$ consists of a linear code $C$ with encoding map $\Enc$ and a family $\cE\subseteq 2^{[n]}$ of bad sets, where $n$ denotes the length of $C$. We let $\emptyset$ denote the trivial decorated code with $n=0$.

  For a set $S\subseteq C$ of codewords in $C$ (where $S$ may not be a linear subspace), we say $y\in\bF_q^n$ is a \emph{$\cE$-deviation of $S$} if there exists $c\in S$ for which $\supp(y-c)$ is $\cE$-avoiding.
\end{definition}

We are now ready to define our notion of a fault-tolerant gadget.

\begin{definition}
  \label{def:faulttol}
  Let $\cR=(R_1,\dots,R_T)$ be a circuit using space $N$ and time $T$, with input and output dits $N_{\mathrm{in}}$ and $N_{\mathrm{out}}$ respectively. Let $\cE_{\mathrm{run}}\subseteq 2^{N^{\sqcup T}}$ be a family of bad sets. For $\alpha\in\{\mathrm{in},\mathrm{out}\}$, let $D_\alpha=(C_\alpha,\Enc_\alpha,\cE_\alpha)$ be a $[n_\alpha,k_\alpha]_q$ decorated code with codeword dits labeled by $N_\alpha$, so that $n_\alpha=|N_\alpha|$ and $\cE_\alpha\subseteq 2^{N_\alpha}$.

  We say the data
  \begin{equation*}
    (\cR,\cE_{\mathrm{run}},D_{\mathrm{in}},D_{\mathrm{out}})
  \end{equation*}
  forms a \emph{fault-tolerant gadget} for a set $\bar{\cO}$ of functions $\bar{O}:\bF_q^{k_{\mathrm{in}}}\rightarrow\bF_q^{k_{\mathrm{out}}}$ if the following hold:
  \begin{enumerate}
  \item\label{it:ftnoerr} For every $x\in\bF_q^{k_{\mathrm{in}}}$, it holds that $\cR\circ\Enc_{\mathrm{in}}(x)\in\Enc_{\mathrm{out}}\circ\bar{\cO}(x)$ (i.e.~$\cR\circ\Enc_{\mathrm{in}}(x)=\Enc_{\mathrm{out}}\circ\bar{O}(x)$ for some $\bar{O}\in\bar{\cO}$).
  \item\label{it:fterr} For every $x\in\bF_q^{k_{\mathrm{in}}}$, every $\cE_{\mathrm{in}}$-deviation $y$ of $\Enc_{\mathrm{in}}(x)$, and every $\cE_{\mathrm{run}}$-avoiding fault $\cF$ for $\cR$, then the output $\cR[\cF](y)$ is a $\cE_{\mathrm{out}}$-deviation of $\Enc_{\mathrm{out}}\circ\bar{\cO}(x)$.
  \end{enumerate}
\end{definition}

In \Cref{def:faulttol}, \Cref{it:ftnoerr} requires that a noiseless execution of $\cR$ outputs a noisless encoding of the input message, while \Cref{it:fterr} requires that a noisy execution of $\cR$ outputs a noisy encoding of the input message. The first condition (\Cref{it:ftnoerr}) is not necessary for fault-tolerance in general, but will be useful to construct PCPs with perfect completeness.

Depending on the choice of $\cE_{\mathrm{in}},\cE_{\mathrm{run}},\cE_{\mathrm{out}}$, we typically design a fault-tolerant gadget to either limit the propagation of errors, actively correct errors, or both.

For our application to PCPs, the notion of fault-tolerance in \Cref{def:faulttol} is stronger than necessary. Specifically, in the PCP setting, it is ok to fail to output (a deviation of) the true output, if we can detect the failure. That is, it is sufficient to detect errors, which can be easier than correcting them. We therefore define the following notion of a fault-detecting gadget.

\begin{definition}
  \label{def:faultdet}
  Let $\cR=(R_1,\dots,R_T)$ be a circuit using space $N$ and time $T$, with active dits $N_{\mathrm{in}}=N_0,N_1\dots,N_T=N_{\mathrm{out}}$. Let $\cE_{\mathrm{run}},\cE_{\mathrm{det}}\subseteq 2^{N^{\sqcup T}}$ be families of sets. For $\alpha\in\{\mathrm{in},\mathrm{out}\}$, let $D_\alpha=(C_\alpha,\Enc_\alpha,\cE_\alpha)$ be a $[n_\alpha,k_\alpha]_q$ decorated code with codeword dits labeled by $N_\alpha$, so that $n_\alpha=|N_\alpha|$ and $\cE_\alpha\subseteq 2^{N_\alpha}$.

  We say the data
  \begin{equation*}
    (\cR,\cE_{\mathrm{run}},D_{\mathrm{in}},D_{\mathrm{out}},\cE_{\mathrm{det}})
  \end{equation*}
  forms a \emph{fault-detecting gadget} for a set $\bar{\cO}$ of functions $\bar{O}:\bF_q^{k_{\mathrm{in}}}\rightarrow\bF_q^{k_{\mathrm{out}}}$ if the following hold:
  \begin{enumerate}
  \item\label{it:fdnoerr} For every $x\in\bF_q^{k_{\mathrm{in}}}$ and every $E\in\cE_{\mathrm{det}}$, it holds that $\cR\circ\Enc_{\mathrm{in}}(x)\in\Enc_{\mathrm{out}}\circ\bar{\cO}(x)$ and $E\cap\supp(\tran(\cR;\Enc_{\mathrm{in}}(x)))=\emptyset$.
  \item\label{it:fderr} For every $x\in\bF_q^{k_{\mathrm{in}}}$, every $\cE_{\mathrm{in}}$-deviation $y$ of $\Enc_{\mathrm{in}}(x)$, and every $\cE_{\mathrm{run}}$-avoiding fault $\cF$ for $\cR$, then at least one of the following holds:
    \begin{enumerate}
    \item\label{it:fdecorr} The output $\cR[\cF](y)$ is a $\cE_{\mathrm{out}}$-deviation of $\Enc_{\mathrm{out}}\circ\bar{\cO}(x)$, or
    \item\label{it:fdedet} There exists a set $E\in\cE_{\mathrm{det}}$ such that $E\subseteq\supp(\tran(\cR[\cF];y))$, where here we view $E\subseteq N^{\sqcup T}\cong N\times\{1,\dots,T\}$ as a subset of $N^{\sqcup T+1}\cong N\times\{0,1,\dots,T\}$.
    \end{enumerate}
  \end{enumerate}

  We say the gadget is furthermore \emph{mending} if for every $y\in\bF_q^{N_{\mathrm{in}}}$ and every $\cE_{\mathrm{run}}$-avoiding fault $\cF$ for $\cR$, then at least one of the following holds:
  \begin{enumerate}[label=(\alph*),start=3]
  \item\label{it:fdecorrmend} There exists $x'\in\bF_q^{k_{\mathrm{in}}}$ such that the output $\cR[\cF](y)$ is a $\cE_{\mathrm{out}}$-deviation of $\Enc_{\mathrm{out}}\circ\bar{\cO}(x')$, or
  \item\label{it:fdedetmend} There exists a set $E\in\cE_{\mathrm{det}}$ such that $E\subseteq\supp(\tran(\cR[\cF];y))$.
  \end{enumerate}
\end{definition}

The union $\bigcup\cE_{\mathrm{det}}$ of all sets in the family $\cE_{\mathrm{det}}\subseteq 2^{N^{\sqcup T}}$ in \Cref{def:faultdet} provides a set of dits in the space-time execution (or equivalently, in the transcript) of a circuit $\cR$ that detect errors. Specifically, \Cref{it:fdnoerr} in \Cref{def:faultdet} says that in the absence of errors, all of these dits must be $0$. Meanwhile, \Cref{it:fderr} in \Cref{def:faultdet} says that in the presence of a $\cE_{\mathrm{in}}$-avoiding input error and a $\cE_{\mathrm{run}}$-avoiding fault, then either the output error is $\cE_{\mathrm{out}}$-avoiding, or else the set of nonzero dits in $\bigcup\cE_{\mathrm{det}}$ is sufficiently large to contain some $E\in\cE_{\mathrm{det}}$.

We will often fix some set $E_{\mathrm{det}}\subseteq N^{\sqcup T}$ and of ``detector dits'' and then choose $\cE_{\mathrm{det}}=2^{E_{\mathrm{det}}}|_{\geq\lambda_{\mathrm{det}}}$ to be the family of subsets of $E_{\mathrm{det}}$ of size at least some threshold $\lambda_{\mathrm{det}}$. Then the criterion for detecting an error (given in \Cref{it:fdedet} in \Cref{def:faultdet} above) is simply that $\geq\lambda_{\mathrm{det}}$ of the detector dits in $E_{\mathrm{det}}$ are nonzero.

By definition, a fault-tolerant gadget is equivalent to a fault-detecting gadget with $\cE_{\mathrm{det}}=\emptyset$:

\begin{fact}
  \label{fact:dettol}
  The data $(\cR,\cE_{\mathrm{run}},D_{\mathrm{in}},D_{\mathrm{out}})$ forms a fault-tolerant gadget for $\bar{\cO}$ if and only if the data $(\cR,\cE_{\mathrm{run}},D_{\mathrm{in}},D_{\mathrm{out}},\emptyset)$ forms a fault-detecting gadget for $\bar{\cO}$.
\end{fact}

We nevertheless retain the language ``fault-tolerant'' to describe this special case for consistency with the prior literature.

The mending property in \Cref{def:faultdet} requires that even when the input error is arbitrarily large, the gadget either outputs a $\cE_{\mathrm{out}}$-deviation of an encoding of \emph{some} output message in $\bar{\cO}(x')$ (for an arbitrary input $x'$), or else detects an error.

An analogue of the mending property in the quantum setting is used in \cite{he_composable_2025,nguyen_quantum_2025} (where it is called the \emph{friendly} property) to allow gadgets to be composed in a simulative manner. That is, the error rate of the (qu)dits in a fault-tolerance scheme can be reduced by recursively simulating each (qu)dit with another fault-tolerance scheme. If one of these recursive simulations experiences a high-weight error, a mending gadget can return the corrupted state to a valid code state, so that the recursive simulation of that (qu)dit can resume.

While we will not need to perform such simulative composition, we will use the mending property for a different purpose in the PCP setting. Specifically, we assume a PCP proof is provided in encoded form, to which we apply a fault-detecting gadget. If the provided encoding has a high-weight error, the mending property ensures that either the error is detected, or else the gadget's output is close to the true output for \emph{some} valid input $x'$. In contrast, without the mending property, we would have no control over our gadget's output under high-weight input errors.



Below, we show that in order to prove (mending) fault-detection, without loss of generality we may restrict attention to \emph{additive} faults, which simply add some fixed error into the circuit's dits at every timestep. This basic result holds because we allow the fault to be chosen after the circuit's input, so the outputs of an arbitrary fault can be simulated by an appropriately chosen additive fault.

\begin{definition}
  For a circuit $\cR$, a fault $\cF=(F_1,\dots,F_T)$ is \emph{additive} if every $F_t:\bF_q^{N_t}\rightarrow\bF_q^{N_t}$ is of the form $F_t(y)=y+f_t$ for some $f_t\in\bF_q^{N_t}$.
\end{definition}

Note that an additive fault $\cF$ acts independently on each dit, i.e.~each $F_t$ is a direct product of the single-dit functions $y_i\mapsto y_i+(f_t)_i$ for $i\in N_t$.

\begin{lemma}
  \label{lem:fdadd}
  Let $(\cR,\cE_{\mathrm{run}},D_{\mathrm{in}},D_{\mathrm{out}},\cE_{\mathrm{det}})$ and $\bar{\cO}$ be some data that satisfy the conditions of \Cref{def:faultdet} upon restricting attention to only \emph{additive} $\cE_{\mathrm{run}}$-avoiding faults $\cF$. Then in fact this data satisfies \Cref{def:faultdet} for general $\cE_{\mathrm{run}}$-avoiding faults $\cF$, that is, this data forms a (mending) fault-detecting gadget.
\end{lemma}
\begin{proof}
  We define the variables $C_\alpha,\Enc_\alpha,\cE_\alpha,n_\alpha,k_\alpha,N_\alpha$ as in \Cref{def:faultdet}. We first prove the claim for non-mending fault-detecting gadgets. Because \Cref{it:fdnoerr} in \Cref{def:faultdet} does not depend on a choice of fault, it holds immediately by assumption from the lemma statement. We next show that \Cref{it:fderr} holds for general faults $\cF$, assuming it holds for additive faults $\cF'$.
  
  Let $x\in\bF_q^{k_{\mathrm{in}}}$, let $y$ be a $\cE_{\mathrm{in}}$-deviation of $\Enc_{\mathrm{in}}$, and let $\cF=(F_1,\dots,F_t)$ be a general $\cE_{\mathrm{run}}$-avoiding fault for $\cR$. We then construct an additive fault $\cF'=(F_1',\dots,F_T')$ for $\cR$ by letting $F_t'(z)=z+f_t'$ for
  \begin{equation*}
    f_t' = \tran(\cR[\cF];y)_t-R_t\circ F_{t-1}\circ R_{t-1}\cdots\circ F_1\circ R_1(y),
  \end{equation*}
  where $\tran(\cR[\cF];y)_t\in\bF_q^{N_t}$ denotes the restriction of the transcript to the $t$th timestep. By definition, every dit $j\in N_t\setminus\supp(F_t)$ must have $(f_t')_j=0$, so $\supp(\cF')\subseteq\supp(\cF)$, and hence $\cF'$ is $\cE_{\mathrm{run}}$-avoiding. This construction also ensures that $\tran(\cR[\cF];y)=\tran(\cR[\cF'];y)$. Thus because \Cref{it:fderr} in \Cref{def:faultdet} holds for $\cF'$, \Cref{it:fderr} must also hold for $\cF$, as \Cref{it:fderr} by definition only depends on values of $x$, $y$, and the transcript $\tran(\cR[\cF];y)=\tran(\cR[\cF'];y)$ (because $\cR[\cF](y)=\tran(\cR[\cF];y)_T$). Thus $(\cR,\cE_{\mathrm{run}},D_{\mathrm{in}},D_{\mathrm{out}},\cE_{\mathrm{det}})$ is a fault-detecting gadget for $\bar{\cO}$, as desired.

  The proof for mending gadgets is analogous as above; we omit the details to avoid redundancy.
\end{proof}

While \Cref{lem:fdadd} shows that we can always restrict attention to additive faults in this paper, one may need to consider more general faults in other settings, such as when the fault-detecting circuit is allowed to be randomized.

\subsection{Gadget Composition}
This section describes how fault-detection (and as a special case, fault-tolerance) is preserved under sequential and parallel composition of gadgets, defined below. Similar results are shown in the quantum setting in \cite{nguyen_quantum_2025,he_composable_2025}; we simply adapt these results to our setting.

\subsubsection{Sequential Composition}
Here we consider sequential composition, in which two gadgets are run one after another on a shared set of dits.

\begin{lemma}
  \label{lem:seqcomp}
  For a set $N$, for $i\in[2]$ let
  \begin{equation*}
    (\cR^i=(R^i_1,\dots,R^i_{T_i}),\; \cE^i_{\mathrm{run}},\; D^i_{\mathrm{in}},\; D^i_{\mathrm{out}},\; \cE^i_{\mathrm{det}})
  \end{equation*}
  be a fault-detecting gadget for $\bar{\cO}^i$ using space $N_i\subseteq N$ and time $T_i$. Assume that the set of output dits of $\cR^1$ equals the set of input dits of $\cR^2$, and that $D^1_{\mathrm{out}}=D^2_{\mathrm{in}}$.

  Define the \emph{sequential composition} of $\cR^1,\cR^2$ to be the circuit acting on dits $N$ given by
  \begin{equation*}
    \cR^2\circ\cR^1 = (R^1_1,\dots,R^1_{T_1}, R^2_1,\dots,R^2_{T_2}).
  \end{equation*}
  Then
  \begin{equation}
    \label{eq:scgad}
    \left(\cR=\cR^2\circ\cR^1,\; \cE_{\mathrm{run}}=\cE^1_{\mathrm{run}}\sqcup\cE^2_{\mathrm{run}},\; D^1_{\mathrm{in}},\; D^2_{\mathrm{out}},\; \cE_{\mathrm{det}}=\cE^1_{\mathrm{det}}\sqcup\cE^2_{\mathrm{det}}\right)
  \end{equation}
  is a fault-detecting gadget for $\bar{\cO}=\bar{\cO}^2\circ\bar{\cO}^1$ using space $N$ (so $|N|\geq\max\{|N_1|,|N_2|\}$) and time $T=T_1+T_2$.

  Furthermore, if $(\cR^1,\cE^1_{\mathrm{run}},D^1_{\mathrm{in}},D^1_{\mathrm{out}},\cE^1_{\mathrm{det}})$ is mending, then the sequential composition gadget in \Cref{eq:scgad} is mending.
\end{lemma}
\begin{proof}
  We begin with the claim regarding non-mending fault-detection. By definition $\cR$ uses space $N$ and time $T=T_1+T_2$, so it suffices to show that the sequential composition gadget in \Cref{eq:scgad} is fault-detecting for $\bar{\cO}$. For this purpose, let each $D^i_\alpha=(C^i_\alpha,\Enc^i_\alpha,\cE^i_\alpha)$ be a $[n^i_\alpha,k^i_\alpha]_q$ code.

  The condition in \Cref{it:fdnoerr} in \Cref{def:faultdet} for the sequential composition follows directly by applying the same condition for $(\cR^i,\cE^i_{\mathrm{run}},D^i_{\mathrm{in}},D^i_{\mathrm{out}},\cE^i_{\mathrm{det}})$ along with the fact that $D^1_{\mathrm{out}}=D^2_{\mathrm{in}}$, that is, for every $x\in\bF_q^{k^1_{\mathrm{in}}}$,
  \begin{equation*}
    \cR^2\circ\cR^1\circ\Enc^1_{\mathrm{in}}(x) \in \cR^2\circ\Enc^1_{\mathrm{out}}\circ\bar{\cO}^1(x) \subseteq \Enc^2_{\mathrm{out}}\circ\bar{\cO}^2\circ\bar{\cO}^1(x).
  \end{equation*}
  
  To show the condition in \Cref{it:fderr} in \Cref{def:faultdet} for the sequential composition, let $x\in\bF_q^{k^1_{\mathrm{in}}}$, let $y$ be a $\cE^1_{\mathrm{in}}$-deviation of $\Enc^1_{\mathrm{in}}(x)$, and let $\cF=(\cF^1,\cF^2)$ be a $\cE_{\mathrm{run}}$-avoiding fault for $\cR$, so that each $\cF^i$ is a $\cE^i_{\mathrm{run}}$-avoiding fault for $\cR^i$. Applying the fault-detection of $(\cR^1,\cE^1_{\mathrm{run}},D^1_{\mathrm{in}},D^1_{\mathrm{out}},\cE^1_{\mathrm{det}})$, we have that either $\cR^1[\cF^1](y)$ is a $\cE^1_{\mathrm{out}}$-deviation of $\Enc^1_{\mathrm{out}}\circ\bar{\cO}^1(x)$, or else there exists some $E\in\cE^1_{\mathrm{det}}$ such that $E\subseteq\supp(\tran(\cR^1[\cF^1];y))$. In the latter case, because $\cE^1_{\mathrm{det}}\subseteq\cE_{\mathrm{det}}$, then \Cref{it:fdedet} in \Cref{def:faultdet} holds for the sequential composition gadget, so we are done. In the former case, we can then apply the fault-detection of $(\cR^2,\cE^2_{\mathrm{run}},D^2_{\mathrm{in}},D^2_{\mathrm{out}},\cE^2_{\mathrm{det}})$ to conclude that either $\cR^2[\cF^2]\circ\cR^1[\cF^1](y)=\cR[\cF](y)$ is a $\cE^2_{\mathrm{out}}$-deviation of $\Enc^2_{\mathrm{out}}\circ\bar{\cO}(x)$, or else there exists some $E\in\cE^2_{\mathrm{det}}$ such that $E\subseteq\supp(\tran(\cR^2[\cF^2];\cR^1[\cF^1](y)))$. Again in the latter case, we are done because $\cE_{\mathrm{det}}^2\subseteq\cE_{\mathrm{det}}$. Now in the former case, then \Cref{it:fdecorr} in \Cref{def:faultdet} holds for the sequential composition gadget, so we are also done.

  The proof of the claim regarding mending fault-detection is similar to the non-mending case above. The difference is that now the input $y$ may be an arbitrary element of $\bF_q^{n^1_{\mathrm{in}}}$, and then either $\cR^1[\cF^1](y)$ is a $\cE^1_{\mathrm{out}}$-deviation of $\Enc^1_{\mathrm{in}}\circ\bar{\cO}^1(x')$ for some $x'\in\bF_q^{k^1_{\mathrm{in}}}$, or else there exists some $E\in\cE^1_{\mathrm{det}}$ such that $E\subseteq\supp(\tran(\cR^1[\cF^1];y))$. Then applying the fault-detection of $(\cR^2,\cE^2_{\mathrm{run}},D^2_{\mathrm{in}},D^2_{\mathrm{out}},\cE^2_{\mathrm{det}})$ as above, we conclude that the sequential composition gadget is mending.
\end{proof}

\subsubsection{Parallel Composition}
Here we consider parallel composition, in which two gadgets are run at the same time on disjoint sets of dits. We will use the following notation for (decorated) codes consisting of the disjoint union of smaller (decorated) codes.

\begin{definition}
  For $i\in[2]$, let codes $D^i=(C^i,\Enc^i,\cE^i)$ be a $[n_i,k_i,d_i]$ decorated code. We define the \emph{disjoint union} of $C^1,C^2$ to be the $[n_1+n_2,\; k_1+k_2,\; \min\{d_1,d_2\}]_q$ code
  \begin{equation*}
    C^1\sqcup C^2 = \{(c_1,c_2):c_1\in C_1,\;c_2\in C_2\},
  \end{equation*} and then we define the \emph{disjoint union} of $D^1,D^2$ by
  \begin{equation*}
    D^1\sqcup D^2=(C^1\sqcup C^2,\; \Enc^1\sqcup\Enc^2,\; \cE^1\sqcup\cE^2).
  \end{equation*}
\end{definition}

\begin{lemma}
  \label{lem:parcomp}
  For $T\in\bN$ and for $i\in[2]$ let
  \begin{equation*}
    (\cR^i=(R^i_1,\dots,R^i_{T_i}),\; \cE^i_{\mathrm{run}},\; D^i_{\mathrm{in}},\; D^i_{\mathrm{out}},\; \cE^i_{\mathrm{det}})
  \end{equation*}
  be a fault-detecting gadget for $\bar{\cO}^i$ using space $N_i$ and time $T$.

  Define the \emph{parallel composition} of $\cR^1,\cR^2$ to be the circuit acting on dits $N=N_1\sqcup N_2$ given by
  \begin{equation*}
    \cR^1\sqcup\cR^2 = (R^1_1\sqcup R^2_1,\dots,R^1_T\sqcup R^2_T).
  \end{equation*}
  Then
  \begin{equation}
    \label{eq:pcgad}
    \left(\cR=\cR^2\sqcup\cR^1,\; \cE_{\mathrm{run}}=\cE^1_{\mathrm{run}}\sqcup\cE^2_{\mathrm{run}},\; D_{\mathrm{in}}=D^1_{\mathrm{in}}\sqcup D^2_{\mathrm{in}},\; D_{\mathrm{out}}=D^1_{\mathrm{out}}\sqcup D^2_{\mathrm{out}},\; \cE_{\mathrm{det}}=\cE^1_{\mathrm{det}}\sqcup\cE^2_{\mathrm{det}}\right)
  \end{equation}
  is a fault-detecting gadget for $\bar{\cO}=\bar{\cO}^1\sqcup\bar{\cO}^1$ using space $N$ (so $|N|=|N_1|+|N_2|$) and time $T$.

  Furthermore, if for every $i\in[2]$ the gadget $(\cR^i,\cE^i_{\mathrm{run}},D^i_{\mathrm{in}},D^i_{\mathrm{out}},\cE^i_{\mathrm{det}})$ is mending, then the parallel composition in \Cref{eq:pcgad} is mending.
\end{lemma}
\begin{proof}
  We begin with the claim regarding non-mending fault-detection. By definition $\cR$ uses space $N=N_1\sqcup N_2$ and time $T$, so it suffices to show that the parallel composition gadget in \Cref{eq:pcgad} is fault-detecting for $\bar{\cO}$. For this purpose, let each $D^i_\alpha$ be a $[n^i_\alpha,k^i_\alpha]_q$ code, so that $D_\alpha$ is a $[n_\alpha=n^1_\alpha+n^2_\alpha,\; k_\alpha=k^1_\alpha+k^2_\alpha]_q$ code.

  The condition in \Cref{it:fdnoerr} in \Cref{def:faultdet} for the parallel composition follows directly by applying the same condition for $(\cR^i,\cE^i_{\mathrm{run}},D^i_{\mathrm{in}},D^i_{\mathrm{out}},\cE^i_{\mathrm{det}})$, that is, for every $x=(x_1,x_2)\in\bF_q^{k_{\mathrm{in}}}=\bF_q^{k^1_{\mathrm{in}}+k^2_{\mathrm{in}}}$,
  \begin{equation*}
    (\cR^1\sqcup\cR^2)\circ(\Enc^1_{\mathrm{in}}\sqcup\Enc^2_{\mathrm{in}})(x_1,x_2) \in (\Enc^1_{\mathrm{out}}\sqcup\Enc^2_{\mathrm{out}})\circ(\bar{\cO}^1\sqcup\bar{\cO}^2)(x_1,x_2).
  \end{equation*}

  To show the condition in \Cref{it:fderr} in \Cref{def:faultdet} for the parallel composition, let $x=(x_1,x_2)\in\bF_q^{k_{\mathrm{in}}}=\bF_q^{k^1_{\mathrm{in}}+k^2_{\mathrm{in}}}$, let $y=(y_1,y_2)$ be a $\cE^1_{\mathrm{in}}\sqcup\cE^2_{\mathrm{in}}$-deviation of of $\Enc_{\mathrm{in}}(x)$, and let $\cF=(F_1,\dots,F_T)$ be a $\cE_{\mathrm{run}}$-avoiding additive fault for $\cR$. As $\cF$ is additive, we can write $\cF=\cF^1\sqcup\cF^2$, meaning that each $F_t=F^1_t\sqcup F^t_t$, where each $\cF^i$ is some $\cE^i_{\mathrm{run}}$-avoiding fault for $\cR^i$. Applying the fault-detection of $(\cR^i,\cE^i_{\mathrm{run}},D^i_{\mathrm{in}},D^i_{\mathrm{out}},\cE^i_{\mathrm{det}})$, we have that either $\cR[\cF](y)=(\cR^1[\cF^1](y_1),\cR^2[\cF^2](y_2))$ is a $\cE_{\mathrm{out}}$-deviation of $\Enc_{\mathrm{out}}\circ\bar{\cO}(x)$, or there exists some $E\in\cE^1_{\mathrm{det}}$ such that $E\subseteq\supp(\tran(\cR^1[\cF^1];y_1))$, or there exists some $E\in\cE^2_{\mathrm{det}}$ such that $E\subseteq\supp(\tran(\cR^2[\cF^2];y_2))$. The latter two cases together are equivalent to the statement that there exists some $E\in\cE_{\mathrm{det}}=\cE^1_{\mathrm{det}}\sqcup\cE^2_{\mathrm{det}}$ such that $E\subseteq\supp(\tran(\cR[\cF];y))$, so we conclude that either \Cref{it:fdecorr} or \Cref{it:fdedet} in \Cref{def:faultdet} holds, as desired.

  Now assume that for both $i\in[2]$ the gadget $(\cR^i,\cE^i_{\mathrm{run}},D^i_{\mathrm{in}},D^i_{\mathrm{out}},\cE^i_{\mathrm{det}})$ is mending. Let $y=(y_1,y_2)\in\bF_q^{n^1_{\mathrm{in}}+n^2_{\mathrm{in}}}$ be an arbitrary vector, and let $\cF$ be an additive $\cE_{\mathrm{run}}$-avoiding fault. Because $\cF$ is additive, we can decompose each $F_t=F^1_t\sqcup F^2_t$ where each $F^i_t$ is an additive function acting on the dits of $\cR^i$ at timestep $t$. Then by definition $\cF^i=(\cF^i_1,\dots,\cF^i_T)$ is an additive $\cE^i_{\mathrm{run}}$-avoiding fault for $\cR^i$. Thus for $i\in[2]$, either $\cR^i[\cF^i](y_i)$ is a $\cE^i_{\mathrm{out}}$-deviation of $\bar{\cO}^i(x_i')$ for some $x_i'\in\bF_q^{k^i_{\mathrm{out}}}$, or else there exists $E\in\cE^i_{\mathrm{det}}$ with $E\subseteq\supp(\tran(\cR^i[\cF^i];y_i))$. By the definition of $\cR,\cE_{\mathrm{out}},\bar{\cO},\cE_{\mathrm{det}}$, it follows that either $\cR[\cF](y)=(\cR^1[\cF^1](y_1),\cR^2[\cF^2](y_2))$ is a $\cE_{\mathrm{out}}$-deviation of $\bar{\cO}(x_1',x_2')$, or else there exists $E\in\cE_{\mathrm{det}}$ with $E\subseteq\supp(\tran(\cR[\cF];y))$. Thus by \Cref{lem:fdadd}, the parallel composition gadget in \Cref{eq:pcgad} is mending, as desired.
\end{proof}

\section{Fault-Tolerant and Fault-Detecting Gadgets For Tensor Codes}
\label{sec:gadgets}
In this section, we present our fault-tolerant and fault-detecting gadgets for tensor codes. We will subsequently show how to compose these gadgets to compile an arbitrary logical circuit into a fault-tolerant or fault-detecting circuit.

\subsection{Error-Detection Gadget}
\label{sec:errdet}
Here present an error-detection gadget for tensor codes. Our gadget simply applies the parity-check matrix to the corrupted codeword, and declares an error has been detected if the resulting syndrome has sufficiently high weight. \Cref{lem:loctest} implies that this procedure successfully detects high-weight errors. In fact, such a procedure applies for arbitrary locally testable codes, though we only describe the tensor code setting that is relevant for our purposes.

\begin{lemma}
  \label{lem:errdet}
  For $u,n,d\in\bN$, for $i\in[u]$ let $C^i$ be a $[n_i=n,\; k_i,\; d_i\geq d]_q$ code with encoding map $\Enc^i$. Let $\rho(u,n,d)$ be the soundness parameter defined in \Cref{lem:loctest}. Let $C=\bigotimes_{i\in[u]}C^i$ be the tensor code with encoding map $\Enc=\bigotimes_{i\in[u]}\Enc^i$. Let $N=\prod_{i\in[u]}[n_i]=[n]^u$, $K=\prod_{i\in[u]}[k_i]$, and $D=d^u$, so that $C$ is a $[|N|,|K|,\geq D]_q$ code. Let $\lambda_{\mathrm{in}},\lambda_{\mathrm{run}},\lambda_{\mathrm{out}},\lambda_{\mathrm{det}}>0$ satisfy
  \begin{align}
    \label{eq:edlams}
    \begin{split}
      \lambda_{\mathrm{in}} &\leq D/2 \\
      \lambda_{\mathrm{out}} &\geq \frac{8u^2n^4}{\rho(u,n,d)}\cdot\lambda_{\mathrm{run}} \\
      \lambda_{\mathrm{det}} &= \frac{\rho(u,n,d)}{8}\cdot\lambda_{\mathrm{out}}.
    \end{split}
  \end{align}
  For $\alpha\in\{\mathrm{in},\mathrm{out}\}$, define the decorated code $D_\alpha=(C,\; \Enc,\; \cE_\alpha=2^N|_{\geq\lambda_\alpha})$.
  Then there exists a circuit $\cR$ using gate set $\{\gInit,\gTerm,\gCX^*\}$, space $|N'|\leq(u+1)n^u$, and time $T\leq un^2+2$ along with a set
  $E_{\mathrm{det}}\subseteq N'$
  such that
  \begin{equation}
    \label{eq:ecgad}
    (\cR,\; \cE_{\mathrm{run}}=2^{N'}|_{\geq\lambda_{\mathrm{run}}}^{\sqcup T},\; D_{\mathrm{in}},\; D_{\mathrm{out}},\; \cE_{\mathrm{det}}=2^{E_{\mathrm{det}}\times\{T-1\}}|_{\geq\lambda_{\mathrm{det}}})
  \end{equation}
  forms a mending fault-detecting gadget for the $|K|$-dit identity function $\bar{O}=I_K:\bF_q^K\rightarrow\bF_q^K$. Above, $E_{\mathrm{det}}\times\{T-1\}$ denotes the subset of $N'\times[T]\cong{N'}^{\sqcup T}$ given by a single copy of the subset $E_{\mathrm{det}}\subseteq N'$ within the $(T-1)$st copy of $N'$.
\end{lemma}

When applying \Cref{lem:errdet}, it will often be useful to set $\lambda_{\mathrm{out}}$ equal to its minimum allowed value $\lambda_{\mathrm{out}}=(8u^2n^4/\rho(u,n,d))\cdot\lambda_{\mathrm{run}}$. The advantage of choosing a larger value of $\lambda_{\mathrm{out}}$ is that $\lambda_{\mathrm{det}}$ becomes larger, meaning that more of the ``detector dits'' in $E_{\mathrm{det}}$ are nonzero when an error is detected.

\begin{proof}[Proof of \Cref{lem:errdet}]
  We will first describe the circuit $\cR$. For this purpose, for $i\in[u]$, fix a full-rank parity-check matrix $H^i\in\bF_q^{(n_i-k_i)\times k_i}$ for $C^i$. Define the set $S=\bigsqcup_{i\in[u]}S^i$ and the tensor-code parity-check matrix $H\in\bF_q^{S\times N}$ as in \Cref{def:pcpc}.
  
  On input $y\in\bF_q^N$, the desired circuit $\cR$ first applies $\gInit^{\sqcup S}$ to initialize a set of ancilla dits labeled by the set $S$. Next, $\cR$ applies the linear map $(y,0^S)\mapsto(y,Hy)$ by sequentially looping through $i=1,\dots,u$, and for each $i$ using $\gCX^*$ gates to apply the map
  \begin{equation}
    \label{eq:edsyn}
    (y,0^{S^i})\mapsto(y,\;(I^{\otimes i-1}\otimes H^i\otimes I^{\otimes u-i})y)
  \end{equation}
  to dits $N\sqcup S^i\subseteq N\sqcup S$. Specifically, to apply this map in \Cref{eq:edsyn}, $\cR$ simply applies the circuit $\cR_{H^i}$ from \Cref{def:matcirc} in parallel to each direction-$i$ column of $S^i$ and $N$, which takes time $(n_i-k_i)n_i\leq n_i^2=n^2$. Finally, $\cR$ applies $\gTerm^{\sqcup S}$ to terminate all dits in $S$.


  Thus $\cR$ uses 1 timestep for the $\gInit$ gates, at most $u\cdot n^2$ timesteps to apply the linear map $(y,0^S)\mapsto(y,Hy)$, and one timestep for the $\gTerm$ gates, for a total of $T\leq un^2+2$ timesteps. Meanwhile, $\cR$ uses space $N'=N\sqcup S$, so $|N'|\leq n^u+un^u\leq(u+1)n^u$.

  We define the set $E_{\mathrm{det}}\subseteq{N'}^{\sqcup T}=N'\times[T]$ to contain all dits in $S$ at timestep $T-1$, that is, $E_{\mathrm{det}}=S\times\{T-1\}$.

  We will now prove that the data in \Cref{eq:ecgad} forms a mending error-detecting gadget for $\bar{O}=I_K$. First, for every input $y\in C$ so that $Hy=0$, then in the absence of errors, following timestep $T-1$ the dits in $N'=N\sqcup S$ have value $(y,Hy)=(y,0^S)$. Thus the restriction of $\tran(\cR;y)$ to $E_{\mathrm{det}}=S\times\{T-1\}$ equals $0^S$, so \Cref{it:fdnoerr} in \Cref{def:faultdet} is satisfied.

  To show \Cref{it:fderr} as well as the mending property in \Cref{def:faultdet}, let $x\in\bF_q^{k_{\mathrm{in}}}$, let $y\in\bF_q^N$, and let $\cF$ be an additive $\cE_{\mathrm{run}}$-avoiding fault for $\cR$. Let $x'=\argmin_{x''\in\bF_q^K}|y-\Enc(x'')|$ and let $f_0=y-\Enc(x')$. To prove \Cref{it:fderr} in \Cref{def:faultdet}, we will assume $y$ is a $\cE_{\mathrm{in}}$-deviation of $\Enc(x)$, meaning that $|y-\Enc(x)|<\lambda_{\mathrm{in}}\leq D/2$ (see \Cref{eq:edlams}), and therefore $x'=x$. To prove the mending property, we instead assume $y$ is arbitrary, and we may have $x'\neq x$.

  For $t\in\{0,\dots,T\}$, let $N_t'\subseteq N'$ be the active dits at time $t$. For $t\in[T]$, define $f_t\in\bF_q^{N_t'}$ to be the vector such that $F_t(z)=z+f_t$. By assumption, for $t\in[T]$ we have $|f_t|<\lambda_{\mathrm{run}}$. By construction, a fault-free execution of $\cR$ (where $f_t=0$ for every $t\in[T]$) has $\tran(\cR;y)|_{E_{\mathrm{det}}}=Hy=Hf_0$. Meanwhile, in an execution of $\cR$ with the fault $\cF$, because every dit is involved in $\sum_{i\in[u]}n_i(n_i-k_i)<un^2$ $\gCX^*$ gates, every fault error (corresponding to a nonzero component of $f_t$ for some $t\in[T]$) can propagate to at most $un^2-1$ other dits. That is, in timestep $T-1$, the difference of the transcript $\tran(\cR[\cF];y)|_{E_{\mathrm{det}}}$ and the fault-free transcript $\tran(\cR;y)|_{E_{\mathrm{det}}}=Hf_0$ has weight
  \begin{align}
    \label{eq:ednoisysyn}
    |\tran(\cR[\cF];y)|_{E_{\mathrm{det}}}-Hf_0| \leq \sum_{t\in[T]}|f_t|\cdot un^2
    &< T\cdot\lambda_{\mathrm{run}}\cdot un^2 
  \end{align}
  dits. Also, because dits in $N$ can only be changed by fault errors (as all $\gCX^*$ gates we run use dits in $N$ as controls, not targets), the difference of the output $\cR[\cF](y)$ and the fault-free output $\cR(y)=y$ has weight
  \begin{align}
    \label{eq:edoutput}
    |\cR[\cF](y)-y|
    &< T\cdot\lambda_{\mathrm{run}}. 
  \end{align}

  If $\cR[\cF](y)$ is a $\cE_{\mathrm{out}}$-deviation of $\Enc(x')$, then \Cref{it:fdecorr} (if $y$ is a $\cE_{\mathrm{in}}$-deviation of $\Enc(x)=\Enc(x')$) or \Cref{it:fdecorrmend} (for general $y\in\bF_q^N$) in \Cref{def:faultdet} holds. Otherwise, we have $|\cR[\cF](y)-\Enc(x')|\geq\lambda_{\mathrm{out}}$, so because $y=\Enc(x')+f_0$, \Cref{eq:edoutput} implies that
  \begin{align}
    \label{eq:edf0}
    |f_0|
    &\geq \lambda_{\mathrm{out}}-|\cR[\cF](y)-y| > \lambda_{\mathrm{out}}-T\cdot\lambda_{\mathrm{run}}.
  \end{align}
  By definition $\Enc(x')$ is the closest codeword in $C$ to $y=\Enc(x')+f_0$, and therefore $0$ is the closest codeword in $C$ to $f_0$, i.e.~$|f_0-C|=|f_0|$. Therefore
  \begin{align*}
    |\tran(\cR[\cF];y)|_{E_{\mathrm{det}}}|
    &> |Hf_0|-T\cdot\lambda_{\mathrm{run}}\cdot un^2 \\
    &\geq \rho(u,n,d)\cdot|f_0|-T\cdot\lambda_{\mathrm{run}}\cdot un^2 \\
    &\geq \rho(u,n,d)\cdot(\lambda_{\mathrm{out}}-T\cdot\lambda_{\mathrm{run}}) - T\cdot\lambda_{\mathrm{run}}\cdot un^2 \\
    &\geq \rho(u,n,d)\cdot(\lambda_{\mathrm{out}}-(un^2+2)\cdot\lambda_{\mathrm{run}}) - un^2(un^2+2)\cdot\lambda_{\mathrm{run}} \\
    &\geq \frac{\rho(u,n,d)}{8}\cdot\lambda_{\mathrm{out}} \\
    &= \lambda_{\mathrm{det}}.
  \end{align*}
  where the first inequality above holds by \Cref{eq:ednoisysyn}, the second inequality holds by \Cref{lem:loctest} and because $\min_{c\in C}|f_0-c|=|f_0|$ by the definition of $f_0$, the third inequality holds by \Cref{eq:edf0}, the fourth inequality holds because $T\leq un^2+2$, and the fifth and sixth inequalities hold by \Cref{eq:edlams}. Thus some $E\in\cE_{\mathrm{det}}=2^{E_{\mathrm{det}}}|_{\geq\lambda_{\mathrm{det}}}$ must lie inside $\supp(\tran(\cR[\cF];y))$, so \Cref{it:fdedet} (if $y$ is a $\cE_{\mathrm{in}}$-deviation of $\Enc(x)$) and \Cref{it:fdecorrmend} (for general $y\in\bF_q^N$) in \Cref{def:faultdet} are satisfied. Therefore by \Cref{lem:fdadd}, our gadget in \Cref{eq:ecgad} is indeed mending fault-detecting for $\bar{O}=I_K$.
\end{proof}

\subsection{Error-Correction Gadget}
In this section, we present a fault-tolerant gadget that \emph{corrects} errors on tensor codes, rather than simply detecting them as in \Cref{sec:errdet}. Such error-correction is not necessary for the PCP setting, and we are not able to correct as many errors as we can detect in \Cref{lem:errdet}. However, the ability to perform error-correction enables us to obtain a complete fault-tolerance scheme for universal classical computation, as opposed to simply a fault-detecting scheme. Such a fault-tolerance scheme protecting against a large number of adversarial errors may be of independent interest.

We begin with the following notion of a decoder, which captures the usual notion in the coding theory literature.

\begin{definition}
  \label{def:raddec}
  A \emph{radius-$\lambda$ decoder} for a $[n,k,d]_q$ code $C$ with encoding map $\Enc$ is a fault-tolerant gadget
  \begin{equation*}
    (\cR,\; \cE_{\mathrm{run}}=2^{N'},\; D_{\mathrm{in}}=(C,\Enc,2^{[n]}|_{\geq\lambda}),\; D_{\mathrm{out}}=(C,\Enc,2^{[n]}))
  \end{equation*}
  for the $k$-dit identity function $\bar{O}=I_k$, acting on some set of dits $N'$.
\end{definition}

In words, a radius-$\lambda$ decoder for a code $C$ is a circuit that takes as input any $y\in\bF_q^n$ that differs from a codeword $c\in C$ in $<\lambda$ dits, and outputs $c$. Such a decoder can only exist when there exists a unique such codeword $c\in C$, meaning we must have $\lambda-1<d/2$. Because $\cE_{\mathrm{run}}=2^{N'}$ in \Cref{def:raddec}, the decoder does not need to output a meaningful answer in the presence of a fault.

Below, we state the well-known result that there exist polynomial-time decoders for Reed-Solomon codes (see \Cref{sec:reedsol}) with decoding radius equal to half the distance. While there is a large literature showing stronger results, such as with more efficient algorithms or more robust decoding properties, we intentionally state the basic result that will be sufficient for our purposes.

\begin{lemma}[Welch-Berlekamp decoder; see e.g.~Chapter 17 of \cite{guruswami_essential_2022}]
  \label{lem:RSdec}
  For a prime power $q$, a subset $E\subseteq\bF_q$, and a positive integer $k\leq n$, let $C=\RS(q,k,E)$ be the $[n=|E|,\; k,\; d=n-k+1]_q$ Reed-Solomon code. Then there exists a radius-$d/2$ decoder for $C$ using space $|N'|=O(n^2\log q)$, time $T=O(n^2\log q)$, and gate set $\cG=\{\gInit,\gTerm,\gX^*,\gCX^*,\gCCX^*\}$ over $q$-ary dits.
\end{lemma}

\begin{remark}
  The space and time bounds in \Cref{lem:RSdec} are not optimal, but are sufficient for our purposes. We now briefly justify them. The Welch-Berlekamp algorithm (see e.g.~\cite[Chapter 17]{guruswami_essential_2022}) simply consists of running Gaussian elimination over $\bF_q$ a constant number of times, where each system of equations has $O(n)$ variables and $O(n)$ linear constraints. Gaussian elimination on such a system can be performed using $O(n^2)$ $q$-ary dits and $O(n^2)$ row/column operations (i.e.~additions or scalar multiplications of entire rows or columns of a $O(n)\times O(n)$ matrix). Each such row/column operation takes $O(n)$ space and $O(1)$ time in our circuit model. The $O(\log q)$ factors in the space and time usage in \Cref{lem:RSdec} arise from the algorithm's control operations, which decide which row/column operations to perform. There are $O(n^2)$ such control operations, each of which either adds or multiplies two elements of $\bF_q$, or else must map an element $a\in\bF_q$ to $0$ if $a=0$ and to $1$ if $a\neq 0$. In the latter case, we must simply compute $a^{q-1}$, which can be implemented using $O(\log q)$ space and time via repeated squaring with the gate set $\cG=\{\gInit,\gTerm,\gX^*,\gCX^*,\gCCX^*\}$. Hence the Welch-Berlekamp algorithm can be implemented using $O(n^2\log q)$ space and time with gate set $\cG=\{\gInit,\gTerm,\gX^*,\gCX^*,\gCCX^*\}$ over $q$-ary dits.
\end{remark}

We now present our error-correction gadget for tensor codes, assuming the factor codes have decoders in the sense of \Cref{def:raddec}.

\begin{lemma}
  \label{lem:errcorr}
  For $u,n,d\in\bN$, for $i\in[u]$ let $C^i$ be a $[n_i=n,\; k_i,\; d_i\geq d]_q$ code with encoding map $\Enc^2$. Assume that each $C^i$ has a radius-$\lambda_{\Dec}\geq 1$ decoder $\cR_{\Dec}^i$ using space $N_{\Dec}$, time $T_{\Dec}$, and gate set $\cG_{\Dec}$. Let $C=\bigotimes_{i\in[u]}C^i$ be the tensor code with encoding map $\Enc=\bigotimes_{i\in[u]}\Enc^i$. Let $N=\prod_{i\in[u]}[n_i]=[n]^u$ and $K=\prod_{i\in[u]}[k_i]$, so that $C$ is a $[|N|,|K|,\geq d^u]_q$ code. For $\lambda_{\mathrm{run}}\geq 0$, let
  \begin{align}
    \label{eq:eclams}
    \begin{split}
      \lambda_{\mathrm{in}} &= \left(\frac{\lambda_{\Dec}}{2}\right)^u \\
      \lambda_{\mathrm{out}} &= \left(\frac{2^u+1}{\lambda_{\Dec}/2n}\right)^u\cdot u^2n\cdot T_{\Dec}\cdot\lambda_{\mathrm{run}}
    \end{split}
  \end{align}
  For $\alpha\in\{\mathrm{in},\mathrm{out}\}$, define the decorated code $D_\alpha=(C,\; \Enc,\; \cE_\alpha=2^N|_{\geq\lambda_\alpha})$.
  Then there exists a circuit $\cR$ using space $|N'|=n^{u-1}\cdot|N_{\Dec}|$, time $T=u\cdot T_{\Dec}$, and gate set $\cG_{\Dec}$ such that 
  \begin{equation}
    \label{eq:ecgad}
    (\cR,\; \cE_{\mathrm{run}}=(2^{N'}|_{\geq\lambda_{\mathrm{run}}})^{\sqcup T},\; D_{\mathrm{in}},\; D_{\mathrm{out}})
  \end{equation}
  forms a fault-tolerant gadget for the $|K|$-dit identity function $\bar{O}=I_K$.
\end{lemma}

To prove \Cref{lem:errcorr}, we will use the following basic combinatorial notion from \cite{kalachev_maximally_2025} (though we slighly modify the notation).

\begin{definition}
  \label{def:setclos}
  For $u,n\in\bN$ and $\lambda>0$, we say a set $E\subseteq[n]^u$ is \emph{$\lambda$-closed} if for every $i\in[u]$ and every direction-$i$ column $L\subseteq[n]^u$, then either $|E\cap L|<\lambda$ or $E\cap L=L$ (i.e.~$|E\cap L|=n$). For $E\subseteq[n]^u$, we then define the \emph{$\lambda$-closure} $[E]_\lambda\subseteq[n]^u$ to be the minimal $\lambda$-closed set that contains $E$.
\end{definition}

\begin{lemma}[\cite{kalachev_maximally_2025}]
  \label{lem:closbound}
  For every $E\subseteq[n]^u$, it holds that
  \begin{align*}
    |[E]_\lambda| &\leq \left(\frac{2^u+1}{\lambda/n}\right)^u\cdot|E|.
  \end{align*}
\end{lemma}

\begin{proof}[Proof of \Cref{lem:errcorr}]
  The desired circuit $\cR$ simply loops through $i=1,\dots,u$, and for each such $i$ runs the decoder $\cR^i_{\Dec}$ in parallel on every direction-$i$ column in $N$. As there are $n^{u-1}$ such direction-$i$ columns, $\cR$ uses space $|N'|=n^{u-1}\cdot|N_{\Dec}|$, and uses time $T=u\cdot T_{\Dec}$.

  We now show that the gadget in \Cref{eq:ecgad} forms a fault-tolerant gadget for $\bar{O}=I_K$. For $c\in C$, then the restriction to every direction-$i$ column lies in $C^i$, and $\cR^i_{\Dec}$ preserves codewords of $C^i$. Therefore $\cR(c)=c$, so \Cref{it:ftnoerr} in \Cref{def:faulttol} holds. It remains to prove \Cref{it:fterr} in \Cref{def:faulttol}.
  
  For this purpose, we first introduce the following notation. For $\lambda>0$, for $E,F\subseteq N=[n]^u$, and for $i\in\{0,\dots,u\}$, we define $S_i^\lambda[F](E)\in\bF_q^N$ inductively as follows. First, we let $S_0^\lambda[F](E)=E$. Then for $i\in[u]$, we let $S_i^\lambda[F](E)$ be the union of all direction-$i$ columns $L$ for which either $|S_{i-1}^\lambda[F](E)\cap L|\geq\lambda$ or $|F\cap L|>0$. In particular, it follows that $S_i^\lambda[F](E)$ contains all or none of each direction-$i$ column. We let $S_i^\lambda(E)=S_i^\lambda[\emptyset](E)$.

  Now let $x\in\bF_q^K$ and let $y$ be a $\cE_{\mathrm{in}}$-deviation of $\Enc(x)$, so that we can write $y=c+e$ where $c=\Enc(x)$ and $e\in\bF_q^N$ has $|e|<\lambda_{\mathrm{in}}$. Let $\cF=(F_1,\dots,F_T)$ be a $\cE_{\mathrm{run}}$-avoiding fault for the circuit $\cR$ defined above. Our goal is to show that $|\cR[\cF](y)-c|<\lambda_{\mathrm{out}}$.

  For $i\in\{0,\dots,u\}$, let $\cR_{\leq i\cdot T_{\Dec}}[\cF](y)\in\bF_q^N$ denote the state of dits in $N$ after running the first $i\cdot T_{\Dec}$ timesteps of $\cR[\cF](y)$, i.e.~after running the decoders $\cR^{i'}_{\Dec}$ in the first $i$ directions $i'=1,\dots,i$ on input $y=c+e$. Also define $F=\bigcup_{t\in[T]}\supp(\cF_t)\subseteq N$ to be the set of all dits that lie in the support of the fault $\cF$, across all timesteps. Therefore
  \begin{equation}
    \label{eq:ecFbound}
    |F| \leq |\cF| < T\cdot\lambda_{\mathrm{run}} = u\cdot T_{\Dec}\cdot\lambda_{\mathrm{run}}.
  \end{equation}
  In the following claim, we show that $S_I^\lambda[F](\supp(e))$ contains the set of errors remaining after these first $i$ rounds of decoding.

  \begin{claim}
    \label{claim:RFmcbound}
    For $i\in\{0,\dots,u\}$, we have $\supp(\cR_{\leq i\cdot T_{\Dec}}[\cF](y)-c)\subseteq S_i^{\lambda_{\Dec}}[F](\supp(e))$.
  \end{claim}
  \begin{proof}
    We prove the claim by induction on $i$. For the base case, when $i=0$ then by definition $\cR_{\leq 0}[\cF](y)-c=e$ and $S_0^{\lambda_{\Dec}}[F](\supp(e))=\supp(e)$. For the inductive step, for $i\in[u]$, assume the claim holds for $i-1$. For a direction-$i$ column $L$, if the decoder $\cR^i_{\Dec}[\cF]$ (with $\cF$ restricted to dits $L$ and timesteps $(i-1)T_{\Dec}+1,\dots,iT_{\Dec}$) applied to dits $\cR_{\leq (i-1)\cdot T_{\Dec}}[\cF](y)|_L$ fails to output $c|_L$, then either $|(\cR_{\leq (i-1)\cdot T_{\Dec}}[\cF](y)-c)|_L|\geq\lambda_{\Dec}$, or else $\cF$ has nonempty support inside $L$. If $|(\cR_{\leq (i-1)\cdot T_{\Dec}}[\cF](y)-c)|_L|\geq\lambda_{\Dec}$, then by the inductive hypothesis $|S_{i-1}^{\lambda_{\Dec}}[F](\supp(e))\cap L|\geq\lambda_{\Dec}$, and then by the definition of $S_i^{\lambda_{\Dec}}$ we must have $L\subseteq S_i^{\lambda_{\Dec}}[F](\supp(e))$. If $\cF$ has nonempty support inside $L$, then $|F\cap L|>0$, so again by the definition of $S_i^{\lambda_{\Dec}}$ we must have $L\subseteq S_i^{\lambda_{\Dec}}[F](\supp(e))$. Thus for every direction-$i$ column $L$ that has nonempty intersection with $\supp(\cR_{\leq i\cdot T_{\Dec}}[\cF](y)-c)$, we have $L\subseteq S_i^{\lambda_{\Dec}}[F](\supp(e))$, completing the inductive step.
  \end{proof}

  Define $F'\subseteq N=[n]^u$ to be the union over all $i\in[u]$ of all direction-$i$ columns that have a nonempty intersection with $F$. Because every point in $F$ lies in $u$ direction-$i$ columns, each containing $n$ elements, we have
  \begin{equation}
    \label{eq:Fpbound}
    |F'| \leq un\cdot|F|.
  \end{equation}
  The following claim shows that any point lying in $S_i^\lambda[F](E)$ but not in $S_i^{\lambda/2}(E)=S_i^{\lambda/2}[\emptyset](E)$ must lie in the $\lambda/2$-closure $[F']_{\lambda/2}$ (see \Cref{def:setclos}) of $F'$.

  \begin{claim}
    \label{claim:SFbound}
    For $\lambda>0$, $E\subseteq[n]^u$, and $i\in\{1,\dots,u\}$, we have $S_i^\lambda[F](E)\subseteq S_i^{\lambda/2}(E)\cup[F']_{\lambda/2}$
  \end{claim}
  \begin{proof}
    We prove the claim by induction on $i$. For the base case, $S_0^\lambda[F](E)=E=S_i^{\lambda/2}(E)$. For the inductive step, for some $i\in[u]$ assume the claim holds for $i-1$. For every direction-$i$ column $L\subseteq S_i^\lambda{\lambda/2}(E)$, then by definition either $|S_{i-1}^\lambda[F](E)\cap L|\geq\lambda$ or $|F\cap L|>0$. If $|F\cap L|>0$, then by definition $L\subseteq F'\subseteq[F']_{\lambda/2}$. Meanwhile, if $|S_{i-1}^\lambda[F](E)\cap L|\geq\lambda$, then by the inductive hypothesis either $|S_{i-1}^{\lambda/2}(E)\cap L|\geq\lambda/2$ or $|[F']_{\lambda/2}\cap L|\geq\lambda/2$. If $|S_{i-1}^{\lambda/2}(E)\cap L|\geq\lambda/2$ then by definition $L\subseteq S_i^{\lambda/2}(E)$, while if $|[F']_{\lambda/2}\cap L|\geq\lambda/2$ then $L\subseteq[F']_{\lambda/2}$ by \Cref{def:setclos}. Thus in all cases, if the direction-$i$ column $L\subseteq S_i^\lambda{\lambda/2}(E)$, then $L\subseteq S_i^{\lambda/2}(E)\cup[F']_{\lambda/2}$, completing the inductive step.
  \end{proof}

  \Cref{claim:RFmcbound}, \Cref{claim:SFbound}, \Cref{lem:closbound}, and \Cref{eq:Fpbound} together imply that
  \begin{align*}
    |\cR[\cF](y)-c|
    &= |\cR_{\leq u\cdot T_{\Dec}}[\cF](y)-c| \\
    &\leq |S_u^{\lambda_{\Dec}}[F](\supp(e))| \\
    &\leq |S_u^{\lambda_{\Dec}/2}(\supp(e))|+[F']_{\lambda_{\Dec}/2} \\
    &\leq |S_u^{\lambda_{\Dec}/2}(\supp(e))|+\left(\frac{2^u+1}{\lambda_{\Dec}/2n}\right)^u\cdot un\cdot|F| \\
    &< |S_u^{\lambda_{\Dec}/2}(\supp(e))|+\left(\frac{2^u+1}{\lambda_{\Dec}/2n}\right)^u\cdot u^2n\cdot T_{\Dec}\cdot\lambda_{\mathrm{run}} \\
    &= |S_u^{\lambda_{\Dec}/2}(\supp(e))|+\lambda_{\mathrm{out}}.
  \end{align*}

  The following claim shows that the $|S_u^{\lambda_{\Dec}/2}(\supp(e))|$ term above vanishes as long as $|\supp(e)|=|e|<(\lambda_{\Dec}/2)^u$, thereby completing the proof of fault-tolerance.

  \begin{claim}
    \label{claim:decnonoise}
    For $\lambda>0$ and for $E\subseteq[n]^u$ of size $|E|<\lambda^u$, we have $S_u^\lambda(E)=\emptyset$.
  \end{claim}
  \begin{proof}
    We show the result by induction on $u$. For the base case, when $u=1$, if $|E|=|S_0^\lambda(E)|<\lambda$ then by definition $S_u^\lambda(E)=\emptyset$. For the inductive step, for some $u\in\bN$, assume the claim holds when $u$ is replaced by $u-1$. Now for $E\subseteq[n]^u$, if $|E|<\lambda^u$ then there are $<\lambda^{u-1}$ direction-$1$ columns $L$ for which $|E\cap L|=|S_0^\lambda(E)\cap L|\geq\lambda$. Thus $S_1^\lambda(E)$ contains $<\lambda^{u-1}$ direction-$1$ columns, so for every $j_1\in[n]$, we have $|S_1^\lambda(E)\cap(\{j_1\}\times[n]^{u-1})|<\lambda^{u-1}$. Now by definition, the restriction of $S_1^\lambda(E),\dots,S_u^\lambda(E)$ to points in $(\{j_1\}\times[n]^{u-1})$ equals the exact sequence $\bar{S}_0^\lambda(\bar{E}),\dots,\bar{S}_{u-1}^\lambda(\bar{E})$ obtained by defining $\bar{S}$ like we defined $S$ but with $u$ replaced by $u-1$, and setting $\bar{E}=S_1^\lambda(E)\cap(\{j_1\}\times[n]^{u-1})$. Thus by the inductive hypothesis we have $S_u^\lambda(E)\cap(\{j_1\}\times[n]^{u-1})=\bar{S}_{u-1}^\lambda(\bar{E})=\emptyset$ for every $j_1\in[u]$, and hence $S_u^\lambda(E)=\emptyset$, completing the inductive step.
  \end{proof}
\end{proof}

\subsection{Code Switching Gadgets}
In this section, we present code-switching gadgets for tensor codes. Specifically, for $i\in[u]$, let $C^i$ be a $k_i$-dimensional code. In \Cref{lem:switchone} below, given a codeword of the tensor code $C^{[u]}=\bigotimes_{i\in[u]}C^i$ that encodes some message $x$, we show how to unencode one of the codes $C^i$ to switch to the code $C^{[u]\setminus\{i\}}(\bigotimes_{i'\in[u]\setminus\{i\}}C^{i'})^{\sqcup k_i}$ while preserving the underlying message $x$. Similarly, given a codeword of $C^{[u]\setminus\{i\}}$, we can encode $C^i$ to switch back to $C^{[u]}$ while preserving the message. Then in \Cref{lem:switchall}, we show how to iteratively apply these gadgets to each $i=1,\dots,u$ along with repeated applications of an error-detection (\Cref{lem:errdet}) or error-correction (\Cref{lem:errcorr}) gadget in order to switch between two tensor codes in which all $u$ factors $C^i$ may differ.

\begin{lemma}
  \label{lem:switchone}
  For $u\in\bN$, for $i\in[u]$ let $C^i$ be a $[n_i,k_i]_q$ code with encoding map $\Enc^i$. For $I\subseteq[u]$, let $N^I=\prod_{i\in I}[n_i]$, $K^I=\prod_{i\in I}[k_i]$, and let $C^I=\bigotimes_{i\in I}C^i$ be the $[|N^I|,|K^I|]_q$ tensor product code with encoding map $\Enc^I=\bigotimes_{i\in I}\Enc^i$. Fix an arbitrary $\bar{i}\in[u]$. Then the following hold:
  \begin{enumerate}
  \item\label{it:sodown} (Downwards switching) For $\lambda_{\mathrm{in}},\lambda_{\mathrm{run}}\geq 0$, let
    \begin{equation*}
      \lambda_{\mathrm{out}} = n_{\bar{i}}\cdot\lambda_{\mathrm{in}}+3n_{\bar{i}}^3\cdot\lambda_{\mathrm{run}}
    \end{equation*}
    and define decorated codes
    \begin{align*}
      D_{\mathrm{in}} &= \left(C^{[u]},\; \Enc^{[u]},\; \cE_{\mathrm{in}}=2^{N^{[u]}}|_{\geq\lambda_{\mathrm{in}}}\right) \\
      D_{\mathrm{out}} &= \left(C^{[u]\setminus\{\bar{i}\}},\; \Enc^{[u]\setminus\{\bar{i}\}},\; \cE_{\mathrm{out}}=2^{N^{[u]\setminus\{\bar{i}\}}}|_{\geq\lambda_{\mathrm{out}}}\right).
    \end{align*}
    Then there exists a circuit $\cR$ using gate set $\{\gInit,\gTerm,\gCX^*\}$, space $|N'|\leq 2|N^{[u]}|$, and time $T\leq n_{\bar{i}}^2+2$ such that $(\cR,\; \cE_{\mathrm{run}}=2^{N'}|_{\geq\lambda_{\mathrm{run}}}^{\sqcup T},\; D_{\mathrm{in}},\; D_{\mathrm{out}}^{\sqcup k_{\bar{i}}})$ forms a fault-tolerant gadget for the $K^{[u]}$-dit identity channel $\bar{O}=I_{K^{[u]}}:\bF_q^{K^{[u]}}\rightarrow\bF_q^{K^{[u]}}$.
  \item\label{it:soup} (Upwards switching) For $\lambda_{\mathrm{in}},\lambda_{\mathrm{run}}\geq 0$, let
    \begin{equation*}
      \lambda_{\mathrm{out}} = n_{\bar{i}}^2\cdot\lambda_{\mathrm{in}}+3n_{\bar{i}}^3\cdot\lambda_{\mathrm{run}}
    \end{equation*}
    and define decorated codes
    \begin{align*}
      D_{\mathrm{in}} &= \left(C^{[u]\setminus\{\bar{i}\}},\; \Enc^{[u]\setminus\{\bar{i}\}},\; \cE_{\mathrm{in}}=2^{N^{[u]\setminus\{\bar{i}\}}}|_{\geq\lambda_{\mathrm{in}}}\right) \\
      D_{\mathrm{out}} &= \left(C^{[u]},\; \Enc^{[u]},\; \cE_{\mathrm{out}}=2^{N^{[u]}}|_{\geq\lambda_{\mathrm{out}}}\right).
    \end{align*}
    Then there exists a circuit $\cR$ using gate set $\{\gInit,\gTerm,\gCX^*\}$, space $|N'|\leq 2|N^{[u]}|$, and time $T\leq n_{\bar{i}}^2+2$ such that $(\cR,\; \cE_{\mathrm{run}}=2^{N'}|_{\geq\lambda_{\mathrm{run}}}^{\sqcup T},\; D_{\mathrm{in}}^{\sqcup k_{\bar{i}}},\; D_{\mathrm{out}}^{\sqcup k_{\bar{i}}})$ forms a fault-tolerant gadget for the $K^{[u]}$-dit identity channel $\bar{O}=I_{K^{[u]}}:\bF_q^{K^{[u]}}\rightarrow\bF_q^{K^{[u]}}$.
  \end{enumerate}
\end{lemma}
\begin{proof}
  \begin{enumerate}
  \item (Downwards switching) Fix some matrix $U\in\bF_q^{k_{\bar{i}}\times n_{\bar{i}}}$ such that $U\circ\Enc^{\bar{i}}=I_{k_{\bar{i}}}$. Our desired circuit $\cR$ first uses $\gInit$ gates to initialize an ancilla block of dits labeled by $(N^{[u]\setminus\{\bar{i}\}})^{\sqcup k_{\bar{i}}}=N^{[u]\setminus\{\bar{i}\}}\times[k_{\bar{i}}]$. Next, $\cR$ applies the circuit $\cR_{U}$ from \Cref{def:matcirc} in parallel to each direction-$\bar{i}$ column of the ancilla block $N^{[u]\setminus\{\bar{i}\}}\times[k_{\bar{i}}]$ and the input block $N^{[u]}$. Specifically, each such direction-$\bar{i}$ column is specified by some $j_{-\bar{i}}\in N^{[u]\setminus\{\bar{i}\}}$, and consists of the $k_{\bar{i}}$ dits labeled $\{j_{-\bar{i}}\}\times[k_{\bar{i}}]$ in the ancilla block, along with the $n_{\bar{i}}$ dits labeled $\{j_{-\bar{i}}\}\times[n_{\bar{i}}]$ in the input block (where for the input block we push the $[n_{\bar{i}}]$ factor to the $i$th of the $u$ components). Finally, $\cR$ applies $\gTerm$ gates to the input block of dits $N^{[u]}$, and returns the ancilla block $(N^{[u]\setminus\{\bar{i}\}})^{\sqcup k_{\bar{i}}}$.

    Therefore $\cR$ uses space $|N'|\leq|N^{[u]}|+k_{\bar{i}}|N^{[u]\setminus\{\bar{i}\}}|\leq 2|N^{[u]}|$ and time $T=1+(n_{\bar{i}}-k_{\bar{i}})n_{\bar{i}}+1\leq n_{\bar{i}}^2+2$.

    For $x\in\bF_q^{K^{[u]}}$, in the absence of an input error and a fault, then
    \begin{align*}
      \cR\circ\Enc^{[u]}(x)
      &= (I^{\otimes\bar{i}-1}\otimes U\otimes I^{\otimes u-\bar{i}})\circ\left(\bigotimes_{i\in[u]}\Enc^i\right)(x) \\
      &= \left(\bigotimes_{i=1}^{\bar{i}}\Enc^i\otimes I_{k_{\bar{i}}}\otimes\bigotimes_{i=\bar{i}+1}^u\Enc^i\right)(x) \\
      &= (\Enc^{[u]\setminus\{\bar{i}\}})^{\sqcup k_{\bar{i}}}(x).
    \end{align*}
    If instead $y=\Enc^{[u]}(x)+f_0$ with $|f_0|<\lambda_{\mathrm{in}}$, and $\cF$ is an additive $\cE_{\mathrm{run}}$-avoiding fault for $\cF$, then upon running $\cR[\cF](y)$, at most $|f_0|+|\cF|<\lambda_{\mathrm{in}}+T\cdot\lambda_{\mathrm{run}}$ direction-$i$ columns will experience an error. Here an ``error'' means a dit at some point in time whose value differs from that of the noiseless execution of $\cR(\Enc^{[u]}(x))$ described above; we crucually use that every gate in $\cR$ acts on dits entirely contained inside some direction-$i$ column. Thus
    \begin{align*}
      |\cR[\cF](y)-(\Enc^{[u]\setminus\{\bar{i}\}})^{\sqcup k_{\bar{i}}}(x)|
      &= |\cR[\cF](y)-\cR(\Enc^{[u]}(x))| \\
      &< n_{\bar{i}}\cdot(\lambda_{\mathrm{in}}+T\cdot\lambda_{\mathrm{run}}) \\
      &\leq \lambda_{\mathrm{out}},
    \end{align*}
    where the final inequality above holds by the definition of $\lambda_{\mathrm{out}}$ becuase $T\leq n_{\bar{i}}^2+2\leq 3n_{\bar{i}}^2$.
  \item (Upwards switching) The proof is analogous as in the downwards switching case above, except now we apply the encoding map $\Enc^{\bar{i}}$ instead of the unencoding map $U$ within each direction-$i$ column. Specifically, $\cR$ first uses $\gInit$ gates to initialize an ancilla block of dits labeled by $N^{[u]}$. Next, $\cR$ applies the circuit $\cR_{\Enc^{\bar{i}}}$ from \Cref{def:matcirc} in parallel to each direction-$\bar{i}$ column of the ancilla block $N^{[u]}$ and the input block $(N^{[u]\setminus\{\bar{i}\}})^{\sqcup k_{\bar{i}}}$. Finally, $\cR$ applies $\gTerm$ gates to the input block, and returns the ancilla block.

    Therefore $\cR$ uses space $|N'|\leq k_{\bar{i}}|N^{[u]\setminus\{\bar{i}\}}|+|N^{[u]}|\leq 2|N^{[u]}|$ and time $T=1+n_{\bar{i}}(n_{\bar{i}}-k_{\bar{i}})+1\leq n_{\bar{i}}^2+2$.

    The proof of correctness and fault-tolerance is exactly analogous as in the downwards switching case above, so we omit the details to avoid redundancy. In brief, in the absence of errors, then $\cR$ simply applies $\Enc^{\bar{i}}$ to each direction-$i$ column, hence mapping $(\Enc^{[u]\setminus\{\bar{i}\}})^{\sqcup k_{\bar{i}}}(x)$ to $\Enc^{[u]}(x)$. In the presence of up to $\lambda_{\mathrm{in}}$ input errors on each of the $k_{\bar{i}}\leq n_{\bar{i}}$ input code blocks of $D_{\mathrm{in}}$, along with up to $\lambda_{\mathrm{run}}$ fault errors in each of the $T$ timesteps, then $<n_{\bar{i}}\cdot\lambda_{\mathrm{in}}+T\cdot\lambda_{\mathrm{run}}$ direction-$i$ columns will experience an error, so the output error has weight $<n_{\bar{i}}(n_{\bar{i}}\cdot\lambda_{\mathrm{in}}+T\cdot\lambda_{\mathrm{run}})\leq\lambda_{\mathrm{out}}$.
  \end{enumerate}
\end{proof}

In \Cref{lem:switchall} below, we repeatedly apply \Cref{lem:switchone} to swap out all $u$ codes $C^i$ in the tensor code $\bigotimes_{i\in[u]}C^i$.


\begin{lemma}
  \label{lem:switchall}
  For $u,n,d\in\bN$, for $i\in[u]$ and $\alpha\in\{\mathrm{in},\mathrm{out}\}$ let $C_\alpha^i$ be a $[n_i=n,\; k_{\alpha,i},\; d_{\alpha,i}\geq d]_q$ code with encoding map $\Enc_\alpha^i$. Let $N=\prod_{i\in [u]}[n_i]=[n]^u$, $K_\alpha=\prod_{i\in[u]}[k_{\alpha,i}]$, and let $C_\alpha=\bigotimes_{i\in[u]}C_\alpha^i$ be the tensor code with encoding map $\Enc_\alpha=\bigotimes_{i\in[u]}C_\alpha^i$. Then the following hold:
  \begin{enumerate}
  \item\label{it:sadet} (Fault-detecting switching) Let $\rho(u,n,d)$ be the soundness parameter defined in \Cref{lem:loctest}, and let $\lambda_{\mathrm{run}},\lambda_{\mathrm{in}},\lambda_{\mathrm{out}},\lambda_{\mathrm{det}}>0$ satisfy
    \begin{align}
      \label{eq:sadlams}
      \begin{split}
        \lambda_{\mathrm{in}} &= \frac{d^u}{2} \\
        \lambda_{\mathrm{run}} &\leq \frac{d^u\cdot\rho(u,n,d)}{32u^2n^7} \\
        \lambda_{\mathrm{out}} &= \frac{8u^2n^4}{\rho(u,n,d)}\cdot\lambda_{\mathrm{run}} \\
        \lambda_{\mathrm{det}} &= u^2n^4\cdot\lambda_{\mathrm{run}}.
      \end{split}
    \end{align}
    For $\alpha\in\{\mathrm{in},\mathrm{out}\}$, define the decorated code $D_\alpha=(C_\alpha,\; \Enc_\alpha,\; \cE_\alpha=2^{N}|_{\geq\lambda_\alpha})$. Then there exists a circuit $\cR$ using gate set $\cG=\{\gInit,\gTerm,\gCX^*\}$, space $|N'|\leq(u+1)n^u$, and time $T\leq 16u^2n^2$ along with sets $E_{\mathrm{det}}\subseteq N'$, $T_{\mathrm{det}}\subseteq[T]$ such that
    \begin{equation}
      \label{eq:sadgad}
      (\cR,\; \cE_{\mathrm{run}}=2^{N'}|_{\geq\lambda_{\mathrm{run}}}^{\sqcup T},\; D_{\mathrm{in}},\; D_{\mathrm{out}},\; \cE_{\mathrm{det}}=2^{E_{\mathrm{det}}}|_{\geq\lambda_{\mathrm{det}}}^{\sqcup T_{\mathrm{det}}})
    \end{equation}
    forms a mending fault-detecting gadget for the function $\bar{O}:\bF_q^{K_{\mathrm{in}}}\rightarrow\bF_q^{K_{\mathrm{out}}}$ given by
    \begin{equation}
      \label{eq:safun}
      \bar{O}(x) = (x|_{K_{\mathrm{in}}\cap K_{\mathrm{out}}},\; 0^{K_{\mathrm{out}}\setminus K_{\mathrm{in}}})
    \end{equation}
    that preserves all input components in $K_{\mathrm{in}}\cap K_{\mathrm{out}}$, and outputs $0$ in all output components outside of $K_{\mathrm{in}}$.
  \item\label{it:sacorr} (Fault-tolerant switching) Assume that each $C_\alpha^i$ has a radius-$\lambda_{\Dec}\geq 1$ decoder $\cR_{\Dec}^i$ using space $N_{\Dec}$, time $T_{\Dec}$, and gate set $\cG_{\Dec}$. Let $\lambda_{\mathrm{run}},\lambda_{\mathrm{in}},\lambda_{\mathrm{out}}>0$ satisfy
    \begin{align}
      \label{eq:saclams}
      \begin{split}
        \lambda_{\mathrm{in}} &= \left(\frac{\lambda_{\Dec}}{2}\right)^u \\
        \lambda_{\mathrm{run}} &\leq \frac{1}{n^5T_{\Dec}}\cdot\left(\frac{\lambda_{\Dec}^2/n}{2^{u+8}}\right)^u \\
        \lambda_{\mathrm{out}} &= \left(\frac{2^{u+4}}{\lambda_{\Dec}/n}\right)^u\cdot nT_{\Dec}\cdot\lambda_{\mathrm{run}}.
      \end{split}
    \end{align}
    For $\alpha\in\{\mathrm{in},\mathrm{out}\}$, define the decorated code $D_\alpha=(C_\alpha,\; \Enc_\alpha,\; \cE_\alpha=2^{N}|_{\geq\lambda_\alpha})$. Then there exists a circuit $\cR$ using gate set $\cG=\{\gInit,\gTerm,\gCX^*\}\cup\cG_{\Dec}$, space $|N'|\leq\max\{n^{u-1}\cdot|N_{\Dec}|,\; 2n^u\}$, and time $T\leq 2u^2T_{\Dec}+8un^2$ such that
    \begin{equation}
      \label{eq:sacgad}
      (\cR,\; \cE_{\mathrm{run}}=2^{N'}|_{\geq\lambda_{\mathrm{run}}}^{\sqcup T},\; D_{\mathrm{in}},\; D_{\mathrm{out}})
    \end{equation}
    forms a fault-tolerant gadget for the function $\bar{O}:\bF_q^{K_{\mathrm{in}}}\rightarrow\bF_q^{K_{\mathrm{out}}}$ in \Cref{eq:safun}.
  \end{enumerate}
\end{lemma}
\begin{proof}
  \begin{enumerate}
  \item The desired circuit $\cR$ loops through $i=1,\dots,u$, and for each $u$ applies the following sequence of gadgets:
    \begin{enumerate}[label=(\arabic*)]
    \item The error-detecting gadget in \Cref{lem:errdet} for the code $\bigotimes_{i'=1}^{i-1}C_{\mathrm{out}}^{i'}\otimes\bigotimes_{i'=i}^uC_{\mathrm{in}}^{i'}$, with
      \begin{align*}
        \lambda_{\mathrm{in}}^{(1)} &= \frac{d^u}{2} \\
        \lambda_{\mathrm{out}}^{(1)} &= \frac{8u^2n^4}{\rho(u,n,d)}\cdot\lambda_{\mathrm{run}} \\
        \lambda_{\mathrm{det}}^{(1)} &= \frac{\rho(u,n,d)}{8}\cdot\lambda_{\mathrm{out}} = u^2n^4\cdot\lambda_{\mathrm{run}}.
      \end{align*}
    \item The downwards switching gadget in \Cref{it:sodown} in \Cref{lem:switchone} to unencode from $C_{\mathrm{in}}^i$ in direction~$i$, with
      \begin{align*}
        \lambda_{\mathrm{in}}^{(2)} &= \lambda_{\mathrm{out}}^{(1)} \\
        \lambda_{\mathrm{out}}^{(2)} &= n\cdot\lambda_{\mathrm{in}}^{(2)}+3n^3\cdot\lambda_{\mathrm{run}} \leq \frac{11u^2n^5}{\rho(u,n,d)}\cdot\lambda_{\mathrm{run}}.
      \end{align*}
    \item\label{it:sadrenumber} A single timestep consisting either of $\gTerm$ gates on dits in $[n]^{[u]\setminus\{i\}}\times([k_{\mathrm{in},i}]\setminus[k_{\mathrm{out},i}])$ if $k_{\mathrm{in},i}\geq k_{\mathrm{out},i}$, or else of $\gInit$ gates on dits in $[n]^{[u]\setminus\{i\}}\times([k_{\mathrm{out},i}]\setminus[k_{\mathrm{in},i}])$ if $k_{\mathrm{out},i}>k_{\mathrm{in},i}$. (For these product sets $[n]^{[u]\setminus\{i\}}\times A$ we implicitly push the factor $A$ to the $i$th position.)
    \item The upwards switching gadget in \Cref{it:soup} in \Cref{lem:switchone} to encode into $C_{\mathrm{out}}^i$ in direction~$i$, with
      \begin{align*}
        \lambda_{\mathrm{in}}^{(4)} &= \lambda_{\mathrm{out}}^{(2)}+\lambda_{\mathrm{run}} \\
        \lambda_{\mathrm{out}}^{(4)} &= n^2\cdot\lambda_{\mathrm{in}}^{(4)}+3n^3\cdot\lambda_{\mathrm{run}} \leq \frac{15u^2n^7}{\rho(u,n,d)}\cdot\lambda_{\mathrm{run}}.
      \end{align*}
    \item The error-detecting gadget in \Cref{lem:errdet} for the code $\bigotimes_{i'=1}^iC_{\mathrm{out}}^{i'}\otimes\bigotimes_{i'=i+1}^uC_{\mathrm{in}}^{i'}$, with
      \begin{align*}
        \lambda_{\mathrm{in}}^{(5)} &= \lambda_{\mathrm{out}}^{(4)} < \frac{d^u}{2} \\
        \lambda_{\mathrm{out}}^{(5)} &= \frac{8u^2n^4}{\rho(u,n,d)}\cdot\lambda_{\mathrm{run}} \\
        \lambda_{\mathrm{det}}^{(1)} &= \frac{\rho(u,n,d)}{8}\cdot\lambda_{\mathrm{out}} = u^2n^4\cdot\lambda_{\mathrm{run}}.
      \end{align*}
    \end{enumerate}
    The upper bound on $\lambda_{\mathrm{out}}^{(4)}$ above follows from the definition of $\lambda_{\mathrm{run}}$ in \Cref{eq:sadlams}. Then because $\cR$ is simply the sequential composition of fault-detecting gadgets (which include fault-tolerant gadgets as a special case; see \Cref{fact:dettol}), and the first gadget in the sequence is mending by \Cref{lem:errdet}, \Cref{lem:seqcomp} implies that the gadget in \Cref{eq:sadgad} is mending fault-detecting for the function $\bar{O}$ given by the sequential composition of the functions implemented by the sequence of gadgets above, where $\lambda_{\mathrm{in}}=\lambda_{\mathrm{in}}^{(1)}$, $\lambda_{\mathrm{out}}=\lambda_{\mathrm{out}}^{(5)}$, and $\lambda_{\mathrm{det}}=\lambda_{\mathrm{det}}^{(1)}=\lambda_{\mathrm{det}}^{(5)}$. The set $E_{\mathrm{det}}\subseteq N'$ is simply the subset given in \Cref{lem:errdet}. Similarly, following \Cref{lem:errdet}, $T_{\mathrm{det}}\subseteq[T]$ is the set of all timesteps of $\cR$ in which we execute the second-to-last timestep of a call to the gadget in \Cref{lem:errdet}.

    \Cref{it:sadrenumber} is the only step above in which we do not apply a gadget that implements the logical identity function on all dits. Instead, across all $i=1,\dots,u$, \Cref{it:sadrenumber} terminates all logical dits in $K_{\mathrm{in}}\setminus K_{\mathrm{out}}$, and initializes all logical dits in $K_{\mathrm{out}}\setminus K_{\mathrm{in}}$ to $0$. Thus the gadget in \Cref{eq:sadgad} indeed implements the function $\bar{O}$ in \Cref{eq:safun}.

    By definition, each gadget used above to define $\cR$ uses gates in $\{\gInit,\gTerm,\gCX^*\}$. Furthermore, the gadget in \Cref{lem:errdet} uses space $\leq(u+1)n^u$ and time $\leq un^2+2$, while the gadgets in \Cref{lem:switchone} use space $\leq 2n^u$ and time $\leq n^2+2$. Thus by \Cref{lem:seqcomp},  $\cR$ uses space $|N'|\leq\max\{(u+1)n^u,2n^u\}=(u+1)n^u$ and time $T\leq u(2(un^2+2)+2(n^2+2)+1)\leq 16u^2n^2$, as desired.
    
  \item The proof is similar to that of the fault-detecting case above, except now we apply the error-correction gadget \Cref{lem:errcorr} instead of the error-detection gadget \Cref{lem:errdet}. Specifically, the desired circuit $\cR$ loops through $i=1,\dots,u$, and for each $u$ applies the following sequence of gadgets:
    \begin{enumerate}[label=(\arabic*)]
    \item The error-correcting gadget in \Cref{lem:errcorr} for the code $\bigotimes_{i'=1}^{i-1}C_{\mathrm{out}}^{i'}\otimes\bigotimes_{i'=i}^uC_{\mathrm{in}}^{i'}$, with
      \begin{align*}
        \lambda_{\mathrm{in}}^{(1)} &= \left(\frac{\lambda_{\Dec}}{2}\right)^u \\
        \lambda_{\mathrm{out}}^{(1)} &= \left(\frac{2^u+1}{\lambda_{\Dec}/2n}\right)^u\cdot u^2n\cdot T_{\Dec}\cdot\lambda_{\mathrm{run}}.
      \end{align*}
    \item The downwards switching gadget in \Cref{it:sodown} in \Cref{lem:switchone} to unencode from $C_{\mathrm{in}}^i$ in direction~$i$, with
      \begin{align*}
        \lambda_{\mathrm{in}}^{(2)} &= \lambda_{\mathrm{out}}^{(1)} \\
        \lambda_{\mathrm{out}}^{(2)} &= n\cdot\lambda_{\mathrm{in}}^{(2)}+3n^3\cdot\lambda_{\mathrm{run}} \leq 2\cdot\left(\frac{2^u+1}{\lambda_{\Dec}/2n}\right)^u\cdot u^2n^3\cdot T_{\Dec}\cdot\lambda_{\mathrm{run}}.
      \end{align*}
    \item\label{it:sacrenumber} A single timestep consisting either of $\gTerm$ gates on dits in $[n]^{[u]\setminus\{i\}}\times([k_{\mathrm{in},i}]\setminus[k_{\mathrm{out},i}])$ if $k_{\mathrm{in},i}\geq k_{\mathrm{out},i}$, or else of $\gInit$ gates on dits in $[n]^{[u]\setminus\{i\}}\times([k_{\mathrm{out},i}]\setminus[k_{\mathrm{in},i}])$ if $k_{\mathrm{out},i}>k_{\mathrm{in},i}$. (For these product sets $[n]^{[u]\setminus\{i\}}\times A$ we implicitly push the factor $A$ to the $i$th position.)
    \item The upwards switching gadget in \Cref{it:soup} in \Cref{lem:switchone} to encode into $C_{\mathrm{out}}^i$ in direction~$i$, with
      \begin{align*}
        \lambda_{\mathrm{in}}^{(4)} &= \lambda_{\mathrm{out}}^{(2)}+\lambda_{\mathrm{run}} \\
        \lambda_{\mathrm{out}}^{(4)} &= n^2\cdot\lambda_{\mathrm{in}}^{(4)}+3n^3\cdot\lambda_{\mathrm{run}} \leq 3\cdot\left(\frac{2^u+1}{\lambda_{\Dec}/2n}\right)^u\cdot u^2n^5\cdot T_{\Dec}\cdot\lambda_{\mathrm{run}}.
      \end{align*}
    \item The error-detecting gadget in \Cref{lem:errdet} for the code $\bigotimes_{i'=1}^iC_{\mathrm{out}}^{i'}\otimes\bigotimes_{i'=i+1}^uC_{\mathrm{in}}^{i'}$, with
      \begin{align*}
        \lambda_{\mathrm{in}}^{(5)} &= \lambda_{\mathrm{out}}^{(4)} < \left(\frac{\lambda_{\Dec}}{2}\right)^u \\
        \lambda_{\mathrm{out}}^{(5)} &= \left(\frac{2^u+1}{\lambda_{\Dec}/2n}\right)^u\cdot u^2n\cdot T_{\Dec}\cdot\lambda_{\mathrm{run}}.
      \end{align*}
    \end{enumerate}
    The upper bound on $\lambda_{\mathrm{out}}^{(4)}$ above follows from the definition of $\lambda_{\mathrm{run}}$ in \Cref{eq:saclams}. Then because $\cR$ is simply the sequential composition of fault-tolerant gadgets, \Cref{lem:seqcomp} implies that the gadget in \Cref{eq:sacgad} is fault-tolerant for the function $\bar{O}$ given by the sequential composition of the functions implemented by the sequence of gadgets above, where by \Cref{eq:saclams} we have $\lambda_{\mathrm{in}}=\lambda_{\mathrm{in}}^{(1)}$ and $\lambda_{\mathrm{out}}\leq\lambda_{\mathrm{out}}^{(5)}$.

    \Cref{it:sacrenumber} is the only step above in which we do not apply a gadget that implements the logical identity function on all dits. Instead, across all $i=1,\dots,u$, \Cref{it:sacrenumber} terminates all logical dits in $K_{\mathrm{in}}\setminus K_{\mathrm{out}}$, and initializes all logical dits in $K_{\mathrm{out}}\setminus K_{\mathrm{in}}$ to $0$. Thus the gadget in \Cref{eq:sacgad} indeed implements the function $\bar{O}$ in \Cref{eq:safun}.

    By definition, each gadget used above to define $\cR$ uses gates in $\{\gInit,\gTerm,\gCX^*\}\cup\cG_{\Dec}$. Furthermore, the gadget in \Cref{lem:errcorr} uses space $\leq n^{u-1}\cdot|N_{\Dec}|$ and time $\leq u\cdot T_{\Dec}$, while the gadgets in \Cref{lem:switchone} use space $\leq 2n^u$ and time $\leq n^2+2$. Thus by \Cref{lem:seqcomp},  $\cR$ uses space $|N'|\leq\max\{n^{u-1}\cdot|N_{\Dec}|,\; 2n^u\}$ and time $T\leq u(2uT_{\Dec}+2(n^2+2)+1)\leq 2u^2T_{\Dec}+8un^2$, as desired.
  \end{enumerate}
\end{proof}


\subsection{Transversal Gate Gadgets}
In this section, we present gadgets to fault-tolerantly perform the gates in \Cref{def:gates} on all dits in a codeword with constant space and time overhead. These gadgets will not rely on a tensor code structure, though our gadget for $\gCCX^*$ gates will require a certain multiplication property, which for instance is satisfied by Reed-Solomon codes.

We borrow the term ``transversal gate'' from the quantum error-correction literature to refer to the application of a gate uniformly across a block of dits:

\begin{definition}
  \label{def:transversal}
  Let $G:\bF_q^{m_{\mathrm{in}}}\rightarrow\bF_q^{m_{\mathrm{out}}}$ be a gate, and let $N$ be a set. The \emph{transversal application} $G^{\sqcup N}:\bF_q^{N\times[m_{\mathrm{in}}]}\rightarrow\bF_q^{N\times[m_{\mathrm{out}}]}$ of $G$ is the gate with input dits labeled by $N^{\sqcup m_{\mathrm{in}}}\cong N\times[m_{\mathrm{in}}]$ and output dits labeled by $N^{\sqcup m_{\mathrm{out}}}\cong N\times[m_{\mathrm{out}}]$, which simply applies $G$ to each set of dits sharing the same label in $N$.

  We extend this definition in the natural way to sets of input dits $N_1\cong\cdots\cong N_{m_{\mathrm{in}}}\cong N$ and output dits $N_1'\cong\cdots\cong N_{m_{\mathrm{out}}}'\cong N$ with fixed isomorphisms to $N$.
\end{definition}

While the term ``transversal gate'' sometimes includes more general constant-depth circuits involving a specified gate, in this paper we use the more restrictive notion in \Cref{def:transversal}.

Below, recall that we use `$*$' to denote component-wise multiplication of vectors.

\begin{lemma}
  \label{lem:transversal}
  Let $C$ be a $[n,k]_q$ code with encoding map $\Enc$. For $\alpha\in\{\mathrm{in},\mathrm{out}\}$, we let $D_\alpha=(C,\; \Enc,\; \cE_\alpha=2^{[n]}|_{\geq\lambda_\alpha})$ for $\lambda_\alpha$ defined below. Then for every $\lambda_{\mathrm{in}},\lambda_{\mathrm{run}}\geq 0$, the following hold:
  \begin{enumerate}
  \item Letting $\lambda_{\mathrm{out}}=\lambda_{\mathrm{run}}$, there exists a fault-tolerant gadget
    \begin{equation*}
      (\cR_{\gInit},\; \cE_{\mathrm{run}}=2^{[n]}|_{\geq\lambda_{\mathrm{run}}},\; \emptyset,\; D_{\mathrm{out}})
    \end{equation*}
    for $\bar{O}_{\gInit}=\gInit^{\sqcup k_{\mathrm{in}}}$ using space $N=[n]$, time $T=1$, and gate set $\{\gInit\}$.
  \item There exists a fault-tolerant gadget
    \begin{equation*}
      (\cR_{\gTerm},\; \cE_{\mathrm{run}}=\emptyset,\; D_{\mathrm{in}},\; \emptyset)
    \end{equation*}
    for $\bar{O}_{\gTerm}=\gTerm^{\sqcup k_{\mathrm{in}}}$ using space $N=[n]$, time $T=1$, and gate set $\{\gTerm\}$.
  \item\label{it:tX} Letting $\lambda_{\mathrm{out}}=\lambda_{\mathrm{in}}+\lambda_{\mathrm{run}}$, for every $a\in\bF_q^k$, there exists a fault-tolerant gadget
    \begin{equation*}
      (\cR_{\gX^a},\; \cE_{\mathrm{run}}=2^{[n]}|_{\geq\lambda_{\mathrm{run}}},\; D_{\mathrm{in}},\; D_{\mathrm{out}})
    \end{equation*}
    for $\bar{O}_{\gX^a}=\gX^a:=\bigsqcup_{i\in[k]}\gX^{a_i}$ using space $N=[n]$, time $T=1$, and gate set $\{\gX^*\}$.
  \item\label{it:tCX} Letting $\lambda_{\mathrm{out}}=2\lambda_{\mathrm{in}}+\lambda_{\mathrm{run}}$, for every $a\in\bF_q$, there exists a fault-tolerant gadget
    \begin{equation*}
      (\cR_{\gCX^a},\; \cE_{\mathrm{run}}=2^{[n]^{\sqcup 2}}|_{\geq\lambda_{\mathrm{run}}},\; D_{\mathrm{in}}^{\sqcup 2},\; D_{\mathrm{out}}^{\sqcup 2})
    \end{equation*}
    for $\bar{O}_{\gCX^a}=(\gCX^a)^{\sqcup k}$ using space $N=[n]^{\sqcup 2}$, time $T=1$, and gate set $\{\gCX^a\}$.
  \item\label{it:tCCX} Let $\lambda_{\mathrm{out}}=3\lambda_{\mathrm{in}}+\lambda_{\mathrm{run}}$. Let $C'\subseteq\bF_q^n$ be a $[n,\; k'\geq k]_q$ code with encoding map $\Enc'$ such that $C*C\subseteq C'$, and such that for every $x_1,x_2\in\bF_q^k$, there exists $x'\in\bF_q^{k'-k}$ such that
    \begin{equation}
      \label{eq:tCCXxp}
      \Enc(x_1)*\Enc(x_2) = \Enc'(x_1*x_2,\; x').
    \end{equation}
    For $\alpha\in\{\mathrm{in},\mathrm{out}\}$, let $D_\alpha'=(C',\; \Enc',\; \cE_\alpha'=2^{[n]}|_{\geq\lambda_\alpha})$. Then for every $a\in\bF_q$, there exists a fault-tolerant gadget
    \begin{equation*}
      (\cR_{\gCCX^a},\; \cE_{\mathrm{run}}=2^{[n]^{\sqcup 3}}|_{\geq\lambda_{\mathrm{run}}},\; D_{\mathrm{in}}\sqcup D_{\mathrm{in}}\sqcup D_{\mathrm{in}}',\; D_{\mathrm{out}}\sqcup D_{\mathrm{out}}\sqcup D_{\mathrm{out}}')
    \end{equation*}
    for the set $\bar{\cO}_{\gCCX^a}$ of functions $\bar{O}:\bF_q^{k+k+k'}\rightarrow\bF_q^{k+k+k'}$ that first apply $(\gCCX^a)^{\sqcup k}$ to the input, and then add $(0^k,0^k,0^k,x')$ to the result for some $x'\in\bF_q^{k'-k}$. That is, $\bar{O}(x_1,x_2,x_3)=(x_1,x_2,x_3+ax_1x_2+x')$, where every $x'\in\bF_q^{k'-k}$ yields a distinct $\bar{O}\in\bar{\cO}$. This gadget uses space $N=[n]^{\sqcup 3}$, time $T=1$, and gate set $\{\gCCX^a\}$.
  \end{enumerate}
\end{lemma}
\begin{proof}
  \begin{enumerate}
  \item The result follows immediately by letting $\cR_{\gInit}=\gInit^{\sqcup n}$.
  \item The result follows immediately by letting $\cR_{\gTerm}=\gTerm^{\sqcup n}$.
  \item We let $\cR_{\gX^a}=\gX^{\Enc(a)}$. On input $x\in\bF_q^k$, a noiseless execution outputs
    \begin{equation*}
      \cR_{\gX^a}(\Enc(x))=\Enc(x)+\Enc(a)=\Enc(\gX^a(x)).
    \end{equation*}
    As $\cR_{\gX^a}$ has a single timestep consisting of single-dit gates, if there are $<\lambda_{\mathrm{in}}$ input errors and $<\lambda_{\mathrm{run}}$ fault errors, then there are $<\lambda_{\mathrm{out}}=\lambda_{\mathrm{in}}+\lambda_{\mathrm{run}}$ output errors, so the desired fault-tolerance holds.
  \item We let $\cR_{\gCX^a}=(\gCX^a)^{\sqcup n}$. On input $(x_1,x_2)\in\bF_q^{k+k}$, a noiseless execution outputs
    \begin{align*}
      \cR_{\gCX^a}(\Enc^{\sqcup 2}(x_1,x_2))
      &= (\Enc(x_1),\Enc(x_2)+a\Enc(x_1)) \\
      &= \Enc^{\sqcup 2}(x_1,x_2+ax_1) \\
      &= \Enc^{\sqcup 2}(\gCX^a(x_1,x_2)).
    \end{align*}
    As $\cR_{\gCX^a}$ has a single timestep consisting of 2-dit gates, any input error on some dit can only propagate to one other dit. Therefore if each of $\Enc(x_1),\Enc(x_2)$ has $<\lambda_{\mathrm{in}}$ input errors and there are $<\lambda_{\mathrm{run}}$ fault errors, then there are $<\lambda_{\mathrm{out}}=2\lambda_{\mathrm{in}}+\lambda_{\mathrm{run}}$ output errors, so the desired fault-tolerance holds.
  \item We let $\cR_{\gCCX^a}=(\gCCX^a)^{\sqcup n}$. On input $(x_1,x_2,x_3)\in\bF_q^{k+k+k'}$, defining $x'\in\bF_q^{k'-k}$ as in \Cref{eq:tCCXxp}, then a noiseless execution outputs
    \begin{align*}
      \hspace{1em}&\hspace{-1em} \cR_{\gCCX^a}(\Enc(x_1),\Enc(x_2),\Enc'(x_3)) \\
      &= (\Enc(x_1),\Enc(x_2),\Enc'(x_3)+a\Enc(x_1)*\Enc(x_2)) \\
      &= (\Enc(x_1),\Enc(x_2),\Enc'(x_3)+a\Enc'(x_1*x_2,x')) \\
      &= (\Enc\sqcup\Enc\sqcup\Enc')(\gCCX^a(x_1,x_2,x_3)+(0^k,0^k,0^k,ax')) \\
      &\in (\Enc\sqcup\Enc\sqcup\Enc')(\bar{\cO}_{\gCCX^a}(x_1,x_2,x_3)).
    \end{align*}
    As $\cR_{\gCX^a}$ has a single timestep consisting of 3-dit gates, any input error on some dit can only propagate to two other dits. Therefore if each of $\Enc(x_1),\Enc(x_2),\Enc'(x_3)$ has $<\lambda_{\mathrm{in}}$ input errors and there are $<\lambda_{\mathrm{run}}$ fault errors, then there are $<\lambda_{\mathrm{out}}=3\lambda_{\mathrm{in}}+\lambda_{\mathrm{run}}$ output errors, so the desired fault-tolerance holds.
  \end{enumerate}
\end{proof}

In \Cref{lem:transversal}, strictly speaking all of the gadgets except the $\gX^a$ gadget (\Cref{it:tX}) implement transversal gates according to \Cref{def:transversal}, as the $\gX^a$ gadget applies different gates $\gX^{a_i}$ to different dits $i$.

\Cref{it:tCCX} in \Cref{lem:transversal} can be generalized to allow for two distinct codes $C_1,C_2$ such that $C_1*C_2\subseteq C'$. However, it will be sufficient for our purposes to consider $C_1=C=C_2$ as in \Cref{lem:transversal}. The reason that \Cref{it:tCCX} needs a set $\bar{\cO}_{\gCCX^a}$ of functions, whereas the other gadgets impelement a single function, is that $C'$ may have larger dimension than $C$. Hence when multiplying two codewords $\Enc(x_1),\Enc(x_2)\in C$, the product is a codeword $\Enc(x_1)*\Enc(x_2)=\Enc'(x_1*x_2,x')\in C'$ with some additional message dits $x'\in\bF_q^{k'-k}$, which may have arbitrary value.

\section{Fault-Detecting and Fault-Tolerance Schemes}
\label{sec:fdft}
In this section, we combine our gadgets in \Cref{sec:gadgets} to present schemes for compiling any logical circuit into a fault-detecting or a fault-tolerant circuit that implements the same underlying function as the logical circuit. Our fault-detecting circuit will be resilient to more errors per timestep than our fault-tolerant circuit, which is not surprising as fault-tolerance is a stronger property than fault-detection. Indeed, the two schemes have the same structure, except that our fault-detecting scheme detects errors using the error-detection gadget in \Cref{lem:errdet}, whereas our fault-tolerant scheme corrects errors using the error-correction gadget in \Cref{lem:errcorr}.

Our schemes in this section are based on tensor products of Reed-Solomon codes, which have growing alphabet size. In \Cref{sec:instan} below, we show how to reduce the alphabet size to a constant.

Our main schemes are stated in the theorem below.

\begin{theorem}
  \label{thm:fdft}
  For every $u,n,k\in\bN$ with
  $u\geq 4$,
  $k\leq n/32$,
  every subset $K\subseteq[k]^u$, and every prime power $q\geq n$, there exists a $[n^u,\; |K|,\; \geq(n-8k+1)^u]_q$ code $C(q,u,n,k,K)$ with encoding map $\Enc_{C(q,u,n,k,K)}:\bF_q^K\rightarrow\bF_q^{[n]^u}$ satisfying the following properties.
  The $n^u$ codeword dits of $C(q,u,n,k,K)$ are labeled by $[n]^u$, and the $|K|=\dim(C(q,u,n,k,K))$ message dits are labeled by the set~$K$.
  
  For every circuit $\bar{\cR}$ acting on a set of $q$-ary dits labeled by $K_{[3]}:=K_1\sqcup K_2\sqcup K_3$ with each $K_b=[k]^u$, using time $\bar{T}$ and gate set $\cG=\{\gInit,\gTerm,\gX^*,\gCX^*,\gCCX^*\}$, with input and output dits $K_{\mathrm{in}},K_{\mathrm{out}}\subseteq K_{[3]}$ respectively, letting
  \begin{align}
    \label{eq:Cab}
    C_{\alpha,b} &= \begin{cases}
      C(q,u,n,k,K_\alpha\cap K_b),&K_\alpha\cap K_b\neq\emptyset\\
      \emptyset,&K_\alpha\cap K_b=\emptyset.
    \end{cases}
  \end{align}
  for $\alpha\in\{\mathrm{in},\mathrm{out}\}$ and $b\in[3]$, then the following hold:
  \begin{enumerate}
  \item\label{it:fd} (Mending fault-detecting gadget) For $d\geq 0$, let $\rho(u,n,d)$ be the soundness parameter defined in \Cref{lem:loctest}. Let $\lambda_{\mathrm{run}}\in[0,\bar{\lambda}_{\mathrm{run}}]$ for
    \begin{align}
      \label{eq:fdbarlamrun}
      \bar{\lambda}_{\mathrm{run}} &= \bar{\lambda}_{\mathrm{run}}(u,n) = \frac{\rho(u,n,n-16k+2)\cdot(n-16k+2)^{u-1}}{2^{10}\cdot u^2n^8},
    \end{align}
    and let
    \begin{align}
      \label{eq:fdlams}
      \begin{split}
        \lambda_{\mathrm{in}} &= \frac{(n-8k+1)^u}{4} \\
        \lambda_{\mathrm{out}} &= \frac{16u^2n^4}{\rho(u,n,n-8k+1)}\cdot\lambda_{\mathrm{run}} \\
        \lambda_{\mathrm{det}} &= (u-1)^2n^4\cdot\lambda_{\mathrm{run}}.
      \end{split}
    \end{align}
    For $\alpha\in\{\mathrm{in},\mathrm{out}\}$, let
    \begin{align*}
      D_{\alpha,b} &= (C_{\alpha,b},\; \Enc_{C_{\alpha,b}},\; 2^{[n]^u}|_{\geq\lambda_\alpha}) \hspace{1em} \forall b\in[3] \\
      D_\alpha &= (C_\alpha,\; \Enc_\alpha,\; \cE_\alpha) = \bigsqcup_{b\in[3]}D_{\alpha,b}.
    \end{align*}
    Then there exists $T\in\bN$ and a set $N'$ with subsets $E_{\mathrm{det},t}\subseteq N'$ for $t\in[T]$ such that there is a mending fault-detecting gadget
    \begin{equation}
      \label{eq:fdgad}
      \left(\cR,\; \cE_{\mathrm{run}}=2^{N'}|_{\geq\lambda_{\mathrm{run}}}^{\sqcup T},\; D_{\mathrm{in}},\; D_{\mathrm{out}},\; \cE_{\mathrm{det}}=\bigsqcup_{t\in[T]}(2^{E_{\mathrm{det},t}}|_{\geq\lambda_{\mathrm{det}}})\right)
    \end{equation}
    for $\bar{O}(\cdot)=\bar{\cR}(\cdot)$ using space $|N'|\leq 3(u+1)n^u$, time $T\leq O(\bar{T}\cdot u^4n^2(\log n)^2)$, and gate set $\cG$.
  \item\label{it:ft} (Fault-tolerant gadget) There exists an absolute constant $\eta_{\Dec}\geq 1$ such that the following holds. Let $\lambda_{\mathrm{run}}\in[0,\bar{\lambda}_{\mathrm{run}}]$ for
    \begin{align}
      \label{eq:ftbarlamrun}
      \bar{\lambda}_{\mathrm{run}} &= \bar{\lambda}_{\mathrm{run}}(u,n) = \frac{n^{u-16}}{\eta_{\Dec}\cdot\log(q)\cdot 2^{u(u+16)}},
    \end{align}
    and let
    \begin{align}
      \label{eq:ftlams}
      \begin{split}
        \lambda_{\mathrm{in}} &= \left(\frac{n-8k+1}{8}\right)^u \\
        \lambda_{\mathrm{out}} &= \eta_{\Dec}\cdot 2^{u(u+8)}\cdot n^3\log(q)\cdot\lambda_{\mathrm{run}}.
      \end{split}
    \end{align}
    For $\alpha\in\{\mathrm{in},\mathrm{out}\}$, let
    \begin{align*}
      D_{\alpha,b} &= (C_{\alpha,b},\; \Enc_{C_{\alpha,b}},\; 2^{[n]^u}|_{\geq\lambda_\alpha}) \hspace{1em} \forall b\in[3] \\
      D_\alpha &= \bigsqcup_{b\in[3]}D_{\alpha,b}.
    \end{align*}
    Then there exists a fault-tolerant gadget
    \begin{equation}
      \label{eq:ftgad}
      \left(\cR,\; \cE_{\mathrm{run}}=2^{N'}|_{\geq\lambda_{\mathrm{run}}}^{\sqcup T},\; D_{\mathrm{in}},\; D_{\mathrm{out}}\right)
    \end{equation}
    for $\bar{O}(\cdot)=\bar{\cR}(\cdot)$ using space $|N'|\leq 3\eta_{\Dec}\cdot\log(q)\cdot n^{u+1}$, time $T\leq O(\bar{T}\cdot u^4n^2(\log n)^2(\log q))$, and gate set $\cG$.
  \end{enumerate}
\end{theorem}

In \Cref{thm:fdft}, the choice to allow up to $3$ input and output code blocks was somewhat arbitrary, to highlight our ability to perform $3$-dit gates such as $\gCCX^*$ on logical dits across different code blocks. It is useful to be able to act on at least $2$ code blocks, so that we can perform logical computations on data encoded in different code blocks, and move message dits between different code blocks. By composing any scheme capable of acting on $2$ more more code blocks with itself, we can fault-tolerantly accumulate message dits from any number of code blocks into a single code block, within which we can then perform a fault-tolerant computation.

\subsection{Logical Circuit Compilation}
\label{sec:compile}
In this section, we show how to compile an arbitrary logical circuit $\bar{\cR}$ into one whose structure is amenable to our fault-tolerant gadgets in \Cref{sec:gadgets}, with only small blowups in the space and time usage. In particular, in light of the fault-tolerant gadgets for transversal gates in \Cref{lem:transversal}, our compiled circuit will consist of such transversal gates (see \Cref{def:transversal}). For this purpose, we introduce the following notation.


\begin{definition}
  For a set $K=\prod_{i\in[u]}[k_i]$, for every $i\in[u]$ and $j\in[k_i]$ we define a subset $K|_{i\rightarrow j}\subseteq K$ by $K|_{i\rightarrow j}=\{\kappa\in K:\kappa_i=j\}$.
\end{definition}

\begin{lemma}
  \label{lem:compile}
  For a prime power $q$ and for $u,k\in\bN$ with $u\geq 4$, let $\bar{\cR}$ be a circuit acting on $3k^u$ $q$-ary dits labeled by the set $K_{[3]}=K_1\sqcup K_2\sqcup K_3$ with each $K_b=[k]^u$, using time $\bar{T}$ and gate set $\cG=\{\gInit,\gTerm,\gX^*,\gCX^*,\gCCX^*\}$.
  For subsets $B_{\mathrm{in}},B_{\mathrm{out}}\subseteq[3]$, we assume $\bar{\cR}$ has input dits $K_{B_{\mathrm{in}}}=\bigsqcup_{b\in B_{\mathrm{in}}}K_b$ and output dits $K_{B_{\mathrm{out}}}=\bigsqcup_{b\in B_{\mathrm{out}}}K_b$.
  Letting $k'=8k$, then there exists a circuit $\bar{\cR}'$ acting on $3\cdot{k'}^u$ dits labeled by $K'_{[3]}=K_1'\sqcup K_2'\sqcup K_3'$ with each $K_b'=[k']^u\supseteq[k]^u=K_b$ and with input dits $K'_{B_{\mathrm{in}}}=\bigsqcup_{b\in B_{\mathrm{in}}}K'_b$ and output dits $K'_{B_{\mathrm{out}}}=\bigsqcup_{b\in B_{\mathrm{out}}}K'_b$, using time $\bar{T}'\leq O(\bar{T}\cdot(u\log k)^2)$ and gate set $\cG$, with associated function $\bar{\cR}'(\cdot):\bF_q^{K'_{B_{\mathrm{in}}}}\rightarrow\bF_q^{K'_{B_{\mathrm{out}}}}$ that maps every $x\in\bF_q^{K'_{B_{\mathrm{in}}}}$ to
  \begin{equation*}
    \bar{\cR}'(x) = (\bar{\cR}(x|_{K_{B_{\mathrm{in}}}}),\; 0^{K'_{B_{\mathrm{out}}}\setminus K_{B_{\mathrm{out}}}}) \in \bF_q^{K'_{B_{\mathrm{out}}}}.
  \end{equation*}
  Furthermore, for every $t\in[\bar{T}']$ , then either:
  \begin{enumerate}
  \item\label{it:comX} $\bar{R}_t'$ consists entirely of $\gX^*$ gates, or
  \item\label{it:comnotX} For some $i\in[u]$ and some $G\in\{\gInit,\gTerm,\gCX^*,\gCCX^*\}$, $\bar{R}_t'$ consists entirely of transversal applications of $G$ to sets $K'_b|_{i\rightarrow j}$ for arbitrary values of $b\in[3]$ and $j\in[k']$.
  \end{enumerate}
\end{lemma}

To illustrate the structure of $\bar{\cR}'$ in \Cref{lem:compile}, consider some $t\in[\bar{T}']$ and $a\in\bF_q$ such that $G=\gCX^a$ in \Cref{it:comnotX} in \Cref{lem:compile}. Then there exists $i\in[u]$ such that the only non-identity gates applied in timestep $t$ of $\bar{\cR}'$ consist of transversal applications (in the sense of \Cref{def:transversal}) of $G=\gCX^a$ to pairs of sets $K'_{b_1}|_{i\rightarrow j_1},K'_{b_2}|_{i\rightarrow j_2}$ for $(b_1,j_1)\neq(b_2,j_2)\in[3]\times[k']$. That is, $G=\gCX^a$ is applied to every $\kappa_1\in K'_{b_1}|_{i\rightarrow j_1},\; \kappa_2\in K'_{b_2}|_{i\rightarrow j_2}$ such that $\kappa_1\in K'_{b_1}\cong[k']^u$ and $\kappa_2\in K'_{b_2}\cong[k']^u$ agree at all coordinates in $[u]\setminus\{i\}$ when viewed as tuples in $[k']^u$.

To prove \Cref{lem:compile}, we will use the following known result on efficiently implementing permutations in circuits with a hypercubic connectivity structure. Below, we define the gate $\gSwap:\bF_q^2\rightarrow\bF_q^2$ in the ordinary way $\gSwap(x_1,x_2)=(x_2,x_1)$. Recall that $\gSwap$ can for instance be implemented using one ancilla dit along with $6$ $\gCX^*$ gates, one $\gInit$ gate, and one $\gTerm$ gate.\footnote{If $\bF_q$ has characteristic $2$, there is a simpler implementation with $3$ $\gCX^*$ gates and no ancillas.}

\begin{lemma}[\cite{batcher_sorting_1968}]
  \label{lem:route}
  For $s\in\bN$, let $S=[2]^s$, and let $\pi:S\rightarrow S$ be a permutation, which we also view as a function $\pi:\bF_q^S\rightarrow\bF_q^S$ in the natural way. Then there exists a circuit $\cR$ acting on dits labeled by $S$ using time $T\leq O(s^2)$  
  and gate set $\{\gSwap\}$, such that $\cR(\cdot)=\pi(\cdot)$. Furthermore, for every time step $t\in[T]$, there exists some $i(t)\in[s]$ such that every $\gSwap$ gate in $R_t$ acts on a pair of dits whose labels differ only in coordinate $i(t)$.
\end{lemma}

\Cref{lem:route} allows us to implement arbitrary logical permutations via transversal controlled-swap gates, which in turn can be implemented using transversal $\gCCX$ gates.
Therefore to prove \Cref{lem:compile}, we will use \Cref{lem:route} to move dits into appropriate positions, such that the desired gates can then be implemented transversally. We provide the details below; we will also use the following basic fact.

\begin{fact}
  \label{fact:CCXnocoeff}
  For every $a\in\bF_q$, the gate $\gCX^a$ can be implemented using one ancilla dit along with one $\gCCX^{+1}$ gate, one $\gX^*$ gate, one $\gInit$ gate, and one $\gTerm$ gate.

  Furthermore, the gate $\gCCX^a$ can be implemented using two ancilla dits along with two $\gCCX^{+1}$ gates, one $\gX^*$ gate, two $\gInit$ gates, and two $\gTerm$ gates.
\end{fact}
\begin{proof}
  We implement $\gCX^a$ by performing
  \begin{align*}
    (x_1,x_2)
    &\xrightarrow{\gInit} (x_1,x_2,0) \xrightarrow{\gX^a} (x_1,x_2,a) \\
    &\xrightarrow{\gCCX^{+1}} (x_1,x_2+ax_1,a) \xrightarrow{\gTerm} (x_1,x_2+ax_1).
  \end{align*}
  We implement $\gCCX^a$ by performing
  \begin{align*}
    (x_1,x_2,x_3)
    &\xrightarrow{\gInit^{\sqcup 2}} (x_1,x_2,x_3,0,0) \xrightarrow{\gX^a} (x_1,x_2,x_3,0,a) \xrightarrow{\gCCX^{+1}} (x_1,x_2,x_3,x_1x_2,a) \\
    &\xrightarrow{\gCCX^{+1}} (x_1,x_2,x_3+ax_1x_2,x_1x_2,a) \xrightarrow{\gTerm^{\sqcup 2}} (x_1,x_2,x_3+ax_1x_2).
  \end{align*}
\end{proof}

\begin{proof}[Proof of \Cref{lem:compile}]
  The circuit $\bar{\cR}'$ will consist of the sequential composition of three parts
  \begin{equation}
    \label{eq:bRpstructure}
    \bar{\cR}' = \bar{\cR}'_{\mathrm{end}} \circ \bar{\cR}'_{\mathrm{mid}} \circ \bar{\cR}'_{\mathrm{start}},
  \end{equation}
  which we now define.
  $\bar{\cR}'_{\mathrm{start}}$ initializes all dits in $K'_{[3]}\setminus K_{B_{\mathrm{in}}}$ to $0$, and then performs transversal $\gSwap$ gates to move all dits in $K_{B_{\mathrm{in}}}$ into $K'_1$. Specifically, for each $b\in[3]\setminus K_{B_{\mathrm{in}}}$ then $\bar{\cR}'_{\mathrm{start}}$ applies $\gInit$ to all dits in $K'_b$, while for each $b\in K_{B_{\mathrm{in}}}$, each $i\in[u]$, and each $j\in[k']\setminus[k]$ then $\bar{\cR}'_{\mathrm{start}}$ applies $\gTerm$ followed by $\gInit$ to all dits in $K'_b|_{i\rightarrow j}$. Next, for each $j\in[k]$, $\bar{\cR}'_{\mathrm{start}}$ swaps the dits in $K'_2|_{1\rightarrow j}$ with the dits in $K'_1|_{1\rightarrow j+k}$, where the swaps are implemented using $\gCX^*$ gates (see above) with $K'_2|_{1\rightarrow j+k}$ as ancillas. Similarly, for each $j\in[k]$, $\bar{\cR}'_{\mathrm{start}}$ swaps the dits in $K'_3|_{2\rightarrow j}$ with the dits in $K'_1|_{2\rightarrow j+k}$, where the swaps are implemented using $\gCX^*$ gates (see above) with $K'_3|_{2\rightarrow j+k}$ as ancillas. This entire procedure takes time $2u+O(1)=O(u)$, as in each step we can parallelize across all $j$.

  Therefore $\bar{\cR}'_{\mathrm{start}}$ swaps the input dits originally in $K_1,K_2,K_3$ into the respective blocks $K_{1,1},K_{1,2},K_{1,3}\subseteq K'_1$ defined by $K_{1,1}=[k]^u$, $K_{1,2}=\{k+1,\dots,2k\}\times[k]^{u-1}$, $K_{1,3}=[k]\times\{k+1,\dots,2k\}\times[k]^{u-2}$. We let $K_{1,[3]}=K_{1,1}\sqcup K_{1,2}\sqcup K_{1,3}\subseteq K'_1$. The circuit $\bar{\cR}'_{\mathrm{end}}$ will invert this operation, by performing the same swaps described above to swap the dits in $K_{1,[3]}$ back to their original positions $K_{[3]}$. Then for each $b\in [3]\setminus B_{\mathrm{out}}$, $\bar{\cR}'_{\mathrm{end}}$ applies $\gTerm$ to all dits in $K'_b$, while for each $b\in B_{\mathrm{out}}$, each $i\in[u]$, and each $j\in[k']\setminus[k]$, $\bar{\cR}'_{\mathrm{end}}$ applies $\gTerm$ followed by $\gInit$ to all dits in $K'_b|_{i\rightarrow j}$. This procedure for $\bar{\cR}'_{\mathrm{end}}$ takes time $O(u)$, similarly as for $\bar{\cR}'_{\mathrm{start}}$ above.

  Therefore $\bar{\cR}'$ as defined in \Cref{eq:bRpstructure} with $\bar{\cR}'_{\mathrm{start}},\bar{\cR}'_{\mathrm{end}}$ defined above has the desired set of input and output dits, and outputs value $0$ on all dits in $K'_{B_{\mathrm{out}}}\setminus K_{B_{\mathrm{out}}}$. For $\alpha\in\{\mathrm{in},\mathrm{out}\}$, let $K_{1,B_\alpha}=\bigsqcup_{b\in B_\alpha}K_{1,b}\cong K_{B_\alpha}$. It only remains to construct $\bar{\cR}'_{\mathrm{mid}}$ such that for every $x\in\bF_q^{K_{1,B_{\mathrm{in}}}}$, then
  \begin{equation}
    \label{eq:Rpmid}
    \bar{\cR}'_{\mathrm{mid}}(x,0^{K'_{[3]}\setminus K_{1,B_{\mathrm{in}}}})|_{K_{1,B_{\mathrm{out}}}}=\bar{\cR}(x)
  \end{equation}
  under the isomorphisms $\bF_q^{K_{B_\alpha}}\cong\bF_q^{K_{1,B_\alpha}}$.


  Let $s=\lceil\log_2 k\rceil+2$, and fix an isomorphism $[2^s]\cong[2]^s$, so that $[2^s]^u\cong[2]^{su}$. Note that $4k\leq 2^s\leq 8k=k'$. We let $[2^s]^u_1\subseteq K'_1$ denote the dits in $K'_1$ labeled by $[2^s]^u\subseteq[k']^u$
  
  We begin by applying \Cref{lem:route} to show \Cref{claim:comperm} below.
  In the remainder of this proof below, we say a circuit acting on dits $K'_{[3]}$ (or on a subset thereof) has \emph{transversal structure} if each timestep of the circuit satisfies \Cref{it:comX} or \Cref{it:comnotX} in the statement of \Cref{lem:compile}.

  \begin{claim}
    \label{claim:comperm}
    For every permutation $\pi:[2^s]^u\times[2^s]^u$, there exists a circuit using space $K'_{[3]}$ and gate set $\cG$ with transversal structure that permutes the dits in $[2^s]^u_1\subseteq K'_1$ by $\pi$ in time $O(su)^2$. Dits in $K'_{[3]}\setminus[2^s]^u_1$ may be acted upon arbitrarily.
  \end{claim}
  \begin{proof}
    \Cref{lem:route} gives a circuit on dits $[2^s]^u_1\cong[2]^{su}_1$ that implements the permutation $\pi$ using time $O(su)^2$ and gate set $\{\gSwap\}$, where each timestep $t$ applies a subset of the $\gSwap$ gates performed in transversal applications to sets $K'_1|_{i\rightarrow j}$ for some $i=i(t)\in[u]$. That is, for various disjoint pairs of sets $K'_1|_{i\rightarrow j_1},K'_1|_{i\rightarrow j_2}\subseteq K'_1$, this circuit from \Cref{lem:route} applies $\gSwap$ to some (but maybe not all) pairs of dits $\kappa_1\in K'_1|_{i\rightarrow j_1},\; \kappa_2\in K'_1|_{i\rightarrow j_2}$ whose labels in $[k']^u$ agree on coordinates $[u]\setminus\{i\}$. For each such $j_1,j_2$, we may instead initialize a block $K'_2|_{i\rightarrow j_3}\subseteq K'_2$ to contain the indicator vector for the $\gSwap$ gates we want to perform on $K'_1|_{i\rightarrow j_1},K'_2|_{i\rightarrow j_2}$; we can in particular set $j_3=j_1$. This initialization can be done with $\gTerm$ followed by $\gInit$ gates to set dits in $K'_2|_{i\rightarrow j_3}$ to value $0$, and then by a layer of $\gX^*$ gates to give the desired values in $\{0,1\}$. Then the desired $\gSwap$ gates can be implemented by permforming transversal controlled-$\gSwap$ across the three blocks $K'_2|_{i\rightarrow j_3},K'_1|_{i\rightarrow j_1},K'_1|_{i\rightarrow j_2}$. Just as $\gSwap$ can be implemented using an ancilla dit initialized to $0$ along with $6$ $\gCX^*$ gates, controlled-$\gSwap$ can be implemented using an ancilla dit initializd to $0$ along with $6$ $\gCCX^*$ gates. Therefore we initialize an additional block $K'_2|_{i\rightarrow j_4}\subseteq K'_2$ for these additional ancilla dits, and then we implement one layer of the desired $\gSwap$ gates using gates $\{\gInit,\gTerm,\gX^*,\gCCX^*\}\subseteq\cG$ in a circuit with transversal structure; we can in particular set $j_4=j_2$. Furthermore, by construction this circuit acts on dits $K'_1\sqcup K'_2$ and runs in time $O(1)$. It follows that the entire permutation circuit, which consists of $O(su)^2$ layers of $\gSwap$ gates, can be implemented in a circuit with transversal structure using gate set $\cG$ in time $O(su)^2$, as desired.
  \end{proof}



  We thus construct $\bar{\cR}'_{\mathrm{mid}}$ from $\bar{\cR}$ as follows. First, by increasing the running time by at most a factor of~$5$, we can assume that every timestep in $\bar{\cR}$ consists entirely of gates from one of the sets $\{\gInit\}$, $\{\gTerm\}$, $\{\gX^*\}$, $\{\gCX^*\}$, or $\{\gCCX^*\}$. We now compile each such timestep into a sub-circuit of $\bar{\cR}'_{\mathrm{mid}}$ of the desired transversal structure, using time $O(su)^2$. Note that instead of terminating and then initializing a dit in $K_{[3]}$ that experiences a $\gTerm$ and then later a $\gInit$ gate in $\bar{\cR}$, we will simply set this dit's value to $0$, but leave it active in $\bar{\cR}'_{\mathrm{mid}}$. Specifically, we consider timesteps of $\bar{\cR}$ with each of the five types of gates separately:
  \begin{enumerate}
  \item $\gX^*$: Simply apply the desired $\gX^*$ gates to dits in $K_{1,[3]}\cong K_{[3]}$.
  \item $\gTerm$ on dits $A\subseteq K_{1,[3]}\cong K_{[3]}$: Run $\gTerm$ followed by $\gInit$ on $K'_1|_{3\rightarrow j+k}$ for every $j\in[k]$. (Recall here that $K'_1|_{3\rightarrow j+k}$ exists because $u\geq 4\geq 3$.) Use \Cref{claim:comperm} to swap some $|A|$ of the $\geq k(2^s)^{u-1}$ dits in $[2^s]^u_1\cap\bigsqcup_{j\in[k]}K'_1|_{3\rightarrow j+k}$ with the dits $A\subseteq K_{1,[3]}\cong K_{[3]}$, while preserving the positions of the other dits in $K_{1,[3]}$. (Recall here that $k(2^s)^{u-1}\geq 3k^u=|K_{1,[3]}|$ because $2^s\geq 4k$ and $u\geq 4$.) Then again run $\gTerm$ followed by $\gInit$ on $K'_1|_{3\rightarrow j+k}$ for every $j\in[k]$.
  \item $\gInit$ on dits $A\subseteq K_{[3]}$: apply the same procedure as for $\gTerm$ above.
  \item $\gCX^*$ on various pairs of dits in $K_{1,[3]}\cong K_{[3]}$: First, apply \Cref{claim:comperm} to permute the dits in $[2^s]^u_1\cap\bigsqcup_{j\in[2k]}K'_1|_{3\rightarrow j}$ such that every $\gCX^*$ gate now has control dit in some $K'_1|_{3\rightarrow j}$ for $j\in[k]$, and target dit in the respective position in $K'_1|_{3\rightarrow j+k}$. For dits acted on the by identity, we simply assign their associated gate to be $\gCX^0=I_2$. We then implement the desired $\gCX^*$ gates using transversal gates in $\cG$ on blocks $K'_b|_{3\rightarrow j}$. Specifically, we implement each gate $\gCX^a$ on a dit in $K'_1|_{3\rightarrow j}$ and a dit in $K'_1|_{3\rightarrow j+k}$ using the procedure in \Cref{fact:CCXnocoeff}, where $K'_2|_{3\rightarrow j}$ provides the ancilla dits, so that each timestep either performs $\gX^*$ gates, or else performs transversal $\gCCX^{+1}$, $\gInit$, or $\gTerm$ gates across different blocks $K'_b|_{3\rightarrow j'}$. Finally, we apply \Cref{claim:comperm} again to permute the dits originally in $K_{1,[3]}$ back to their starting positions.
  \item $\gCCX^*$ on various triples of dits in $K_{[3]}$: First, apply \Cref{claim:comperm} to permute the dits in $[2^s]^u_1\cap\bigsqcup_{j\in[3k]}K'_1|_{3\rightarrow j}$ such that every $\gCCX^*$ gate now has first control dit in some $K'_1|_{3\rightarrow j}$ for $j\in[k]$, second control dit in the respective position in $K'_1|_{3\rightarrow j+k}$, and target dit in the respective position in $K'_1|_{3\rightarrow j+2k}$. For dits acted on the by identity, we simply assign their associated gate to be $\gCCX^0=I_3$. We then implement the desired $\gCCX^*$ gates using transversal gates in $\cG$ on blocks $K'_b|_{3\rightarrow j}$. Specifically, we implement each gate $\gCCX^a$ using the procedure in \Cref{fact:CCXnocoeff}, where $K'_2|_{3\rightarrow j},K'_2|_{3\rightarrow j+k}$ provide the ancilla dits, so that each timestep either performs $\gX^*$ gates, or else performs transversal $\gCCX^{+1}$, $\gInit$, or $\gTerm$ gates across different blocks $K'_b|_{3\rightarrow j'}$. Finally, we apply \Cref{claim:comperm} again to permute the dits originally in $K_{1,[3]}$ back to their starting positions.
  \end{enumerate}

  By construction, in each cases above, the resulting circuit acting on dits $K'_{[3]}$ has the desired transversal structure, and by \Cref{claim:comperm} uses time $O(su)^2$. For each timestep of the original circuit consisting of either $\gX^*$, $\gCX^*$, or $\gCCX^*$ gates, the above construction for $\bar{\cR}'_{\mathrm{mid}}$ by definition induces the same function as the original circuit on dits $K_{1,[3]}\cong K_{[3]}$. Meanwhile, for each timestep of the original circuit consisting of either $\gInit$ or $\gTerm$ gates on some set $A\subseteq K_{1,[3]}\cong K_{[3]}$ of dits, the above construction simply sets the values of all dits in $A$ to $0$ while preserving the values of all other dits; this change does not affect the overall function implemented by the entire circuit, as a dit that is terminated cannot be used until it is reinitialized to $0$. Therefore when we combine all timesteps, $\bar{\cR}'_{\mathrm{mid}}$ satisfies the desired condition \Cref{eq:Rpmid}.

  It follows from the construction above that $\bar{\cR}'_{\mathrm{mid}}$ uses time $O(\bar{T}\cdot(su)^2)$. We showed above that $\bar{\cR}'_{\mathrm{start}},\bar{\cR}'_{\mathrm{end}}$ both use time $O(u)$, and all three circuits use space $K'_{[3]}$. Thus $\bar{\cR}'$ uses time $O(\bar{T}\cdot(su)^2)=O(\bar{T}\cdot(u\log k)^2)$ and space $K'_{[3]}$, as desired.
\end{proof}

\subsection{Code Family}
\label{sec:fdcodes}
In this section, we describe the codes $C(q,u,n,k,K)$ and associated encoding maps $\Enc_{C(q,u,n,k,K)}$ that we use to prove \Cref{thm:fdft}. Specifically, defining $k'=8k$ as in \Cref{lem:compile} above, will construct $C(q,u,n,k,K)$ as a subcode of the $u$th tensor power $C_0^{\otimes u}$ of a $[n,k',d=n-k'+1]_q$ code $C_0$ that exhibits a \emph{multiplication property}, meaning that $C_0^{*2}=C_0*C_0=\spn\{c*c':c,c'\in C_0\}$ has good distance. Such codes $C_0$ are in turn given by Reed-Solomon codes, as described below (see \Cref{def:RS}). Note that our techniques do not rely on the distance being exactly $n-k'+1$ (i.e.~on the Singleton bound), though we assume such a distance for simplicity, as it is achieved by the Reed-Solomon codes we use.

\begin{fact}
  \label{fact:basecode}
  For every $k'\leq n\in\bN$ and every prime power $q\geq n$, there exists a $[n,k',d=n-k'+1]_q$ code $C_0=C_0(q,n,k')$ such that $C_0^{*2}$ has distance $n-2k'+2$. Furthermore, there exists a radius-$(n-k'+1)/2$ decoder for $C_0$, and a radius-$(n-2k'+2)/2$-decoder for $C_0^{*2}$, where both decoders use space $|N'|=O(n^2\log q)$, time $T=O(n^2\log q)$, and get set $\cG=\{\gInit,\gTerm,\gX^*,\gCX^*,\gCCX^*\}$ over $q$-ary dits.
\end{fact}
\begin{proof}
  For some subset $E\subseteq\bF_q$ of size $|E|=n$, simply take $C_0$ to be the Reed-Solomon code $\RS(q,k',E)$, so that by \Cref{fact:RSmult} the code $C_0^{*2}=\RS(q,2k'-1,E)$ has distance $n-2k'+1$. The desired decoders exist by \Cref{lem:RSdec}.
\end{proof}

For the code $C_0$ in \Cref{fact:basecode}, fix some set $K'_0\subseteq[n]$ of size $|K'_0|=k'$ giving a systematic encoding for $C_0$, as shown to exist in \Cref{fact:systematic}. That is, under a fixed isomorphism $[k']\cong K'_0$, which induces an isomorphism $\bF_q^{k'}\cong\bF_q^{K'_0}$, we have a well-defined systematic encoding map $\Enc_0:\bF_q^{k'}\cong\bF_q^{K'_0}\xrightarrow{\sim}C_0\subseteq\bF_q^n$, meaning that each $x\in\bF_q^{k'}\cong\bF_q^{K'_0}$ is mapped to the unique codeword $\Enc_0(x)\in C$ such that $\Enc_0(x)|_{K'_0}=x$.

Furthermore, because $C_0^{*2}|_{K'_0}=(C_0|_{K'_0})^{*2}=(\bF_q^{K'_0})^{*2}=\bF_q^{K'_0}$, \Cref{fact:systematic} implies that there exists a subset $K_0''\subseteq[n]^u\setminus K'_0$ such that $K'_0\sqcup K_0''$ gives a systematic encoding map $\Enc_0':\bF_q^{K'_0\sqcup K_0''}\xrightarrow{\sim}C_0^{*2}\subseteq\bF_q^n$ for $C_0^{*2}$. That is, for every $(x,x')\in\bF_q^{K'_0}\oplus\bF_q^{K_0''}$, we have $\Enc_0'(x,x')|_{K'_0\sqcup K_0''}=(x,x')$.

\begin{claim}
  \label{claim:tensormultprop}
  For every $I\subseteq[u]$ and every $x_1,x_2\in\bF_q^{{K'_0}^I}$, there exists some $x'\in\bF_q^{(K'_0\sqcup K_0'')^I\setminus {K'_0}^I}$ such that
  \begin{equation}
    \label{eq:tensormultprop}
    \Enc_0^{\otimes I}(x_1)*\Enc_0^{\otimes I}(x_2) = {\Enc_0'}^{\otimes I}(x_1*x_2,x').
  \end{equation}
  In other words, the codes $(C_0^{\otimes I},\Enc_0^{\otimes I}),\; ((C_0^{*2})^{\otimes I},{\Enc_0'}^{\otimes I})$ satisfy \Cref{eq:tCCXxp} in \Cref{it:tCCX} in \Cref{lem:transversal}.
\end{claim}
\begin{proof}
  By definition $\Enc_0^{\otimes I}(x_1)*\Enc_0^{\otimes I}(x_2)\in(C_0^{*2})^{\otimes I}$. Because tensor products of systematic encodings are in turn systematic encodings, we have $\Enc_0^{\otimes I}(x_1)*\Enc_0^{\otimes I}(x_2)|_{{K'_0}^I}=x_1*x_2\in\bF_q^{{K'_0}^I}$, so there exists some $x'\in\bF_q^{(K'_0\sqcup K_0'')^I\setminus {K'_0}^I}$ satisfying \Cref{eq:tensormultprop}.
\end{proof}


Recall that $k'=8k$. Let $K_0\subseteq K_0'$ be the subset of size $|K_0|=k$ corresponding to the set $[k]\subseteq[k']$ under the isomorphism $[k']\cong K_0'$, that is, $K_0\cong[k]\subseteq[k']\cong K_0'$.

For $K\subseteq[k]^u\cong K_0^u$, we then define $\Enc_{C(q,u,n,k,K)}:\bF_q^K\rightarrow\bF_q^{[n]^u}$ and $C(q,u,n,k,K)\subseteq\bF_q^{[n]^u}$ by
\begin{align*}
  \Enc_{C(q,u,n,k,K)} &:= \Enc_0^{\otimes u}|_K \\
  C(q,u,n,k,K) &:= \Enc_{C(q,u,n,k,K)}(\bF_q^K) = \{c\in C_0^{\otimes u}:c|_{{K'_0}^u\setminus K}=0\}.
\end{align*}
By definition, $\Enc_0^{\otimes u}$ and hence $\Enc_{C(q,u,n,k,K)}$ is also a systematic encoding map. That is, for $x\in\bF_q^K$ we have $\Enc_{C(q,u,n,k,K)}(x)|_K=x$.


\subsection{Fault-Detecting and Fault-Tolerant Gadgets}
In this section, we prove \Cref{thm:fdft} using the codes $C(q,u,n,k,K)$ defined in \Cref{sec:fdcodes}. The main idea is to compile the logical circuit $\bar{\cR}$ to a circuit $\bar{\cR}'$ with the transversal structure given by \Cref{lem:compile}, which we in turn render fault-tolerant using the gadgets in \Cref{sec:gadgets}.

Throughout this section, we define $q,u,n,k$ as in \Cref{thm:fdft}, we define $k'=8k$, $C_0=C_0(q,n,k')$, $\Enc_0$, $\Enc_0'$ as in \Cref{sec:fdcodes}, and for $K\subseteq[k]^u$ we define $C(q,u,n,k,K)$ as in \Cref{sec:fdcodes}. We also define $K'_{[3]}=K'_1\sqcup K'_2\sqcup K'_3$ with each $K'_b=[k']^u\supseteq[k]^u=K_b$ as in \Cref{lem:compile}. Furthermore, we let $N_{[3]}=N_1\sqcup N_2\sqcup N_3$ with each $N_b=[n]^u$. For a subset $B\subseteq[3]$, we define $K_B=\bigsqcup_{b\in B}K_b$, $K'_B=\bigsqcup_{b\in B}K'_b$, and $N_B=\bigsqcup_{b\in B}N_b$.

\begin{proof}[Proof of \Cref{thm:fdft}]
  By its definition in \Cref{sec:fdcodes}, the code $C(q,u,n,k,K)$ has $n^u$ codeword dits labeled by $[n]^u$, has $|K|=\dim(C(q,u,n,k,K))$ message dits labeled by $K$, and has distance $\geq(n-k'+1)^u=(n-8k+1)^u$. It therefore remains to prove the mending fault-detecting and fault-tolerant claims in the statement of \Cref{thm:fdft}.

  For this purpose, as in \Cref{thm:fdft}, fix a circuit $\bar{\cR}$ using space $K_{[3]}$, time $\bar{T}$, and gate set $\cG=\{\gInit,\gTerm,\gX^*,\gCX^*,\gCCX^*\}$, with input and output dits $K_{\mathrm{in}},K_{\mathrm{out}}\subseteq K_{[3]}$. For $\alpha\in\{\mathrm{in},\mathrm{out}\}$ and $b\in[3]$, define $C_{\alpha,b}$ as in \Cref{eq:Cab}, and let $B_\alpha=\{b\in[3]:K_\alpha\cap K_b\neq\emptyset\}$ be the set of $b\in[3]$ for which $C_{\alpha,b}\neq\emptyset$. 

  We first define a circuit $\bar{\cR}^0$ that is a slight modification of $\bar{\cR}$ in order to fit the structure needed to apply \Cref{lem:compile}. Specfically, $\bar{\cR}^0$ uses space $K_{[3]}$ and time $\bar{T}+2$, with input dits $K_{B_{\mathrm{in}}}\supseteq K_{\mathrm{in}}$ and output dits $K_{B_{\mathrm{out}}}\supseteq K_{\mathrm{out}}$. We let the first timestep of $\bar{\cR}^0$ apply $\gTerm$ to all dits in $K_{B_{\mathrm{in}}}\setminus K_{\mathrm{in}}$, and we let the last timestep of $\bar{\cR}^0$ apply $\gInit$ to all dits in $K_{B_{\mathrm{out}}}\setminus K_{\mathrm{out}}$; the remaining $\bar{T}$ timesteps of $\bar{\cR}^0$ simply apply $\bar{\cR}$. Therefore $\bar{\cR}^0(\cdot)$ implements the same function as $\bar{\cR}(\cdot)$, except with extra input dits in $K_{B_{\mathrm{in}}}\supseteq K_{\mathrm{in}}$ that are ignored, and extra output dits in $K_{B_{\mathrm{out}}}\setminus K_{\mathrm{out}}$ that will always be set to $0$.
  
  We can then apply \Cref{lem:compile} to $\bar{\cR}^0$ to obtain a circuit $\bar{\cR}'=(\bar{R}'_1,\dots,\bar{R}'_{\bar{T}'})$ using space $K'_{[3]}$ and time $\bar{T}'\leq O(\bar{T}\cdot(u\log k)^2)$, with input and output dits $K'_{B_{\mathrm{in}}},K'_{B_{\mathrm{out}}}$ respectively, such that $\bar{\cR}'$ has the transversal structure described in \Cref{lem:compile}, and such that for every $x\in\bF_q^{K'_{B_{\mathrm{in}}}}$ we have
  \begin{equation}
    \label{eq:Rpfun}
    \bar{\cR}'(x) = (\bar{\cR}^0(x|_{K_{B_{\mathrm{in}}}}),\; 0^{K'_{B_{\mathrm{out}}}\setminus K_{B_{\mathrm{out}}}}) = (\bar{\cR}(x|_{K_{\mathrm{in}}}),\; 0^{K'_{B_{\mathrm{out}}}\setminus K_{\mathrm{out}}}) \in \bF_q^{K'_{B_{\mathrm{out}}}}.
  \end{equation}

  We show \Cref{it:fd,it:ft} in \Cref{thm:fdft} separately below. In both cases, we construct the desired circuit $\cR$ by making each timestep of $\bar{\cR}'$ fault-detecting/fault-tolerant using \Cref{lem:transversal}, for which purpose we first switch to an associated code using \Cref{lem:switchone,lem:switchall}. The main difference between the two cases is in how we prevent error propagation: to prove \Cref{it:fd} we use \Cref{lem:errdet} to detect errors, whereas to prove \Cref{it:ft} we use \Cref{lem:errcorr} to correct errors.

  Recall from \Cref{sec:fdcodes} that for $b\in B_\alpha$, then $C_{\alpha,b}=C(q,u,n,k,K_\alpha\cap K_b)$ is simply the subcode of $C_0^{\otimes u}$ given by setting message dits $K'_b\setminus K_\alpha$ to $0$. As we will see below, this structure is perfectly aligned with the definition of the function $\bar{\cR}'(\cdot)$ (\Cref{eq:Rpfun} above), which simply applies $\bar{\cR}$ while setting dits in $K'_{B_\alpha}\setminus K_\alpha$ to $0$.

  \begin{enumerate}
  \item (Mending fault-detecting gadget) We define the desired circuit $\cR$ as follows. We let $\cR$ act on dits $N'=N_{[3]}\sqcup N_{\mathrm{anc}}$, where $N_{[3]}=([n]^u)^{\sqcup 3}$ as defined above, and $N_{\mathrm{anc}}$ is a set of ancilla dits that can be used by gadgets that we will invoke in the definition of $\cR$. By definition, $\cR$ has input dits $N_{B_{\mathrm{in}}}\subseteq N_{[3]}$ and output dits $N_{B_{\mathrm{out}}}\subseteq N_{[3]}$. We let the first timestep of $\cR$ apply $\gInit$ to all dits in $N_{[3]\setminus B_{\mathrm{in}}}$, and we let the last timestep of $\cR$ apply $\gTerm$ to all dits in $N_{[3]\setminus B_{\mathrm{out}}}$. To define the remainder of $\cR$, for each timestep $\bar{t}\in[\bar{T}']$ of $\bar{\cR}'$, we define an associated subcircuit $\cR_{\bar{t}}$ of $\cR$, acting on all dits $N'$ across $O(u^2n^2)$ timesteps, as follows. We then let $\cR$ be the sequential composition
    \begin{equation}
      \label{eq:fdR}
      \cR = \gTerm^{\sqcup N_{[3]\setminus B_{\mathrm{out}}}}\circ\cR_{\bar{T}'}\circ\cdots\circ\cR_1\circ\gInit^{\sqcup N_{[3]\setminus B_{\mathrm{in}}}}.
    \end{equation}

    By definition, $\bar{R}'_{\bar{t}}$ has the transversal structure described in \Cref{lem:compile}, meaning that there exists some $i\in[u]$ such that either $\bar{R}'_{\bar{t}}$ consists entirely of $\gX^*$ gates, or else entirely of transversal applications of some gate $G\in\{\gInit,\gTerm,\gCX^*,\gCCX^*\}$ to sets $K'_b|_{i\rightarrow j}$ for arbitrary values of $b\in[3]$ and $j\in[k]$. Therefore we let $\cR_{\bar{t}}$ perform the following:
    \begin{enumerate}[label=(\arabic*)]
    \item For every $b\in[3]$, run the error-detecting gadget in \Cref{lem:errdet} for the code $C_0^{\otimes u}$ on dits $N_b$ with
      \begin{align*}
        \lambda_{\mathrm{in}}^{(1)} &= \frac{(n-k'+1)^u}{2} \\
        \lambda_{\mathrm{out}}^{(1)} &= \frac{8u^2n^4}{\rho(u,n,n-k'+1)}\cdot\lambda_{\mathrm{run}} \\
        \lambda_{\mathrm{det}}^{(1)} &= \frac{\rho(u,n,n-k'+1)}{8}\cdot\lambda_{\mathrm{out}} = u^2n^4\cdot\lambda_{\mathrm{run}}.
      \end{align*}
    \item For every $b\in[3]$, run the downwards switching gadget in \Cref{it:sodown} in \Cref{lem:switchone} to unencode from $(C_0,\Enc_0)$ in direction~$i$ on dits $N_b$, with
      \begin{align*}
        \lambda_{\mathrm{in}}^{(2)} &= \lambda_{\mathrm{out}}^{(1)} \\
        \lambda_{\mathrm{out}}^{(2)} &= n\cdot\lambda_{\mathrm{in}}^{(2)}+3n^3\cdot\lambda_{\mathrm{run}} \leq \frac{11u^2n^5}{\rho(u,n,n-k'+1)}\cdot\lambda_{\mathrm{run}}.
      \end{align*}
      Specifically, each application of this gadget takes as input a codeblock of $C_0^{\otimes[u]}$ in dits $N_b$, and outputs a codeblock of $(C_0^{\otimes[u]\setminus\{i\}})^{\sqcup k'}$ in dits $\bigsqcup_{j\in[k']}N_b|_{i\rightarrow j}$.
    \item By construction $\bar{R}'_{\bar{t}}$ consists entirely of gates from one of the five sets $\{\gInit\}$, $\{\gTerm\}$, $\{\gX^*\}$, $\{\gCX^*\}$, or $\{\gCCX^*\}$, and these gates act transversally on blocks $K'_b|_{i\rightarrow j}$ in all cases except for $\{\gX^*\}$. We treat each case separately:
      \begin{description}
      \item[$\gInit$:] For every $b\in[3],j\in[k']$ for which $\bar{R}'_{\bar{t}}$ applies transversal $\gInit$ to dits $K'_b|_{i\rightarrow j}$, we apply transversal $\gTerm$ followed by $\gInit$ to dits $N_b|_{i\rightarrow j}$.
      \item[$\gTerm$:] Perform the same procedure as for the $\gInit$ case above.
      \item[$\gX^*$:] For every $b\in[3],j\in[k']$ for which $\bar{R}'_{\bar{t}}$ applies $\gX^a$ to $K'_b|_{i\rightarrow j}$ for some $a\in\bF_q^{[k']^{[u]\setminus\{i\}}}$, we apply the gadget $\cR_{\gX^a}$ in \Cref{it:tX} in \Cref{lem:transversal} to $(C_0^{\otimes[u]\setminus\{i\}},\Enc_0^{\otimes[u]\setminus\{i\}})$ on dits $N_b|_{i\rightarrow j}$.
      \item[$\gCX^*$:] For every $b_1,b_2\in[3],j_1,j_2\in[k']$ for which $\bar{R}'_{\bar{t}}$ applies $\gCX^a$ to $(C_0^{\otimes[u]\setminus\{i\}},\Enc_0^{\otimes[u]\setminus\{i\}})$ on dits $K'_{b_1}|_{i\rightarrow j_1},K'_{b_2}|_{i\rightarrow j_2}$ for some $a\in\bF_q$, we apply the gadget $\cR_{\gCX^a}$ in \Cref{it:tCX} in \Cref{lem:transversal} to $N_{b_1}|_{i\rightarrow j_1},N_{b_2}|_{i\rightarrow j_2}$.

        For each of the four cases $\gInit,\gTerm,\gX^*,\gCX^*$ above, we have run a gadget consisting of $\leq 2$ timesteps, which is by definition fault-tolerant with input and output error thresholds on each block of the code $C_0^{\otimes[u]\setminus\{i\}}$ given by
        \begin{align*}
          \lambda_{\mathrm{in}}^{(3)} &= \lambda_{\mathrm{out}}^{(2)} \\
          \lambda_{\mathrm{out}}^{(3)} &= 2\lambda_{\mathrm{in}}^{(3)}+2\lambda_{\mathrm{run}} \leq \frac{25u^2n^5}{\rho(u,n,n-k'+1)}\cdot\lambda_{\mathrm{run}}.
        \end{align*}
        
      \item[$\gCCX^*$] For every $b_1,b_2,b_3\in[3],j_1,j_2,j_3\in[k']$ for which $\bar{R}'_{\bar{t}}$ applies $\gCCX^a$ to $K'_{b_1}|_{i\rightarrow j_1}$, $K'_{b_2}|_{i\rightarrow j_2}$, $K'_{b_3}|_{i\rightarrow j_3}$ for some $a\in\bF_q$, we apply:
        \begin{itemize}
        \item The gadget in \Cref{it:sadet} in \Cref{lem:switchall} with $(C_{\mathrm{in}},\Enc_{\mathrm{in}})=(C_0^{\otimes[u]\setminus\{i\}},\Enc_0^{\otimes[u]\setminus\{i\}})$ and $(C_{\mathrm{out}},\Enc_{\mathrm{out}})=((C_0^{*2})^{\otimes[u]\setminus\{i\}},{\Enc_0'}^{\otimes[u]\setminus\{i\}})$ on dits $N_{b_3}|_{i\rightarrow j_3}$. (Note that here we replace $u$ in the statement of \Cref{lem:switchall} with $u-1$.) For symmetry across codeblocks, in parallel we also apply the gadget in \Cref{lem:switchall} with $C_{\mathrm{in}}=C_{\mathrm{out}}=C_0^{\otimes[u]\setminus\{i\}}$ on dits $N_{b_1}|_{i\rightarrow j_1}$ and $N_{b_2}|_{i\rightarrow j_2}$, though this step could be removed.
        \item The gadget $\cR_{\gCCX^a}$ in \Cref{it:tCCX} in \Cref{lem:transversal} to $(C_0^{\otimes[u]\setminus\{i\}},\Enc_0^{\otimes[u]\setminus\{i\}})$ on dits $N_{b_1}|_{i\rightarrow j_1},N_{b_2}|_{i\rightarrow j_2}$ and $((C_0^{*2})^{\otimes[u]\setminus\{i\}},{\Enc_0'}^{\otimes[u]\setminus\{i\}})$ on dits $N_{b_3}|_{i\rightarrow j_3}$. These codes satisfy the necessary condition in \Cref{eq:tCCXxp} by \Cref{claim:tensormultprop}.
        \item The gadget in \Cref{it:sadet} in \Cref{lem:switchall} with
          \begin{align*}
            (C_{\mathrm{in}},\Enc_{\mathrm{in}}) &= ((C_0^{*2})^{\otimes[u]\setminus\{i\}},{\Enc_0'}^{\otimes[u]\setminus\{i\}}) \\
            (C_{\mathrm{out}},\Enc_{\mathrm{out}}) &= (C_0^{\otimes[u]\setminus\{i\}},\Enc_0^{\otimes[u]\setminus\{i\}})
          \end{align*}
          on dits $N_{b_3}|_{i\rightarrow j_3}$. (Again we replace $u$ in the statement of \Cref{lem:switchall} with $u-1$.) Again for symmetry, in parallel we also apply the gadget in \Cref{lem:switchall} with $C_{\mathrm{in}}=C_{\mathrm{out}}=C_0^{\otimes[u]\setminus\{i\}}$ on dits $N_{b_1}|_{i\rightarrow j_1}$ and $N_{b_2}|_{i\rightarrow j_2}$.
        \end{itemize}
        By \Cref{lem:switchall,lem:transversal} along with \Cref{lem:seqcomp,lem:parcomp}, the composition of the gadgets above is a fault-detecting gadget for transversal $\gCCX^a$ on blocks $K_{b_1}|_{i\rightarrow j_1},K_{b_2}|_{i\rightarrow j_2},K_{b_3}|_{i\rightarrow j_3}$, where within each block of the code $C_0^{\otimes[u]\setminus\{i\}}$ we can take the input, output, and error detection thresholds to be
        \begin{align*}
          \lambda_{\mathrm{in}}^{(3)} &= \lambda_{\mathrm{out}}^{(2)} \\
          \lambda_{\mathrm{out}}^{(3)} &= \frac{8(u-1)^2n^4}{\rho(u-1,n,n-2k'+2)}\cdot\lambda_{\mathrm{run}} \\
          \lambda_{\mathrm{det}}^{(3)} &= (u-1)^2n^4\cdot\lambda_{\mathrm{run}}
        \end{align*}
        The upper bounds on the parameters $\lambda_{\mathrm{in}},\lambda_{\mathrm{run}}$ in \Cref{eq:sadlams} needed to apply \Cref{lem:switchall} hold by the definition of $\lambda_{\mathrm{run}}\leq\bar{\lambda}_{\mathrm{run}}$ in \Cref{eq:fdbarlamrun}.
      \end{description}
    \item For every $b\in[3]$, run the upwards switching gadget in \Cref{it:soup} in \Cref{lem:switchone} to encode into $(C_0,\Enc_0)$ in direction~$i$ on dits $N_b$, with
      \begin{align*}
        \lambda_{\mathrm{in}}^{(4)} &= \lambda_{\mathrm{out}}^{(3)} \leq \frac{25u^2n^5}{\rho(u,n,n-2k'+2)}\cdot\lambda_{\mathrm{run}} \\
        \lambda_{\mathrm{out}}^{(4)} &= n^2\cdot\lambda_{\mathrm{in}}^{(4)}+3n^3\cdot\lambda_{\mathrm{run}} \leq \frac{28u^2n^7}{\rho(u,n,n-2k'+2)}\cdot\lambda_{\mathrm{run}}.
      \end{align*}
      Specifically, each application of this gadget takes as input a codeblock of $(C_0^{\otimes[u]\setminus\{i\}})^{\sqcup k'}$ in dits $\bigsqcup_{j\in[k]}N_b|_{i\rightarrow j}$, and outputs a codeblock of $C_0^{\otimes[u]}$ in dits $N_b$.
    \item For every $b\in[3]$, run the error-detecting gadget in \Cref{lem:errdet} for the code $C_0^{\otimes u}$ on dits $N_b$ with
      \begin{align*}
        \lambda_{\mathrm{in}}^{(5)} &= \lambda_{\mathrm{out}}^{(4)} \\
        \lambda_{\mathrm{out}}^{(5)} &= \frac{8u^2n^4}{\rho(u,n,n-k'+1)}\cdot\lambda_{\mathrm{run}} \\
        \lambda_{\mathrm{det}}^{(5)} &= \frac{\rho(u,n,n-k'+1)}{8}\cdot\lambda_{\mathrm{out}} = u^2n^4\cdot\lambda_{\mathrm{run}}.
      \end{align*}
      The necessary condition $\lambda_{\mathrm{in}}^{(5)}\leq(n-k'+1)^u/2$ in \Cref{eq:edlams} needed to apply \Cref{lem:errdet} holds by the definition of $\lambda_{\mathrm{run}}\leq\bar{\lambda}_{\mathrm{run}}$ in \Cref{eq:fdbarlamrun}.
    \end{enumerate}

    By \Cref{eq:fdbarlamrun,eq:fdlams}, it holds that
    \begin{align*}
      \lambda_{\mathrm{in}}^{(1)} &\geq \lambda_{\mathrm{in}}+\lambda_{\mathrm{run}} \\
      \lambda_{\mathrm{in}}^{(1)} &\geq \lambda_{\mathrm{out}}^{(5)} \\
      \lambda_{\mathrm{out}} &\leq \lambda_{\mathrm{out}}^{(5)}+\lambda_{\mathrm{run}} \\
      \lambda_{\mathrm{det}} &\leq \min\{\lambda_{\mathrm{det}}^{(1)},\lambda_{\mathrm{det}}^{(3)},\lambda_{\mathrm{det}}^{(5)}\}.
    \end{align*}
    Therefore we may apply \Cref{lem:seqcomp} to each sequential composition in the definition of the circuit $\cR$ in \Cref{eq:fdR}, where the above bounds imply that the output error threshold from each gadget in the sequence does not exceed the input error threshold of the subsequent gadget. Thus for $\alpha\in\{\mathrm{in},\mathrm{out}\}$ we define the decorated code
    \begin{align*}
      D_\alpha' &= (C_\alpha',\; \Enc_\alpha',\; \cE_\alpha'=\cE_\alpha) = \bigsqcup_{b\in B_\alpha}(C_0^{\otimes u},\; \Enc_0^{\otimes u},\; 2^{[n]^u}|_{\geq\lambda_\alpha}),
    \end{align*}
    so that
    \begin{equation}
      \label{eq:fdgadnores}
      \left(\cR,\; \cE_{\mathrm{run}}=2^{N'}|_{\geq\lambda_{\mathrm{run}}}^{\sqcup T},\; D_{\mathrm{in}}',\; D_{\mathrm{out}}',\; \cE_{\mathrm{det}}=\bigsqcup_{t\in[T]}(2^{E_{\mathrm{det},t}}|_{\geq\lambda_{\mathrm{det}}})\right)
    \end{equation}
    provides a fault-detecting gadget for the function $\bar{O}'$ given by the composition of the functions implemented by the sequence of gadgets in \Cref{eq:fdR}, where the sets $E_{\mathrm{det},t}\subseteq N'$ for $t\in[T]$ are simply given by taking the union over $b\in[3]$ of the corresponding sets given by the applications of \Cref{lem:errdet,lem:switchall} above. Furthermore, because $\cR_1$ begins by applying the mending gadget in \Cref{lem:errdet}, \Cref{lem:seqcomp} implies that the gadget in \Cref{eq:fdgadnores} is mending. Here note that the single layer of $\gInit$ gates in \Cref{eq:fdR} preceding the application of $\cR_1$ simply adds an additional $<\lambda_{\mathrm{run}}$ errors to the input (accounted for by the bound $\lambda_{\mathrm{in}}^{(1)} \geq \lambda_{\mathrm{in}}+\lambda_{\mathrm{run}}$ above), but otherwise does not affect the mending property.

    Meanwhile, by construction the function $\bar{O}':\bF_q^{K'_{B_{\mathrm{in}}}}\rightarrow\bF_q^{K'_{B_{\mathrm{out}}}}$ simply applies the circuit $\bar{\cR}'$ defined previously to its input, except that for every dit that is terminated at some point (and possibly later initialized again), we instead simply set that dit's value to $0$ but leave it as an active dit. This change has no effect on the output, so it follows that $\bar{O}'(\cdot)=\bar{\cR}'(\cdot)$. Therefore by \Cref{eq:Rpfun},
    \begin{equation}
      \label{eq:fdOp}
      \bar{O}'(x) = (\bar{\cR}(x|_{K_{\mathrm{in}}}),\; 0^{K'_{B_{\mathrm{out}}}\setminus K_{\mathrm{out}}}) \in \bF_q^{K'_{B_{\mathrm{out}}}}.
    \end{equation}

    We will now complete the proof of \Cref{it:fd} in \Cref{thm:fdft} by showing that the gadget in \Cref{eq:fdgad} is mending fault-detecting for $\bar{\cR}(\cdot)$, with the desired space and time usage; by definition $\cR$ uses gate set $\cG$. First, we have shown that the gadget in \Cref{eq:fdgadnores} is fault-detecting for $\bar{O}'$. Meanwhile, by definition the decorated code $D_\alpha$ is given by restricting $D_\alpha'$ to codewords that encode messages whose dits in $K'_{B_\alpha}\setminus K_\alpha$ are set to $0$, that is for $x\in\bF_q^{K_\alpha}$ then
    \begin{equation}
      \label{eq:fdEncs}
      \Enc_\alpha(x) = \Enc_\alpha'(x,0^{K'_{B_\alpha}\setminus K_\alpha}).
    \end{equation}
    Then because \Cref{it:fdnoerr} in \Cref{def:faultdet} holds for the gadget in \Cref{eq:fdgadnores}, it also holds for the gadget in \Cref{eq:fdgad}. That is, for every $x\in\bF_q^{K_{\mathrm{in}}}$ and every $E\in\cE_{\mathrm{det}}$, we have
    \begin{equation*}
      \cR\circ\Enc_{\mathrm{in}}(x) = \cR\circ\Enc_{\mathrm{in}}'(x,0^{K'_{B_{\mathrm{in}}}\setminus K_{\mathrm{in}}}) = \Enc_{\mathrm{out}}'(\bar{\cR}(x),0^{K'_{B_{\mathrm{out}}}\setminus K_{\mathrm{out}}}) = \Enc_{\mathrm{out}}(\bar{\cR}(x)),
    \end{equation*}
    where the second equality above holds by \Cref{eq:fdOp}, and also
    \begin{equation*}
      E\cap\supp(\tran(\cR;\Enc_{\mathrm{in}}(x))) = E\cap\supp(\tran(\cR;\Enc_{\mathrm{in}}'(x,0^{K'_{B_{\mathrm{in}}}\setminus K_{\mathrm{in}}}))) = \emptyset.
    \end{equation*}
    Similarly, because \Cref{it:fderr} in \Cref{def:faultdet} holds for the gadget in \Cref{eq:fdgadnores}, it also holds for the gadget in \Cref{eq:fdgad}. That is, for every $x\in\bF_q^{K_{\mathrm{in}}}$, every $\cE_{\mathrm{in}}=\cE_{\mathrm{in}}'$-deviation $y$ of $\Enc_{\mathrm{in}}(x)=\Enc_{\mathrm{in}}'(x,0^{K'_{B_{\mathrm{in}}}\setminus K_{\mathrm{in}}})$, and every $\cE_{\mathrm{run}}$-avoiding fault $\cF$ for $\cR$, then either $\cR[\cF](y)$ is a $\cE_{\mathrm{out}}=\cE_{\mathrm{out}}'$-deviation of $\Enc_{\mathrm{out}}'\circ\bar{O}'(x,0^{K'_{B_{\mathrm{in}}}\setminus K_{\mathrm{in}}})=\Enc_{\mathrm{out}}\circ\bar{\cR}(x)$ (where the equality holds by \Cref{eq:fdOp,eq:fdEncs}), or else there exists a set $E\in\cE_{\mathrm{det}}$ such that $E\subseteq\supp(\tran(\cR[\cF];y))$.

    Therefore the gadget in \Cref{eq:fdgad} is fault-detecting for $\bar{\cR}(\cdot)$. It similarly follows that because the gadget in \Cref{eq:fdgadnores} is mending, the gadget in \Cref{eq:fdgad} is also mending. That is, for every $y\in\bF_q^{N_{B_{\mathrm{in}}}}$ and every $\cE_{\mathrm{run}}$-avoiding fault $\cF$ for $\cR$, then either there exists $x'\in\bF_q^{K'_{B_{\mathrm{in}}}}$ such that $\cR[\cF](y)$ is a $\cE_{\mathrm{out}}=\cE_{\mathrm{out}}'$-deviation of $\Enc_{\mathrm{out}}'\circ\bar{O}'(x')=\Enc_{\mathrm{out}}\circ\bar{\cR}(x'|_{K_{\mathrm{in}}})$ (where the equality holds by \Cref{eq:fdOp,eq:fdEncs}), or else there exists a set $E\in\cE_{\mathrm{det}}$ such that $E\subseteq\supp(\tran(\cR[\cF];y))$.

    It remains to analyze the space and time usage of $\cR$. By by the definition of $\cR$ in \Cref{eq:fdR} along with \Cref{lem:errdet,lem:switchone,lem:switchall,lem:transversal} and \Cref{lem:seqcomp,lem:parcomp}, $\cR$ uses space
    \begin{align*}
      |N'|
      &\leq 3\cdot\max\{(u+1)n^u,\; 2n^u\} = 3(u+1)n^u
    \end{align*}
    and time
    \begin{align*}
      T
      &\leq 2 + \bar{T}'\cdot(2(un^2+2) + 2(n^2+2) + 2\cdot 16u^2n^2 + 2) \\
      &\leq 2 + O(\bar{T}\cdot(u\log k)^2)\cdot(2(un^2+2) + 2(n^2+2) + 2\cdot 16u^2n^2 + 2) \\
      &\leq O(\bar{T}\cdot u^4n^2(\log n)^2),
    \end{align*}
    as desired.

  \item (Fault-tolerant gadget) The construction of the circuit $\cR$ and the analysis will be similar as for the mending fault-detecting case above: the main difference is that we use the error-correction gadget \Cref{lem:errcorr} in place of the error-detection gadget \Cref{lem:errdet}, which in turn chanes the error thresholds throughout. For this purpose, we will apply the decoders for $C_0$ and $C_0^{*2}$ in \Cref{fact:basecode}. We let $\eta_{\Dec}\geq 1$ be a sufficiently large constant such that these decoders have space and time usage
    \begin{equation}
      \label{eq:spacetimedec}
      |N_{\Dec}|,T_{\Dec} \leq \eta_{\Dec}\cdot n^2\log q,
    \end{equation}
    i.e.~$\eta_{\Dec}$ is an upper bound on the hidden constants in the big-$O$s in \Cref{fact:basecode}.

    Specifically, we define the desired circuit $\cR$ as follows. We again let $\cR$ act on dits $N'=N_{[3]}\sqcup N_{\mathrm{anc}}$, where $N_{[3]}=([n]^u)^{\sqcup 3}$ as defined above, and $N_{\mathrm{anc}}$ is a set of ancilla dits that can be used by gadgets that we will invoke in the definition of $\cR$. By definition, $\cR$ has input dits $N_{B_{\mathrm{in}}}\subseteq N_{[3]}$ and output dits $N_{B_{\mathrm{out}}}\subseteq N_{[3]}$. We again let the first timestep of $\cR$ apply $\gInit$ to all dits in $N_{[3]\setminus B_{\mathrm{in}}}$, and we let the last timestep of $\cR$ apply $\gTerm$ to all dits in $N_{[3]\setminus B_{\mathrm{out}}}$. To define the remainder of $\cR$, for each timestep $\bar{t}\in[\bar{T}']$ of $\bar{\cR}'$, we define an associated subcircuit $\cR_{\bar{t}}$ of $\cR$, acting on all dits $N'$ across $O(u^2n^2)$ timesteps, as follows. We then let $\cR$ be the sequential composition
    \begin{equation}
      \label{eq:ftR}
      \cR = \gTerm^{\sqcup N_{[3]\setminus B_{\mathrm{out}}}}\circ\cR_{\bar{T}'}\circ\cdots\circ\cR_1\circ\gInit^{\sqcup N_{[3]\setminus B_{\mathrm{in}}}}.
    \end{equation}

    By definition, $\bar{R}'_{\bar{t}}$ has the transversal structure described in \Cref{lem:compile}, meaning that there exists some $i\in[u]$ such that either $\bar{R}'_{\bar{t}}$ consists entirely of $\gX^*$ gates, or else entirely of transversal applications of some gate $G\in\{\gInit,\gTerm,\gCX^*,\gCCX^*\}$ to sets $K'_b|_{i\rightarrow j}$ for arbitrary values of $b\in[3]$ and $j\in[k]$. Therefore we let $\cR_{\bar{t}}$ perform the following:
    \begin{enumerate}[label=(\arabic*)]
    \item For every $b\in[3]$, run the error-correcting gadget in \Cref{lem:errcorr} for the code $C_0^{\otimes u}$ on dits $N_b$ with
      \begin{align*}
        \lambda_{\mathrm{in}}^{(1)} &= \left(\frac{n-k'+1}{4}\right)^u \\
        \lambda_{\mathrm{out}}^{(1)} &= \left(\frac{2^u+1}{(n-k'+1)/4n}\right)^u\cdot u^2n\cdot (\eta_{\Dec}\cdot n^2\log q)\cdot\lambda_{\mathrm{run}} \\
                                    &\leq \eta_{\Dec}\cdot 2^{u(u+6)}\cdot n^3\log(q)\cdot\lambda_{\mathrm{run}},
      \end{align*}
      where we use the radius-$(n-k'+1)/2$ decoder for $C_0$ in \Cref{fact:basecode}. The inequality above holds because $k\leq n/32$ so $(n-k'+1)/4n\geq 1/8$, and because $u^2\leq 2^{2u}$.
    \item For every $b\in[3]$, run the downwards switching gadget in \Cref{it:sodown} in \Cref{lem:switchone} to unencode from $(C_0,\Enc_0)$ in direction~$i$ on dits $N_b$, with
      \begin{align*}
        \lambda_{\mathrm{in}}^{(2)} &= \lambda_{\mathrm{out}}^{(1)} \\
        \lambda_{\mathrm{out}}^{(2)} &= n\cdot\lambda_{\mathrm{in}}^{(2)}+3n^3\cdot\lambda_{\mathrm{run}} \leq \eta_{\Dec}\cdot 2^{u(u+6)+1}\cdot n^3\log(q)\cdot\lambda_{\mathrm{run}}.
      \end{align*}
      Specifically, each application of this gadget takes as input a codeblock of $C_0^{\otimes[u]}$ in dits $N_b$, and outputs a codeblock of $(C_0^{\otimes[u]\setminus\{i\}})^{\sqcup k'}$ in dits $\bigsqcup_{j\in[k']}N_b|_{i\rightarrow j}$.
    \item By construction $\bar{R}'_{\bar{t}}$ consists entirely of gates from one of the five sets $\{\gInit\}$, $\{\gTerm\}$, $\{\gX^*\}$, $\{\gCX^*\}$, or $\{\gCCX^*\}$, and these gates act transversally on blocks $K'_b|_{i\rightarrow j}$ in all cases except for $\{\gX^*\}$. We treat each case separately:
      \begin{description}
      \item[$\gInit$:] For every $b\in[3],j\in[k']$ for which $\bar{R}'_{\bar{t}}$ applies transversal $\gInit$ to dits $K'_b|_{i\rightarrow j}$, we apply transversal $\gTerm$ followed by $\gInit$ to dits $N_b|_{i\rightarrow j}$.
      \item[$\gTerm$:] Perform the same procedure as for the $\gInit$ case above.
      \item[$\gX^*$:] For every $b\in[3],j\in[k']$ for which $\bar{R}'_{\bar{t}}$ applies $\gX^a$ to $K'_b|_{i\rightarrow j}$ for some $a\in\bF_q^{[k']^{[u]\setminus\{i\}}}$, we apply the gadget $\cR_{\gX^a}$ in \Cref{it:tX} in \Cref{lem:transversal} to $(C_0^{\otimes[u]\setminus\{i\}},\Enc_0^{\otimes[u]\setminus\{i\}})$ on dits $N_b|_{i\rightarrow j}$.
      \item[$\gCX^*$:] For every $b_1,b_2\in[3],j_1,j_2\in[k']$ for which $\bar{R}'_{\bar{t}}$ applies $\gCX^a$ to $(C_0^{\otimes[u]\setminus\{i\}},\Enc_0^{\otimes[u]\setminus\{i\}})$ on dits $K'_{b_1}|_{i\rightarrow j_1},K'_{b_2}|_{i\rightarrow j_2}$ for some $a\in\bF_q$, we apply the gadget $\cR_{\gCX^a}$ in \Cref{it:tCX} in \Cref{lem:transversal} to $N_{b_1}|_{i\rightarrow j_1},N_{b_2}|_{i\rightarrow j_2}$.

        For each of the four cases $\gInit,\gTerm,\gX^*,\gCX^*$ above, we have run a gadget consisting of $\leq 2$ timesteps, which is by definition fault-tolerant with input and output error thresholds on each block of the code $C_0^{\otimes[u]\setminus\{i\}}$ given by
        \begin{align*}
          \lambda_{\mathrm{in}}^{(3)} &= \lambda_{\mathrm{out}}^{(2)} \\
          \lambda_{\mathrm{out}}^{(3)} &= 2\lambda_{\mathrm{in}}^{(3)}+2\lambda_{\mathrm{run}} \leq \eta_{\Dec}\cdot 2^{u(u+6)+3}\cdot n^3\log(q)\cdot\lambda_{\mathrm{run}}.
        \end{align*}
        
      \item[$\gCCX^*$] For every $b_1,b_2,b_3\in[3],j_1,j_2,j_3\in[k']$ for which $\bar{R}'_{\bar{t}}$ applies $\gCCX^a$ to $K'_{b_1}|_{i\rightarrow j_1}$, $K'_{b_2}|_{i\rightarrow j_2}$, $K'_{b_3}|_{i\rightarrow j_3}$ for some $a\in\bF_q$, we apply:
        \begin{itemize}
        \item The gadget in \Cref{it:sacorr} in \Cref{lem:switchall} with $(C_{\mathrm{in}},\Enc_{\mathrm{in}})=(C_0^{\otimes[u]\setminus\{i\}},\Enc_0^{\otimes[u]\setminus\{i\}})$ and $(C_{\mathrm{out}},\Enc_{\mathrm{out}})=((C_0^{*2})^{\otimes[u]\setminus\{i\}},{\Enc_0'}^{\otimes[u]\setminus\{i\}})$ on dits $N_{b_3}|_{i\rightarrow j_3}$. (Note that here we replace $u$ in the statement of \Cref{lem:switchall} with $u-1$.) For symmetry across codeblocks, in parallel we also apply the gadget in \Cref{lem:switchall} with $C_{\mathrm{in}}=C_{\mathrm{out}}=C_0^{\otimes[u]\setminus\{i\}}$ on dits $N_{b_1}|_{i\rightarrow j_1}$ and $N_{b_2}|_{i\rightarrow j_2}$, though this step could be removed.
        \item The gadget $\cR_{\gCCX^a}$ in \Cref{it:tCCX} in \Cref{lem:transversal} to $(C_0^{\otimes[u]\setminus\{i\}},\Enc_0^{\otimes[u]\setminus\{i\}})$ on dits $N_{b_1}|_{i\rightarrow j_1},N_{b_2}|_{i\rightarrow j_2}$ and $((C_0^{*2})^{\otimes[u]\setminus\{i\}},{\Enc_0'}^{\otimes[u]\setminus\{i\}})$ on dits $N_{b_3}|_{i\rightarrow j_3}$. These codes satisfy the necessary condition in \Cref{eq:tCCXxp} by \Cref{claim:tensormultprop}.
        \item The gadget in \Cref{it:sacorr} in \Cref{lem:switchall} with $(C_{\mathrm{in}},\Enc_{\mathrm{in}})=((C_0^{*2})^{\otimes[u]\setminus\{i\}},{\Enc_0'}^{\otimes[u]\setminus\{i\}})$ and $(C_{\mathrm{out}},\Enc_{\mathrm{out}})=(C_0^{\otimes[u]\setminus\{i\}},\Enc_0^{\otimes[u]\setminus\{i\}})$ on dits $N_{b_3}|_{i\rightarrow j_3}$. (Again we replace $u$ in the statement of \Cref{lem:switchall} with $u-1$.) Again for symmetry, in parallel we also apply the gadget in \Cref{lem:switchall} with $C_{\mathrm{in}}=C_{\mathrm{out}}=C_0^{\otimes[u]\setminus\{i\}}$ on dits $N_{b_1}|_{i\rightarrow j_1}$ and $N_{b_2}|_{i\rightarrow j_2}$.
        \end{itemize}
        In the applications of \Cref{lem:switchall} above, we use the radius $\geq(n-2k'+2)/2$ decoders for $C_0$ and $C_0^{*2}$ in \Cref{fact:basecode}, which have space and time usage given by \Cref{eq:spacetimedec}. By \Cref{lem:switchall,lem:transversal} along with \Cref{lem:seqcomp,lem:parcomp}, the composition of the gadgets above is a fault-tolerant gadget for transversal $\gCCX^a$ on blocks $K_{b_1}|_{i\rightarrow j_1},K_{b_2}|_{i\rightarrow j_2},K_{b_3}|_{i\rightarrow j_3}$, where within each block of the code $C_0^{\otimes[u]\setminus\{i\}}$ we can take the input, output, and error detection thresholds to be
        \begin{align*}
          \lambda_{\mathrm{in}}^{(3)} &= \lambda_{\mathrm{out}}^{(2)} \\
          \lambda_{\mathrm{out}}^{(3)} &= \eta_{\Dec}\cdot 2^{u(u+5)}\cdot n^3\log(q)\cdot\lambda_{\mathrm{run}}.
        \end{align*}
        The upper bounds on the parameters $\lambda_{\mathrm{in}},\lambda_{\mathrm{run}}$ in \Cref{eq:saclams} needed to apply \Cref{lem:switchall} hold by the definition of $\lambda_{\mathrm{run}}\leq\bar{\lambda}_{\mathrm{run}}$ in \Cref{eq:ftbarlamrun}.
      \end{description}
    \item For every $b\in[3]$, run the upwards switching gadget in \Cref{it:soup} in \Cref{lem:switchone} to encode into $(C_0,\Enc_0)$ in direction~$i$ on dits $N_b$, with
      \begin{align*}
        \lambda_{\mathrm{in}}^{(4)} &= \lambda_{\mathrm{out}}^{(3)} \leq \eta_{\Dec}\cdot 2^{u(u+6)+3}\cdot n^3\log(q)\cdot\lambda_{\mathrm{run}} \\
        \lambda_{\mathrm{out}}^{(4)} &= n^2\cdot\lambda_{\mathrm{in}}^{(4)}+3n^3\cdot\lambda_{\mathrm{run}} \leq \eta_{\Dec}\cdot 2^{u(u+6)+4}\cdot n^5\log(q)\cdot\lambda_{\mathrm{run}}.
      \end{align*}
      Specifically, each application of this gadget takes as input a codeblock of $(C_0^{\otimes[u]\setminus\{i\}})^{\sqcup k'}$ in dits $\bigsqcup_{j\in[k]}N_b|_{i\rightarrow j}$, and outputs a codeblock of $C_0^{\otimes[u]}$ in dits $N_b$.
    \item For every $b\in[3]$, run the error-correcting gadget in \Cref{lem:errcorr} for the code $C_0^{\otimes u}$ on dits $N_b$ with
      \begin{align*}
        \lambda_{\mathrm{in}}^{(5)} &= \lambda_{\mathrm{out}}^{(4)} \\
        \lambda_{\mathrm{out}}^{(5)} &= \left(\frac{2^u+1}{(n-k'+1)/4n}\right)^u\cdot u^2n\cdot (\eta_{\Dec}\cdot n^2\log q)\cdot\lambda_{\mathrm{run}} \\
                                    &\leq \eta_{\Dec}\cdot 2^{u(u+6)}\cdot n^3\log(q)\cdot\lambda_{\mathrm{run}},
      \end{align*}
      where we use the radius-$(n-k'+1)/2$ decoder for $C_0$ in \Cref{fact:basecode}. The necessary condition $\lambda_{\mathrm{in}}^{(5)}\leq((n-k'+1)/4)^u$ in \Cref{eq:eclams} needed to apply \Cref{lem:errcorr} holds by the definition of $\lambda_{\mathrm{run}}\leq\bar{\lambda}_{\mathrm{run}}$ in \Cref{eq:ftbarlamrun}.
    \end{enumerate}

    By \Cref{eq:ftbarlamrun,eq:ftlams}, it holds that
    \begin{align*}
      \lambda_{\mathrm{in}}^{(1)} &\geq \lambda_{\mathrm{in}}+\lambda_{\mathrm{run}} \\
      \lambda_{\mathrm{in}}^{(1)} &\geq \lambda_{\mathrm{out}}^{(5)} \\
      \lambda_{\mathrm{out}} &\leq \lambda_{\mathrm{out}}^{(5)}+\lambda_{\mathrm{run}}.
    \end{align*}
    Therefore we may apply \Cref{lem:seqcomp} to each sequential composition in the definition of the circuit $\cR$ in \Cref{eq:ftR}, where the above bounds imply that the output error threshold from each gadget in the sequence does not exceed the input error threshold of the subsequent gadget. Thus for $\alpha\in\{\mathrm{in},\mathrm{out}\}$ we define the decorated code
    \begin{align*}
      D_\alpha' &= (C_\alpha',\; \Enc_\alpha',\; \cE_\alpha'=\cE_\alpha) = \bigsqcup_{b\in B_\alpha}(C_0^{\otimes u},\; \Enc_0^{\otimes u},\; 2^{[n]^u}|_{\geq\lambda_\alpha}),
    \end{align*}
    so that
    \begin{equation}
      \label{eq:ftgadnores}
      \left(\cR,\; \cE_{\mathrm{run}}=2^{N'}|_{\geq\lambda_{\mathrm{run}}}^{\sqcup T},\; D_{\mathrm{in}}',\; D_{\mathrm{out}}',\; \cE_{\mathrm{det}}=\bigsqcup_{t\in[T]}(2^{E_{\mathrm{det},t}}|_{\geq\lambda_{\mathrm{det}}})\right)
    \end{equation}
    provides a fault-tolerant gadget for the function $\bar{O}'$ given by the composition of the functions implemented by the sequence of gadgets in \Cref{eq:ftR}. Note that the single layer of $\gInit$ gates in \Cref{eq:ftR} preceding the application of $\cR_1$ simply adds an additional $<\lambda_{\mathrm{run}}$ errors to the input, which is accounted for by the bound $\lambda_{\mathrm{in}}^{(1)} \geq \lambda_{\mathrm{in}}+\lambda_{\mathrm{run}}$ above.

    Meanwhile, by construction the function $\bar{O}':\bF_q^{K'_{B_{\mathrm{in}}}}\rightarrow\bF_q^{K'_{B_{\mathrm{out}}}}$ simply applies the circuit $\bar{\cR}'$ defined previously to its input, except that for every dit that is terminated at some point (and possibly later initialized again), we instead simply set that dit's value to $0$ but leave it as an active dit. This change has no effect on the output, so it follows that $\bar{O}'(\cdot)=\bar{\cR}'(\cdot)$. Therefore by \Cref{eq:Rpfun},
    \begin{equation}
      \label{eq:ftOp}
      \bar{O}'(x) = (\bar{\cR}(x|_{K_{\mathrm{in}}}),\; 0^{K'_{B_{\mathrm{out}}}\setminus K_{\mathrm{out}}}) \in \bF_q^{K'_{B_{\mathrm{out}}}}.
    \end{equation}

    We will now complete the proof of \Cref{it:ft} in \Cref{thm:fdft} by showing that the gadget in \Cref{eq:ftgad} is fault-tolerant for $\bar{\cR}(\cdot)$, with the desired space and time usage; by definition $\cR$ uses gate set $\cG$. First, we have shown that the gadget in \Cref{eq:ftgadnores} is fault-tolerant for $\bar{O}'$. Meanwhile, by definition the decorated code $D_\alpha$ is given by restricting $D_\alpha'$ to codewords that encode messages whose dits in $K'_{B_\alpha}\setminus K_\alpha$ are set to $0$, that is for $x\in\bF_q^{K_\alpha}$ then
    \begin{equation}
      \label{eq:ftEncs}
      \Enc_\alpha(x) = \Enc_\alpha'(x,0^{K'_{B_\alpha}\setminus K_\alpha}).
    \end{equation}
    Then because \Cref{it:ftnoerr} in \Cref{def:faulttol} holds for the gadget in \Cref{eq:ftgadnores}, it also holds for the gadget in \Cref{eq:ftgad}. That is, for every $x\in\bF_q^{K_{\mathrm{in}}}$, we have
    \begin{equation*}
      \cR\circ\Enc_{\mathrm{in}}(x) = \cR\circ\Enc_{\mathrm{in}}'(x,0^{K'_{B_{\mathrm{in}}}\setminus K_{\mathrm{in}}}) = \Enc_{\mathrm{out}}'(\bar{\cR}(x),0^{K'_{B_{\mathrm{out}}}\setminus K_{\mathrm{out}}}) = \Enc_{\mathrm{out}}(\bar{\cR}(x)),
    \end{equation*}
    where the second equality above holds by \Cref{eq:ftOp}.
    Similarly, because \Cref{it:fterr} in \Cref{def:faulttol} holds for the gadget in \Cref{eq:ftgadnores}, it also holds for the gadget in \Cref{eq:ftgad}. That is, for every $x\in\bF_q^{K_{\mathrm{in}}}$, every $\cE_{\mathrm{in}}=\cE_{\mathrm{in}}'$-deviation $y$ of $\Enc_{\mathrm{in}}(x)=\Enc_{\mathrm{in}}'(x,0^{K'_{B_{\mathrm{in}}}\setminus K_{\mathrm{in}}}),$ and every $\cE_{\mathrm{run}}$-avoiding fault $\cF$ for $\cR$, then $\cR[\cF](y)$ is a $\cE_{\mathrm{out}}=\cE_{\mathrm{out}}'$-deviation of
    \begin{equation*}
      \Enc_{\mathrm{out}}'\circ\bar{O}'(x,0^{K'_{B_{\mathrm{in}}}\setminus K_{\mathrm{in}}}) = \Enc_{\mathrm{out}}\circ\bar{\cR}(x),
    \end{equation*}
    where the equality holds by \Cref{eq:ftOp,eq:ftEncs}.

    Therefore the gadget in \Cref{eq:fdgad} is fault-tolerant for $\bar{\cR}(\cdot)$.

    It remains to analyze the space and time usage of $\cR$. By by the definition of $\cR$ in \Cref{eq:ftR} along with \Cref{lem:errcorr,eq:spacetimedec,lem:switchone,lem:switchall,lem:transversal} and \Cref{lem:seqcomp,lem:parcomp}, $\cR$ uses space
    \begin{align*}
      |N'|
      &\leq 3\cdot\max\{n^{u-1}\cdot\eta_{\Dec}\cdot n^2\log(q),\; 2n^u\} = 3\eta_{\Dec}\cdot\log(q)\cdot n^{u+1}
    \end{align*}
    and time
    \begin{align*}
      T
      &\leq 2 + \bar{T}'\cdot(2\cdot u\cdot\eta_{\Dec}\cdot n^2\log(q) + 2(n^2+2) + 2(2u^2\cdot\eta_{\Dec}\cdot n^2\log(q)+8un^2) + 2) \\
      &\leq 2 + O(\bar{T}\cdot(u\log k)^2) \cdot O(u^2n^2\log(q)) \\
      &\leq O(\bar{T}\cdot u^4n^2(\log n)^2(\log q)).
    \end{align*}
    as desired.
  \end{enumerate}
\end{proof}

\section{Instantiation and Alphabet Reduction}
\label{sec:instan}
In this section, we instantiate our schemes in \Cref{thm:fdft} with specific parameters, and show how to reduce the alphabet size down to a constant. By doing so, we obtain schemes for compiling an arbitrary logical circuit with constant alphabet size into a mending fault-detecting or a fault-tolerant physical circuit. If the logical circuit uses space $|\bar{N}|$, the physical circuit uses space $|N|\leq|\bar{N}|^{1+\epsilon}$ for an arbitrarily small constant $\epsilon>0$. The mending fault-detecting circuit can withstand faults with $|N|/\poly\log|N|$ errors per timestep, while the fault-tolerant circuit can withstand $|N|^{1-o(1)}$ errors per timestep. The formal statements are given below.

\begin{theorem}[Main result on fault-detection and fault-tolerance]
  \label{thm:main}
  The following two statements hold:
  \begin{enumerate}
  \item\label{it:mainfd} (Mending fault-detecting gadget) For every prime power $r$ and every $0<\epsilon<1/8$, there exists an $\bar{N}_0(r,\epsilon)\geq 0$ such that the following holds.
  For every set $\bar{N}$ of size $|\bar{N}|\geq\bar{N}_0(r,\epsilon)$ and every $K\subseteq\bar{N}$, there is an integer $M=M(\epsilon,\bar{N})\in[|\bar{N}|^{1+\epsilon},|\bar{N}|^{1+4\epsilon}]$ and an $[n(r,\epsilon,\bar{N}),\; |K|,\; d(r,\epsilon,\bar{N})]_r$ code $C(r,\epsilon,K,\bar{N})$ with associated encoding map $\Enc_{C(r,\epsilon,K,\bar{N})}$, where the $|K|$ message dits are labeled by the set $K$ and
    \begin{align*}
      n(r,\epsilon,\bar{N}) &\leq \frac{\log\log|\bar{N}|}{\epsilon}\cdot M \\
      d(r,\epsilon,\bar{N}) &\geq \frac{M}{2^{32}},
    \end{align*}
    such that the following statement holds.
    
    Let $\eta>1$ be the absolute constant in \Cref{lem:loctest}. For every circuit $\bar{\cR}$ acting on a set of $r$-ary dits labeled by $\bar{N}$ using time $\bar{T}$ and gate set $\cG=\{\gInit,\gTerm,\gX^*,\gCX^*,\gCCX^*\}$ with input and output dits $K_{\mathrm{in}},K_{\mathrm{out}}\subseteq\bar{N}$ respectively, letting
    \begin{align*}
      D_{\mathrm{in}} &= (C_{\mathrm{in}}=C(r,\epsilon,K_{\mathrm{in}},\bar{N}),\; \Enc_{\mathrm{in}}=\Enc_{C_{\mathrm{in}}},\; \cE_{\mathrm{in}}=2^{[n(r,\epsilon,\bar{N})]}|_{\geq\lambda_{\mathrm{in}}}) \\
      D_{\mathrm{out}} &= (C_{\mathrm{out}}=C(r,\epsilon,K_{\mathrm{out}},\bar{N}),\; \Enc_{\mathrm{out}}=\Enc_{C_{\mathrm{out}}},\; \cE_{\mathrm{out}}=2^{[n(r,\epsilon,\bar{N})]}|_{\geq\lambda_{\mathrm{out}}}),
    \end{align*}
    then for every $\lambda_{\mathrm{run}}\in[0,\bar{\lambda}_{\mathrm{run}}]$ with
    \begin{equation}
      \label{eq:mdbarlamrun}
      \bar{\lambda}_{\mathrm{run}} = \bar{\lambda}_{\mathrm{run}}(r,\epsilon,\bar{N}) = \frac{M}{(\log|\bar{N}|)^{16\eta/\epsilon}},
    \end{equation}
    there exists $T\in\bN$ and a set $N$ with subsets $E_{\mathrm{det},t}\subseteq N$ for $t\in[T]$ such that there is a mending fault-detecting gadget
    \begin{equation*}
      \left(\cR,\; \cE_{\mathrm{run}}=2^N|_{\geq\lambda_{\mathrm{run}}},\; D_{\mathrm{in}},\; D_{\mathrm{out}},\; \cE_{\mathrm{det}}=\bigsqcup_{t\in[T]}(2^{E_{\mathrm{det},t}}|_{\geq\lambda_{\mathrm{det}}})\right)
    \end{equation*}
    for $\bar{\cR}(\cdot)$, where $\cR$ is a circuit using space $|N|\leq (\log|\bar{N}|)^2 \cdot M$, time $T\leq (\log|\bar{N}|)^{8/\epsilon}\cdot\bar{T}$, and gate set $\cG$, and where
    \begin{align}
      \label{eq:mdlams}
      \begin{split}
        \lambda_{\mathrm{in}} &= \frac{M}{2^{32}} \\
        \lambda_{\mathrm{out}} &= (\log|\bar{N}|)^{16\eta/\epsilon}\cdot\lambda_{\mathrm{run}} \\
        \lambda_{\mathrm{det}} &= \lambda_{\mathrm{run}}.
      \end{split}
    \end{align}
    
  \item\label{it:mainft} (Fault-tolerant gadget) For every prime power $r$, there exists $\bar{N}_0(r)\geq 0$ such that the following holds. For every set $\bar{N}$ of size $|\bar{N}|\geq\bar{N}_0(r)$ and every $K\subseteq\bar{N}$, there is an integer $M=M(\bar{N})\in[|\bar{N}|,\;2^{12(\log|\bar{N}|)^{2/3}}\cdot|\bar{N}|]$ and a $[n(r,\bar{N}),\; |K|,\; d(r,\bar{N})]_r$ code $C(r,K,\bar{N})$ with associated encoding map $\Enc_{C(r,K,\bar{N})}$, where the $|K|$ message dits are labeled by the set $K$ and
    \begin{align*}
      n(r,\bar{N}) &\leq (\log|\bar{N}|)\cdot M \\
      d(r,\bar{N}) &\geq \frac{M}{2^{2(\log|\bar{N}|)^{1/3}}},
    \end{align*}
    such that the following statement holds:
    
    For every circuit $\bar{\cR}$ acting on a set of $r$-ary dits labeled by $\bar{N}$ using time $\bar{T}$ and gate set $\cG=\{\gInit,\gTerm,\gX^*,\gCX^*,\gCCX^*\}$ with input and output dits $K_{\mathrm{in}},K_{\mathrm{out}}\subseteq\bar{N}$ respectively, letting
    \begin{align*}
      D_{\mathrm{in}} &= (C_{\mathrm{in}}=C(r,K_{\mathrm{in}},\bar{N}),\; \Enc_{\mathrm{in}}=\Enc_{C_{\mathrm{in}}},\; \cE_{\mathrm{in}}=2^{[n(r,\bar{N})]}|_{\geq\lambda_{\mathrm{in}}}) \\
      D_{\mathrm{out}} &= (C_{\mathrm{out}}=C(r,K_{\mathrm{out}},\bar{N}),\; \Enc_{\mathrm{out}}=\Enc_{C_{\mathrm{out}}},\; \cE_{\mathrm{out}}=2^{[n(r,\bar{N})]}|_{\geq\lambda_{\mathrm{out}}}),
    \end{align*}
    then for every $\lambda_{\mathrm{run}}\in[0,\bar{\lambda}_{\mathrm{run}}]$ with
    \begin{equation}
      \label{eq:mtbarlamrun}
      \bar{\lambda}_{\mathrm{run}} = \bar{\lambda}_{\mathrm{run}}(r,\epsilon,\bar{N}) = \frac{M}{2^{32(\log|\bar{N}|)^{2/3}}},
    \end{equation}
    there exists a fault-tolerant gadget
    \begin{equation*}
      \left(\cR,\; \cE_{\mathrm{run}}=2^N|_{\geq\lambda_{\mathrm{run}}},\; D_{\mathrm{in}},\; D_{\mathrm{out}}\right)
    \end{equation*}
    for $\bar{\cR}(\cdot)$, where $\cR$ is a circuit using space $|N|\leq 2^{2(\log|\bar{N}|)^{2/3}}\cdot M$, time $T\leq 2^{4(\log|\bar{N}|)^{2/3}}\cdot\bar{T}$, and gate set $\cG$, and where
    \begin{align}
      \label{eq:mtlams}
      \begin{split}
        \lambda_{\mathrm{in}} &= \frac{M}{2^{8(\log|\bar{N}|)^{1/3}}} \\
        \lambda_{\mathrm{out}} &= 2^{8(\log|\bar{N}|)^{2/3}}\cdot\lambda_{\mathrm{run}}.
      \end{split}
    \end{align}
  \end{enumerate}
  Furthermore, in both statements above, if $r,\epsilon$ are fixed constants, then $\cR$ can be constructed in $\poly(|\bar{N}|)$ time from $\bar{\cR}$.
\end{theorem}

The remainder of this section is dedicated to proving \Cref{thm:main}. Our construction of $\cR$ from $\bar{\cR}$ in \Cref{thm:main} described below is inherently explicit, as it consists of performing a sequence of explicit operations on tensor products of explicit Reed-Solomon codes. Therefore the claim in \Cref{thm:main} that $\cR$ can be constructed in polynomial time is immediate.

In \Cref{sec:alphchange} below, we show how to simulate circuits over a given alphabet size using dits of a different alphabet size. We will use such changes in alphabet size to prove \Cref{thm:main} using \Cref{thm:fdft}, as the latter provides fault-detecting and fault-tolerance schemes only over alphabets that grow with the circuit's space usage. Specifically, for both \Cref{it:mainfd} and \Cref{sec:mainft} of \Cref{thm:main}, given a logical circuit $\bar{\cR}$ over $r$-ary dits, we first simulate $\bar{\cR}$ using a circuit over larger $q$-ary dits, where $\bF_q\supseteq\bF_r$ is an extension field. We are able to apply \Cref{thm:fdft} to this logical $q$-ary circuit to obtain a physical fault-detecting/fault-tolerant $q$-ary circuit. We then simulate this physical $q$-ary circuit using $r$-ary dits to obtain our desired gadget over $r$-ary dits.

There is one additional complication in proving \Cref{it:mainfd} in \Cref{thm:main}: when simulating the fault-detecting $q$-ary circuit using an $r$-ary circuit, to ensure the mending property holds, we must check that the circuit's input dits lie in the subfield $\bF_r\subseteq\bF_q$. For this purpose, for each logical input $x\in\bF_q$, we compute $x^r-x$, and check that the result is $0$; if this check is violated, then we have detected that $x\in\bF_q\setminus\bF_r$.

We formalize these ideas to prove \Cref{it:mainfd} and \Cref{sec:mainft} of \Cref{thm:main} in \Cref{sec:mainfd} and \Cref{sec:mainft} below, respectively.

\subsection{Increasing and Reducing Alphabet Size}
\label{sec:alphchange}
In this section, we show how to simulate circuits over a small alphabet using a larger alphabet, and vice versa. Throughout this section, for an $r$-ary gate $G=G_r\in\{\gInit,\gTerm,\gX^*,\gCX^*,\gCCX^*\}$ and a field extension $\bF_q=\bF_{r^\kappa}\supseteq\bF_r$, we let $G_q$ denote the $q$-ary version of $G$. For instance, if $G_r=\gCX^a:\bF_r^2\rightarrow\bF_r^2$ for $a\in\bF_r$, then $G_q=\gCX^a:\bF_q^2\rightarrow\bF_q^2$ with the same $a\in\bF_r\subseteq\bF_q$.

\Cref{lem:alphup} below shows that we can replace every $r$-ary gate with the analogous $q$-ary gate, while preserving the gate's functionality on inputs in the subfield $\bF_r\subseteq\bF_q$.

\begin{lemma}[Increase alphabet size]
  \label{lem:alphup}
  Let $\iota:\bF_r\hookrightarrow\bF_q$ be a field inclusion. For every $a\in\bF_r$ and every gate $G\in\{\gInit,\gTerm,\gX^a,\gCX^a,\gCCX^a\}$ with $m_{\mathrm{in}}$ input dits and $m_{\mathrm{out}}$ output dits, then
  \begin{equation*}
    G_q\circ\iota^{\sqcup m_{\mathrm{in}}} = \iota^{\sqcup m_{\mathrm{out}}}\circ G_r.
  \end{equation*}
\end{lemma}
\begin{proof}
  The result follows immediately by the definition of the gates $\gInit,\gTerm,\gX^a,\gCX^a,\gCCX^a$, as restricting the $q$-ary version of each of these gates to the subfield $\bF_r\subseteq\bF_q$ simply yields the $r$-ary version of the gate.
\end{proof}

\Cref{lem:alphdown} below shows that we can every $q$-ary gate with a small $r$-ary circuit, while preserving the circuit's functionality up to an isomorphism $\bF_r^\kappa\cong\bF_q$.

\begin{lemma}[Reduce alphabet size]
  \label{lem:alphdown}
  Let $r$ be a prime power, and $\kappa\in\bN$. Let $q=r^\kappa$, and let $\phi:\bF_r^\kappa\xrightarrow{\sim}\bF_q$ be an arbitrary fixed $\bF_r$-linear isomorphism. Then for every gate $G\in\cG:=\{\gInit,\gTerm,\gX^*,\gCX^*,\gCCX^*\}$ over $q$-ary dits, there exists a circuit $\cR=\cR_G$ using gate set $\cG$ over $r$-ary dits, using space $|N|=\kappa\cdot\max\{m_{\mathrm{in}},m_{\mathrm{out}}\}$ and time $T=\kappa^3$, such that
  \begin{equation}
    \label{eq:Gsim}
    \phi^{\sqcup m_{\mathrm{out}}}\circ\cR = G\circ\phi^{\sqcup m_{\mathrm{in}}},
  \end{equation}
  where $m_{\mathrm{in}}$ and $m_{\mathrm{out}}$ denote the number of input and output dits in $G$, respectively, so that $\cR$ has $\kappa\cdot m_{\mathrm{in}}$ input dits and $\kappa\cdot m_{\mathrm{out}}$ output dits.
\end{lemma}
\begin{proof}
  If $G=\gInit$ or $G=\gTerm$, we simply let $\cR$ apply $G^{\sqcup\kappa}$. If $G=\gX^a$ for $a\in\bF_q$, we simply let $\cR$ apply $\gX^{\phi^{-1}(a)}$.

  If $G=\gCX^a$ for $a\in\bF_q$, then
  \begin{equation*}
    (\phi^{-1})^{\sqcup 2}\circ\gCX^a\circ\phi^{\sqcup 2}(x_1,x_2) = (x_1,x_2+\phi^{-1}(a\phi(x_1)))
  \end{equation*}
  where $x_1\mapsto\phi^{-1}(a\phi(x_1))$ is a linear map on $\bF_r^\kappa$, so that there exists a matrix $A\in\bF_r^{[\kappa]\times[\kappa]}$ such that $Ax_1=\phi^{-1}(a\phi(x_1))$. Then for every $(i,j)\in[\kappa]^2$, we let $\cR$ apply a $\gCX^{A_{i,j}}$ gate from dit $j$ of $x_1$ to dit $i$ of $x_2$; these $\kappa^2$ gates together implement the function $(x_1,x_2)\mapsto(x_1,x_2+Ax_1)$.

  If $G=\gCCX^a$ for $a\in\bF_q$, then
  \begin{equation*}
    (\phi^{-1})^{\sqcup 3}\circ\gCCX^a\circ\phi^{\sqcup 3}(x_1,x_2,x_3) = (x_1,x_2,x_3+\phi^{-1}(a\phi(x_1)\phi(x_2)))
  \end{equation*}
  where $(x_1,x_2)\mapsto\phi^{-1}(a\phi(x_1)\phi(x_2))$ is a bilinear map on $\bF_r^\kappa\times\bF_r^\kappa$. Because every such bilinear map factors through the tensor product $\bF_r^\kappa\otimes\bF_r^\kappa=\bF_q^{[\kappa]^2}$, there exists a matrix $B\in\bF_r^{[\kappa]\times[\kappa]^2}$ such that $B(x_1\otimes x_2)=\phi^{-1}(a\phi(x_1)\phi(x_2))$. Then for every $(i,j,k)\in[\kappa]^3$, we let $\cR$ apply $\gCCX^{B_{i,(j,k)}}$ to dit $j$ of $x_1$, dit $k$ of $x_2$, and dit $i$ of $x_3$; these $\kappa^3$ gates together implement the function $(x_1,x_2,x_3)\mapsto(a_1,x_2,x_3+B(x_1\otimes x_2))$.

  For every gate $G\in\cG$ acting on $m=\max\{m_{\mathrm{in}},m_{\mathrm{out}}\}$ $q$-ary dits, the associated circuit $\cR=\cR_G$ defined above by construction satisfies \Cref{eq:Gsim}, acts on $|N|=\kappa\cdot m$ $r$-ary dits, and uses time $\leq\kappa^3$. We may add additional idling timesteps to use time exactly $T=\kappa^3$ for all $G$, as desired.
\end{proof}

\subsection{Mending Fault-Detecting Gadget}
\label{sec:mainfd}
In this section, we prove \Cref{it:mainfd} in \Cref{thm:main}. Here, all logarithms should be assumed to be base $2$ unless explicitly stated otherwise. In this section, we will also need the following extended version of \Cref{def:faultdet}, which describes gadgets over an alphabet $\bF_q$ that implement functions $\bar{O}$ acting on a smaller alphabet $\bF_r\subseteq\bF_q$:

\begin{definition}
  \label{def:faultdetalph}
  Let $\bF_r\subseteq\bF_q$ be a subfield. Let $\cR$ be a circuit over $q$-ary dits with input dits $N_{\mathrm{in}}$ and output dits $N_{\mathrm{out}}$. For $\alpha\in\{\mathrm{in},\mathrm{out}\}$, let $D_\alpha$ be a $[|N_\alpha|,k_\alpha]_q$ decorated code. We say $(\cR,\cE_{\mathrm{run}},D_{\mathrm{in}},D_{\mathrm{out}},\cE_{\mathrm{det}})$ forms a \emph{fault-detecting gadget} for a set $\bar{\cO}$ of functions $\bar{O}:\bF_r^{k_{\mathrm{in}}}\rightarrow\bF_r^{k_{\mathrm{out}}}$ if it satisfies \Cref{it:fdnoerr} and \Cref{it:fderr} in \Cref{def:faultdet}, but where we only consider every $x\in\bF_r^{k_{\mathrm{in}}}$ instead of $x\in\bF_q^{k_{\mathrm{in}}}$. Similarly, we say the gadget is furthermore \emph{mending} if either \Cref{it:fdecorrmend} or \Cref{it:fdedetmend} in \Cref{def:faultdet} holds, but where in \Cref{it:fdecorrmend} we require $x'\in\bF_r^{k_{\mathrm{in}}}\subseteq\bF_q^{k_{\mathrm{in}}}$ instead of $x\in\bF_q^{k_{\mathrm{in}}}$.
\end{definition}

\begin{proof}[Proof of \Cref{it:mainfd} in \Cref{thm:main}]
  At various points throughout the proof below, we will require that $\bar{N}_0(r,\epsilon)$ be sufficiently large compared to $r,\; 1/\epsilon$ so that certain bounds on $|\bar{N}|$ that hold asymptotically in fact hold for every $|\bar{N}|\geq\bar{N}_0(r,\epsilon)$. The final value of $\bar{N}_0(r,\epsilon)$ will then implicitly be given by the maximum of the values required in these various instances.
  
  Given the set $\bar{N}$, we first let $\bar{N}'=\bar{N}\times[4]\cong\bar{N}^{\sqcup 4}$.
  We then define\footnote{Note that the variable $n$ here is distinct from $n(r,\epsilon,\bar{N})$ in the statement of \Cref{thm:main}.}
  \begin{align}
    \label{eq:mdnvars}
    \begin{split}
      n &= \left\lfloor\left(\frac{\epsilon}{1-2\epsilon}\cdot\frac{\log|\bar{N}'|}{\log\log|\bar{N}'|}\right)^{1/\epsilon}\right\rfloor \\
      k &= \lfloor n^{1-\epsilon}\rfloor \\
      u &= \lfloor n^\epsilon\rfloor \\
      \kappa &= \lceil\log(n)/\log(r)\rceil \\
      q &= r^\kappa.
    \end{split}
  \end{align}
  We let
  \begin{align*}
    M &= M(\epsilon,\bar{N}) = n^u,
  \end{align*}
  so that if $|\bar{N}|\geq\bar{N}_0(r,\epsilon)$ is sufficiently large, then
  \begin{align*}
    M
    &\geq 2^{(n^\epsilon-1)\log n} \\
    &\geq 2^{\frac{\epsilon}{1-\epsilon}\cdot\frac{\log|\bar{N}'|}{\log\log|\bar{N}'|}\cdot\frac{1}{\epsilon}\log\log|\bar{N}'|} \\
    &\geq |\bar{N}'|^{1+\epsilon} \\
    &\geq |\bar{N}|^{1+\epsilon}
  \end{align*}
  and
  \begin{align*}
    M
    &\leq 2^{n^\epsilon\log n} \\
    &\leq 2^{\frac{\epsilon}{1-2\epsilon}\cdot\frac{\log|\bar{N}'|}{\log\log|\bar{N}'|}\cdot\frac{1}\epsilon\log\log|\bar{N}'|} \\
    &\leq 2^{\frac{\log|\bar{N}|}{1-1.9\epsilon}} \\
    &\leq |\bar{N}|^{1+4\epsilon},
  \end{align*}
  where the final inequality above uses the assumption that $0<\epsilon<1/8$. Thus $M\in[|\bar{N}|^{1+\epsilon},|\bar{N}|^{1+4\epsilon}]$, as desired.
  Furthermore, assuming $|\bar{N}|\geq\bar{N}_0(r,\epsilon)$ is sufficiently large, then
  \begin{align*}
    k^u
    &\geq 2^{(n^\epsilon-1)\log(n^{1-\epsilon}-1)} \\
    &\geq 2^{\frac{\epsilon}{1-1.9\epsilon}\cdot\frac{\log|\bar{N}'|}{\log\log|\bar{N}'|}\cdot(1-\epsilon)\cdot\frac{1}{\epsilon}\log\log|\bar{N}'|} \\
    &\geq |\bar{N}'|\\
  \end{align*}
  Therefore we may fix an arbitrary set inclusion $\bar{N}'\hookrightarrow[k]^u$, so that we may view $\bar{N}'$ as a subset of $[k]^u$. Also if $|\bar{N}|\geq\bar{N}_0(r,\epsilon)$ is sufficiently large, then we have $u\geq 4$, $k\leq n/32$, and $q\geq n$. Thus for $K\subseteq\bar{N}'\subseteq[k]^u$, the inequalities above ensure that the code $C(q,u,n,k,K)$ with encoding map $\Enc_{C(q,u,n,k,K)}$ in \Cref{thm:fdft} is well-defined.
  By \Cref{thm:fdft}, $C(q,u,n,k,K)$ has parameters $[n^u,|K|,\geq(n-8k+1)^u]_q$.

  To prove \Cref{it:mainfd} in \Cref{thm:main}, we will first construct fault-detecting gadgets over the alphabet $\bF_q$, before reducing the alphabet to $\bF_r$ using \Cref{lem:alphdown}.

  For this purpose, we first construct a circuit $\bar{\cR}^1$ acting on $q$-ary dits labeled by $\bar{N}'=\bar{N}\times[4]\cong\bar{N}^{\sqcup 4}$ as follows. $\bar{\cR}^1$ has input dits $K_{\mathrm{in}}\times\{1\}$ and output dits $K_{\mathrm{in}}\times\{1,2\}$. The first timestep of $\bar{\cR}^1$ applies $\gInit$ gates to all dits in $K_{\mathrm{in}}\times[4]$, and the last timestep of $\bar{\cR}^1$ applies $\gTerm$ gates to all dits in $K_{\mathrm{in}}\times\{3,4\}$. On input $x\in\bF_q^{K_{\mathrm{in}}}$, for each $j\in K_{\mathrm{in}}$, the remaining timesteps of $\bar{\cR}^1$ use dits $\{j\}\times[4]$ to perform repeated squaring in order to compute $x_j^r-x_j$, and output this value in dit $(j,2)$; the output in dit $(j,1)$ is simply the input value $x_j$. Note that this value $x_j^r-x_j$ equals $0$ iff $x_j$ lies in the subfield $\bF_r\subseteq\bF_q$; below we will apply \Cref{thm:fdft} to $\bar{\cR}^1$, and use these values $x_j^r-x_j$ to check that the logical inputs $x_j$ lie in the subfield $\bF_r$, in order to ultimately obtain a mending gadget over $r$-ary dits.

  Specifically, for each $j\in K_{\mathrm{in}}$, $\bar{\cR}^1$ initializes dits $\{j\}\times\{2,3,4\}$ to value $0$; the value $x_j$ of dit $(j,1)$ is never changed. $\bar{\cR}^1$ then performs $\gCX^1$ to copy value $x_j$ into dit $(j,2)$. $\bar{\cR}^1$ then loops through the bits in the binary representation of $r$ from least to most significant; for each bit equal to $0$, then $\bar{\cR}^1$ squares the value in dit $(j,2)$, while for each bit equal to $1$, then $\bar{\cR}^1$ squares the value in dit $(j,2)$ and subsequently multiplies $x_j$ (from dit $(j,1)$) into the value in dit $(j,2)$. Upon completion of this procedure, by definition dit $(j,2)$ will contain $x_j^r$. $\bar{\cR}^1$ then performs $\gCX^{-1}$ on dits $(j,1),(j,2)$ to compute $x_j^r-x_j$ in dit $(j,2)$. Each squaring and multiplication operation can be performed using a sequence of $\leq 10$ gates in $\cG$ acting on the dits $\{j\}\times[4]$, where dits $\{j\}\times\{3,4\}$ are used as ancillas into which values can be copied using $\gCX^*$ gates (after resetting the ancillas to $0$ using $\gTerm$ and $\gInit$), and multiplication (and therefore also squaring) is performed using $\gCCX^*$ gates.
  Thus $\bar{\cR}^1$ uses space $\bar{N}'$ and time $\bar{T}^1\leq 16\log r$.

  Let
  \begin{align}
    \label{eq:mdtlams}
    \begin{split}
      \tilde{\lambda}_{\mathrm{run}} &= 3\kappa^3\cdot\lambda_{\mathrm{run}} \\
      \tilde{\lambda}_{\mathrm{in}} &= \frac{(n-8k+1)^u}{4} \\ 
      \tilde{\lambda}_{\mathrm{out}} &= \frac{16u^2n^4}{\rho(u,n,n-8k+1)}\cdot\tilde{\lambda}_{\mathrm{run}} \\ 
      \tilde{\lambda}_{\mathrm{det}} &= (u-1)^2n^4\cdot\tilde{\lambda}_{\mathrm{run}} 
    \end{split}
  \end{align}
  and for $\alpha\in\{\mathrm{in},\mathrm{out}\}$ let
  \begin{align*}
    \tilde{D}_\alpha &= (C(q,u,n,k,K_\alpha),\; \Enc_{C(q,u,n,k,K_\alpha)},\; 2^{[n]^u}|_{\geq\tilde{\lambda}_\alpha}) \\
    \tilde{D}_\alpha' &= (C(q,u,n,k,K_\alpha),\; \Enc_{C(q,u,n,k,K_\alpha)},\; 2^{[n]^u}|_{\geq\tilde{\lambda}_\alpha-2\tilde{\lambda}_{\mathrm{run}}}).
  \end{align*}
  
  Assuming $|\bar{N}|\geq\bar{N}_0(r,\epsilon)$ is sufficiently large, we have
  \begin{align*}
    (n-16k+2)^u
    &\geq M\cdot\left(1-\frac{16}{n^\epsilon}\right)^{n^\epsilon} \geq \frac{M}{2^{64}},
  \end{align*}
  so the expression $\bar{\lambda}_{\mathrm{run}}(u,n)$ in \Cref{eq:fdbarlamrun} satisfies
  \begin{align}
    \label{eq:mdblams}
    \begin{split}
      \bar{\lambda}_{\mathrm{run}}(u,n)
      &\geq \left(\frac{(n-16k+2)^u}{un\cdot n^u}\right)^\eta\cdot\frac{(n-16k+2)^{u-1}}{2^{10}\cdot u^2n^8} \\
      &\geq \frac{M}{(un\cdot 2^{64})^\eta\cdot 2^{74}\cdot u^2n^9} \\
      &\geq \frac{M}{(\log|\bar{N}|)^{15\eta/\epsilon}} \\
      &\geq 3\kappa^3 \cdot \bar{\lambda}_{\mathrm{run}}(r,\epsilon,\bar{N}),
    \end{split}
  \end{align}
  where the final equality above holds by \Cref{eq:mdbarlamrun,eq:mdnvars}. Then because $\tilde{\lambda}_{\mathrm{run}}=3\kappa^3\cdot\lambda_{\mathrm{run}}\leq 3\kappa^3\cdot\bar{\lambda}_{\mathrm{run}}(r,\epsilon,\bar{N})$, we have $\tilde{\lambda}_{\mathrm{run}}\leq\bar{\lambda}_{\mathrm{run}}(u,n)$. By \Cref{eq:mdtlams,eq:mdblams,eq:fdbarlamrun}, we also have
  \begin{align*}
    \tilde{\lambda}_{\mathrm{out}}+\tilde{\lambda}_{\mathrm{run}}
    &\leq \frac{32u^2n^4}{\rho(u,n,n-8k+1)} \cdot \bar{\lambda}_{\mathrm{run}}(u,n) \\
    &= \frac{32u^2n^4}{\rho(u,n,n-8k+1)} \cdot \frac{\rho(u,n,n-16k+2)\cdot(n-16k+2)^{u-1}}{2^{10}\cdot u^2n^8} \\
    &\leq \frac{(n-16k+2)^{u-1}}{2^5n^4} \\
    &\leq \tilde{\lambda}_{\mathrm{in}}.
  \end{align*}

  Therefore by \Cref{it:fd} in \Cref{thm:fdft}, there exists a circuit $\tilde{\cR}^1$ using space $|\tilde{N}^1|\leq 3(u+1)n^u$, time $\tilde{T}^1\leq O(\bar{T}^1\cdot u^4n^2(\log n)^2)\leq O((\log r)\cdot u^4n^2(\log n)^2)$, and gate set $\cG$, and there exist subsets $\tilde{E}^1_{\mathrm{det},\tilde{t}}\subseteq\tilde{N}^1$ for $\tilde{t}\in\tilde{T}^1$, such that
  \begin{equation}
    \label{eq:tR1gad}
    \left(\tilde{\cR}^1,\; \tilde{\cE}^1_{\mathrm{run}}=2^{\tilde{N}^1}|_{\geq\tilde{\lambda}_{\mathrm{run}}}^{\sqcup\tilde{T}^1},\; \tilde{D}_{\mathrm{in}},\; (\tilde{D}_{\mathrm{in}}')^{\sqcup 2},\; \tilde{\cE}^1_{\mathrm{det}}=\bigsqcup_{\tilde{t}\in[\tilde{T}^1]}(2^{\tilde{E}^1_{\mathrm{det},\tilde{t}}}|_{\geq\tilde{\lambda}_{\mathrm{det}}})\right)
  \end{equation}
  is a mending fault-detecting gadget for $\tilde{\bar{O}}^1(\cdot)=\bar{\cR}^1(\cdot)$.
  Specifically, the two output blocks of $\tilde{D}_{\mathrm{in}}'$ above encode output dits $K_{\mathrm{in}}\times\{1\}$ and $K_{\mathrm{in}}\times\{2\}$ of $\bar{\cR}^1$, respectively. Here we may use $\tilde{D}_{\mathrm{in}}$ for the output blocks because $\tilde{\lambda}_{\mathrm{out}}\leq\tilde{\lambda}_{\mathrm{in}}-2\tilde{\lambda}_{\mathrm{run}}$ as shown above.

  Let $(\tilde{\cR}^1)'$ be the circuit that first runs $\tilde{\cR}^1$, and then has two additional timesteps: timestep $\tilde{T}^1+1$ applies identity gates, and timestep $\tilde{T}^1+2$ applies $\gTerm$ gates to all dits in $K_{\mathrm{in}}\times\{2\}$. Let $\tilde{E}^1_{\mathrm{det},\tilde{T}^1+1}=[n]^u\times\{2\}$ and $\tilde{E}^1_{\mathrm{det},\tilde{T}^1+2}=\emptyset$, where $[n]^u\times\{2\}\subseteq\tilde{N}^1$ denotes the second of the two output blocks of $\tilde{\cR}^1$.

  \begin{claim}
    The data
    \begin{equation}
      \label{eq:tR1pgad}
      \left((\tilde{\cR}^1)',\; (\tilde{\cE}^1)_{\mathrm{run}}=2^{\tilde{N}^1}|_{\geq\tilde{\lambda}_{\mathrm{run}}}^{\sqcup\tilde{T}^1+2},\; \tilde{D}_{\mathrm{in}},\; \tilde{D}_{\mathrm{in}},\; (\tilde{\cE}^1)'_{\mathrm{det}}=\bigsqcup_{\tilde{t}\in[\tilde{T}^1+2]}(2^{\tilde{E}^1_{\mathrm{det},\tilde{t}}}|_{\geq\tilde{\lambda}_{\mathrm{det}}})\right)
    \end{equation}
    forms a mending fault-detecting gadget (in the sense of \Cref{def:faultdetalph}) for the identity map $(\bar{O}^1)'=I_r:\bF_r^{K_{\mathrm{in}}}\rightarrow\bF_r^{K_{\mathrm{in}}}$ on $r$-ary dits labeled by $K_{\mathrm{in}}$.
  \end{claim}
  \begin{proof}
    In the absence of errors, the second output block of $\tilde{\cR}^1(x)$ on input $x\in\bF_q^{K_{\mathrm{in}}}$ contains an encoding of the values $x_j^r-x_j$ for every $j\in K_{\mathrm{in}}$. These values equal $0$ iff $x_j\in\bF_r$. Thus the gadget in \Cref{eq:tR1pgad} is fault-detecting because the gadget in \Cref{eq:tR1gad} is fault-detecting. Note that the output error threshold of $\tilde{D}_{\mathrm{in}}'$ is $\tilde{\lambda}_{\mathrm{in}}-2\tilde{\lambda}_{\mathrm{run}}$, and hence the two additional timesteps at the end of $(\tilde{\cR}^1)'$ bring the error threshold up to the error threshold $\tilde{\lambda}_{\mathrm{in}}$ of $D_{\mathrm{in}}$.

    To see that the gadget in \Cref{eq:tR1pgad} is mending, consider an arbitrary input $y\in\bF_q^{n^u}$ and a $\tilde{\cE}^1_{\mathrm{run}}$-avoiding fault $\cF$ for $(\tilde{\cR}^1)'$. By the mending property of the gadget in \Cref{eq:tR1gad}, either there exists $x'\in\bF_q^{K_{\mathrm{in}}}$ such that $(\tilde{\cR}^1)'[\cF](y)$ is a $\tilde{\cE}_{\mathrm{in}}$-deviation of $\Enc_{C(q,u,n,k,K)}(x')$, or else there exists a set $E\in\tilde{\cE}^1_{\mathrm{det}}$ such that $E\subseteq\supp(\tran(\tilde{\cR}^1[\cF];y))$. If the first case holds with $x'\in\bF_r^{K_{\mathrm{in}}}\subseteq\bF_q^{K_{\mathrm{in}}}$, or if the second case holds, then the desired mending property in the sense of \Cref{def:faultdetalph} holds.
    The only remaining possibility is if the first case holds, but for some $x'\in\bF_q^{K_{\mathrm{in}}}\setminus\bF_r^{K_{\mathrm{in}}}$.

    To analyze this possiblity, define $x''\in\bF_q^{K_{\mathrm{in}}}$ by $x''_j=(x_j')^r-x_j'$. Then $x''\neq 0$, and the second output block of $\tilde{\cR}^1(x')$ differs from $\Enc_{C(q,u,n,k,K)}(x'')$ by an error of weight $<\tilde{\lambda}_{\mathrm{in}}-2\tilde{\lambda}_{\mathrm{run}}$. As $\Enc_{C(q,u,n,k,K)}(x'')$ is a nonzero codeword of $C(q,u,n,k,K)$, its weight is at least the code's distance $(n-8k+1)^u$. Hence the support $E$ of the restriction of $\tran((\tilde{\cR}^1)'[\cF];y)$ to dits $\tilde{E}^1_{\mathrm{det},\tilde{T}^1+1}$ at timestep $\tilde{T}^1+1$ has weight at least
    \begin{equation*}
      |E|\geq(n-8k+1)^u-\tilde{\lambda}_{\mathrm{in}} \geq \frac{(n-8k+1)^u}{2} \geq \tilde{\lambda}_{\mathrm{det}},
    \end{equation*}
    where the first inequality above holds by \Cref{eq:mdtlams}, and the second inequality holds because by \Cref{eq:mdblams,eq:fdbarlamrun},
    \begin{align*}
      \tilde{\lambda}_{\mathrm{det}}
      &\leq (u-1)^2n^4\cdot\bar{\lambda}_{\mathrm{run}}(u,n) \leq \frac{(n-8k+1)^u}{2}.
    \end{align*}
    Thus $E\in(\tilde{\cE}^1)'_{\mathrm{det}}$, completing the proof of the mending property (in the sense of \Cref{def:faultdetalph}).
  \end{proof}
  
  Now we let $\bar{\cR}^2$ denote the circuit obtained from $\bar{\cR}$ by replacing every $r$-ary gate with the analogous $q$-ary gate. Therefore $\bar{\cR}^2$ acts on $q$-ary dits, but otherwise has the same space usage $\bar{N}\cong\bar{N}\times\{1\}\subseteq\bar{N}'$ and time usage $\bar{T}$ as $\bar{\cR}$. 
  Therefore by \Cref{it:fd} in \Cref{thm:fdft}, there exists a circuit $\tilde{\cR}^2$ using space $|\tilde{N}^2|\leq 3(u+1)n^u$, time $\tilde{T}^2\leq O(\bar{T}\cdot u^4n^2(\log n)^2)$, and gate set $\cG$, and there exist subsets $\tilde{E}^2_{\mathrm{det},\tilde{t}}\subseteq\tilde{N}^2$ for $\tilde{t}\in\tilde{T}^2$, such that
  \begin{equation}
    \label{eq:tR2gad}
    \left(\tilde{\cR}^2,\; \tilde{\cE}^2_{\mathrm{run}}=2^{\tilde{N}^2}|_{\geq\tilde{\lambda}_{\mathrm{run}}}^{\sqcup\tilde{T}^2},\; \tilde{D}_{\mathrm{in}},\; \tilde{D}_{\mathrm{out}},\; \tilde{\cE}^2_{\mathrm{det}}=\bigsqcup_{\tilde{t}\in[\tilde{T}^2]}(2^{\tilde{E}^2_{\mathrm{det},\tilde{t}}}|_{\geq\tilde{\lambda}_{\mathrm{det}}})\right)
  \end{equation}
  is a mending fault-detecting gadget for $\tilde{\bar{O}}^2(\cdot)=\bar{\cR}^2(\cdot):\bF_q^{K_{\mathrm{in}}}\rightarrow\bF_q^{K_{\mathrm{out}}}$. Letting $\iota:\bF_r\hookrightarrow\bF_q$ denote the inclusion, then by \Cref{lem:alphup} we have $\bar{\cR}^2\circ\iota^{\sqcup K_{\mathrm{in}}}=\iota^{\sqcup K_{\mathrm{out}}}\circ\bar{\cR}$, so the gadget above is therefore also a (non-mending) fault-detecting gadget for $\bar{\cR}(\cdot):\bF_r^{K_{\mathrm{in}}}\rightarrow\bF_r^{K_{\mathrm{out}}}$ in the sense of \Cref{def:faultdetalph}.

  \begin{claim}
    \label{claim:mdla}
    The sequential composition
    \begin{equation}
      \label{eq:mdlagad}
      \left(\tilde{\cR}=\tilde{\cR}^2\circ(\tilde{\cR}^1)',\; \tilde{\cE}_{\mathrm{run}}=\tilde{\cE}^2_{\mathrm{run}}\sqcup(\tilde{\cE}^1)_{\mathrm{run}},\; \tilde{D}_{\mathrm{in}},\; \tilde{D}_{\mathrm{out}},\; \tilde{\cE}_{\mathrm{det}}=\tilde{\cE}^2_{\mathrm{det}}\sqcup(\tilde{\cE}^1)'_{\mathrm{det}}\right)
    \end{equation}
    of the gadgets in \Cref{eq:tR1pgad,eq:tR2gad} forms a mending fault-detecting gadget for $\tilde{\bar{O}}(\cdot)=\bar{\cR}(\cdot):\bF_r^{K_{\mathrm{in}}}\rightarrow\bF_r^{K_{\mathrm{out}}}$, using space $|\tilde{N}|=\max\{|\tilde{N}^1|,|\tilde{N}^2|\}\leq 3(u+1)n^u$, time $\tilde{T}=\tilde{T}^1+2+\tilde{T}^2\leq O((\bar{T}+\log r)\cdot u^4n^2(\log n)^2)$, and gate set $\cG$ over $q$-ary dits.
  \end{claim}
  \begin{proof}
    The result holds by the same reasoning used to prove \Cref{lem:seqcomp}, but applied to gadgets in the form of \Cref{def:faultdetalph} in which the physical circuit over $q$-ary dits implements a function over smaller $r$-ary dits. The proof in this case is essentially the same as that of \Cref{lem:seqcomp}, so we omit the details to avoid redundancy.
  \end{proof}

  To complete the proof of \Cref{it:mainfd} in \Cref{thm:main}, we will apply \Cref{lem:alphdown} to reduce the alphabet size of the physical dits of $\tilde{\cR}$ in the gadget in \Cref{eq:mdlagad} from $q$ to $r$; the desired mending fault-detecting property over $r$-ary dits will follow from \Cref{claim:mdla}. Specifically, fix an $\bF_r$-linear isomorphism $\phi:\bF_r^\kappa\rightarrow\bF_q$ such that that first input component maps to the subfield $\bF_r\subseteq\bF_q$, that is, $\phi|_{\bF_r\times\{0\}^{\kappa-1}}=\iota$. For $K\subseteq\bar{N}=\bar{N}\times\{1\}\subseteq\bar{N}'\subseteq[k]^u$, we obtain the code $C(r,\epsilon,K,\bar{N})$ with encoding map $\Enc_{C(r,\epsilon,K,\bar{N})}$ in \Cref{it:mainfd} in \Cref{thm:main} by restricting $C(q,u,n,k,K)$ to messages in $\bF_r^K\subseteq\bF_q^K$ and applying $\phi^{-1}$ to all codeword dits, that is,
  \begin{align*}
    \Enc_{C(r,\epsilon,K,\bar{N})} &= (\phi^{-1})^{\sqcup[n]^u}\circ\Enc_{C(q,u,n,k,K)}\circ\iota^{\sqcup K} \\
    C(r,\epsilon,K,\bar{N}) &= \Enc_{C(r,\epsilon,K,\bar{N})}(\bF_r^K).
  \end{align*}
  Then $C(r,\epsilon,K,\bar{N})$ is a code over $r$-ary dits with length
  \begin{align*}
    n(r,\epsilon,\bar{N})
    &= \kappa\cdot n^u \leq \frac{\log\log|\bar{N}|}{\epsilon}\cdot M
  \end{align*}
  dimension $|K|$, and distance $d(r,\epsilon,\bar{N})$ at least the distance of $C(q,u,n,k,K)$, so that
  \begin{align*}
    d(r,\epsilon,\bar{N})
    &\geq (n-8k+1)^u \geq M\cdot\left(1-\frac{8}{n^\epsilon}\right)^{n^\epsilon} \geq \frac{M}{2^{32}}.
  \end{align*}

  Now we define a circuit $\cR$ on $r$-ary dits using space $N=[\kappa]\times\tilde{N}$ and time $T=\kappa^3\cdot\tilde{T}$ as follows: for each $q$-ary gate $G\in\cG$ acting on some set of dits $B\subseteq\tilde{N}$ in some timestep $\tilde{t}\in[\tilde{T}]$ of $\tilde{\cR}$, we let $\cR$ apply the corresponding $r$-ary circuit $\cR_G$ from \Cref{lem:alphdown} to dits $[\kappa]\times B$ in timesteps $\{\kappa^3(\tilde{t}-1)+1,\dots,\kappa^3\tilde{t}\}\subseteq[T]$.

  By construction, $\tilde{E}_{\mathrm{det}}$ is of the form $\tilde{E}_{\mathrm{det}}=\bigsqcup_{\tilde{t}\in[t]}(2^{\tilde{E}_{\mathrm{det},\tilde{t}}}|_{\geq\tilde{\lambda}_{\mathrm{det}}})$ for some subsets $\tilde{E}_{\mathrm{det},\tilde{t}}\subseteq\tilde{N}$ for $\tilde{t}\in[\tilde{T}]$. Then for $t\in[T]$, we define
  \begin{equation}
    \label{eq:mdEdet}
    E_{\mathrm{det},t} = \begin{cases}
      [\kappa]\times\tilde{E}_{\mathrm{det},t/\kappa^3},&t\in\kappa^3\cdot[\tilde{T}]\\
      \emptyset,&t\notin\kappa^3\cdot[\tilde{T}].
    \end{cases}
  \end{equation}

  Assuming $|\bar{N}|\geq\bar{N}_0(r,\epsilon)$ is sufficiently large, then by the bounds on $|\tilde{N}|,\tilde{T}$ in \Cref{claim:mdla}, the space and time usages of $\cR$ satisfy
  \begin{align*}
    |N|
    &= \kappa\cdot|\tilde{N}| \\
    &\leq \kappa\cdot 3(u+1)n^u \\
    &\leq \frac{1}{\epsilon}(\log\log|\bar{N}|) \cdot \log|\bar{N}| \cdot M \\
    &\leq (\log|\bar{N}|)^2 \cdot M \\
    T 
    &= \kappa^3\cdot\tilde{T} \\
    &\leq O(\kappa^3\cdot(\bar{T}+\log r)\cdot u^4n^2(\log n)^2) \\
    &\leq \bar{T}\cdot(\log|\bar{N}|)^{8/\epsilon},
  \end{align*}
  where the above inequalities apply the bounds in \Cref{eq:mdnvars}.

  We will now complete the proof of \Cref{it:mainfd} in \Cref{thm:main} by showing the following.
  
  \begin{claim}
    \label{claim:mdfinal}
    The data
    \begin{equation}
      \label{eq:mdgad}
      \left(\cR,\; \cE_{\mathrm{run}}=2^N|_{\geq\lambda_{\mathrm{run}}},\; D_{\mathrm{in}},\; D_{\mathrm{out}},\; \cE_{\mathrm{det}}=\bigsqcup_{t\in[T]}(2^{E_{\mathrm{det},t}}|_{\geq\lambda_{\mathrm{det}}})\right)
    \end{equation}
    forms a mending fault-detecting gadget for $\bar{\cR}(\cdot)$, where $D_{\mathrm{in}},D_{\mathrm{out}}$ and $\lambda_{\mathrm{run}},\lambda_{\mathrm{in}},\lambda_{\mathrm{out}},\lambda_{\mathrm{det}}$ are defined as in \Cref{it:mainfd} in \Cref{thm:main}.
  \end{claim}

  To prove \Cref{claim:mdfinal}, we will use the following bounds relating the variables defined in \Cref{eq:mdtlams} to those defined in \Cref{eq:mdlams}, where below we assume $|\bar{N}|\geq\bar{N}_0(r,\epsilon)$ is sufficiently large:
  \begin{align}
    \label{eq:mdlamcomp}
    \begin{split}
      \tilde{\lambda}_{\mathrm{in}}
      &\geq \frac{M}{4}\cdot\left(1-\frac{8}{n^\epsilon}\right)^{n^\epsilon} \geq \frac{M}{2^{32}} = \lambda_{\mathrm{in}} \\
      \tilde{\lambda}_{\mathrm{out}}
      &= 16u^2n^4\cdot(un)^\eta\left(\frac{n}{n-8k+1}\right)^{u\cdot\eta} \cdot \tilde{\lambda}_{\mathrm{run}} \\
      &\leq 16n^{10\eta}\cdot\left(1+\frac{10}{n^\epsilon}\right)^{n^\epsilon\cdot\eta} \cdot \tilde{\lambda}_{\mathrm{run}} \\
      &\leq n^{12\eta} \cdot (3\kappa^3\cdot\lambda_{\mathrm{run}}) \\
      &\leq \frac{(\log|\bar{N}|)^{16\eta/\epsilon}\cdot\lambda_{\mathrm{run}}}{\kappa} \\
      &= \frac{\lambda_{\mathrm{out}}}{\kappa} \\
      \tilde{\lambda}_{\mathrm{det}}
      &= (u-1)^un^4\cdot(3\kappa^3\cdot\lambda_{\mathrm{run}}) \geq \lambda_{\mathrm{run}} = \lambda_{\mathrm{det}}.
    \end{split}
  \end{align}
  
  \begin{proof}[Proof of \Cref{claim:mdfinal}]
    We begin with the proof of the (non-mending) fault-detecting property. First, \Cref{it:fdnoerr} in \Cref{def:faultdet} holds for the gadget in \Cref{eq:mdgad} by \Cref{claim:mdla,lem:alphdown}. Specifically, in the absence of errors, then $\cR$ simply runs $\tilde{\cR}$, but expresses each $q$-ary dit using $\kappa$ $r$-ary dits according to the isomorphism $\phi$, and spends $\kappa^3$ timesteps simulating each $q$-ary gate $G$ using a circuit $\cR_G$ of $r$-ary gates. At every timestep $t$ that is a multiple of $\kappa^3$, so that a layer of $q$-ary gates has just finished being simulated, the detector set $E_{\mathrm{det},t}$ simply contains all $\kappa$ of the $r$-ary dits associated to every $q$-ary dit in the corresponding detector set $\tilde{E}_{\mathrm{det},t/\kappa^3}$. Therefore when we run $\cR$ on input $y=\Enc_{C(r,\epsilon,K_{\mathrm{in}},\bar{N})}(x)$ for $x\in\bF_r^{K_{\mathrm{in}}}$, every dit in every detector set $E_{\mathrm{det},t}$ has value~$0$ (because when running $\tilde{\cR}$ on input $\Enc_{C(q,u,n,k,K_{\mathrm{in}})}(x)$, every dit in every detector set $\tilde{E}_{\mathrm{det},t/\kappa^3}$ must have value~$0$), and we must have
    \begin{align*}
      \cR(y)
      &= (\phi^{-1})^{\sqcup[n]^u}\circ\tilde{\cR}\circ\Enc_{C(q,u,n,k,K_{\mathrm{in}})}(x) \\
      &= (\phi^{-1})^{\sqcup[n]^u}\circ\Enc_{C(q,u,n,k,K_{\mathrm{out}})}\circ\bar{\cR}_q(x) \\
      &= \Enc_{C(r,\epsilon,K_{\mathrm{out}},\bar{N})}\circ\bar{\cR}_r(x).
    \end{align*}

    We now similarly show that \Cref{it:fderr} in \Cref{def:faultdet} holds for the gadget in \Cref{eq:mdgad} by using the fact that this property holds for the gadget in \Cref{eq:mdlagad}. Let $x\in\bF_r^{K_{\mathrm{in}}}$, let $y$ be an input to $\cR$ that differs from $\Enc_{C(r,\epsilon,K_{\mathrm{in}},\bar{N})}(x)$ by an error of weight $<\lambda_{\mathrm{in}}$, and let $\cF$ be a fault with weight $<\lambda_{\mathrm{run}}$ in each timestep. We want to show that either $\cR[\cF](y)$ differs from $\Enc_{C(r,\epsilon,K_{\mathrm{out}},\bar{N})}\circ\bar{\cR}(x)$ by an error of weight $<\lambda_{\mathrm{out}}$, or else that there exists some $t\in[T]$ such that $|\tran(\cR[\cF];y)|_{E_{\mathrm{det},t}\times\{t\}}|\geq\lambda_{\mathrm{det}}$.

    By definition, $\tilde{y}:=\phi^{\sqcup[n]^u}(y)$ differs from $\Enc_{C(q,u,n,k,K_{\mathrm{in}})}(x)=\phi^{\sqcup [n]^u}\circ\Enc_{C(r,\epsilon,K_{\mathrm{in}},\bar{N})}(x)$ at $<\lambda_{\mathrm{in}}$ dits, and $\lambda_{\mathrm{in}}\leq\tilde{\lambda}_{\mathrm{in}}$ by \Cref{eq:mdlamcomp}. Furthermore, $\cR[\cF](y)$ simulates the execution of $\tilde{\cR}[\tilde{\cF}](\tilde{y})$ for an appropriate fault $\tilde{\cF}$. Specifically, for $\tilde{t}\in[\tilde{T}]$, we have
    \begin{equation}
      \label{eq:mdtranrestricts}
      \tran(\cR[\cF];y)|_{N\times\{\kappa^3\cdot\tilde{t}\}} = \phi^{\sqcup\tilde{N}}(\tran(\tilde{\cR}[\tilde{\cF}];\tilde{y})|_{\tilde{N}\times\{\tilde{t}\}}),
    \end{equation}
    where $\tilde{\cF}$ is an appropriate fault on $\tilde{\cR}$ that simulates the effect of the fault $\cF$ on $\cR$. For every gate $G$ acting on some dits $B\subseteq\tilde{N}$ in some timestep $\tilde{t}\in[\tilde{T}]$ of $\tilde{\cR}$, if $\cF$ has nontrivial support on the restriction to the associated subcircuit $\cR_G$ of $\cR$, then at timestep $\tilde{t}$ we let $\tilde{\cF}$ set the value of dits $B$ to equal $(\phi^{-1})^{\sqcup 3}(\tran(\cR[\cF];y)|_{[\kappa]\times B\times\{\kappa^3\cdot\tilde{t}\}})$.
    If instead $\cF$ has no support on the restriction to this subcircuit $\cR_G$ of $\cR$, then $\tilde{\cF}$ has no support on dits $B$ in timestep $\tilde{t}$. Because every gate $G\in\cG$ in our gate set has $|B|\leq 3$, for each $\tilde{t}\in[\tilde{T}]$, each of the $\kappa^3$ timesteps $\{\kappa^3(\tilde{t}-1)+1,\dots,\kappa^3\tilde{t}\}$ of~$\cR[\cF]$ can induce errors on up to $3\lambda_{\mathrm{run}}$ dits in $\tilde{\cR}[\tilde{\cF}]$. Therefore $\tilde{\cF}$ has weight $<\tilde{\lambda}_{\mathrm{run}}=3\kappa^3\cdot\lambda_{\mathrm{run}}$ in timestep $\tilde{t}$ (see \Cref{eq:mdtlams}).

    Thus by \Cref{claim:mdla}, either $\tilde{\cR}[\tilde{\cF}](\tilde{y})$ differs from
    \begin{equation*}
      \Enc_{C(q,u,n,k,K_{\mathrm{out}})}\circ\bar{\cR}(x) = \phi^{\sqcup [n]^u}\circ\Enc_{C(r,\epsilon,K_{\mathrm{out}},\bar{N})}\circ\bar{\cR}(x)
    \end{equation*}
    at $<\tilde{\lambda}_{\mathrm{out}}$ dits, or else there exists $\tilde{t}\in[\tilde{T}]$ such that $|\tran(\tilde{\cR}[\tilde{\cF}];\tilde{y})|_{\tilde{E}_{\mathrm{det},\tilde{t}}\times\{\tilde{t}\}}|\geq\tilde{\lambda}_{\mathrm{det}}$. In the first case, by \Cref{eq:mdtranrestricts,eq:mdlamcomp}, we conclude that $\cR[\cF](y)$ differs from $\Enc_{C(r,\epsilon,K_{\mathrm{out}},\bar{N})}\circ\bar{\cR}(x)$ at $<\kappa\cdot\tilde{\lambda}_{\mathrm{out}}\leq\lambda_{\mathrm{out}}$ dits. In the latter case, by \Cref{eq:mdtranrestricts,eq:mdlamcomp}, then there exists $t=\kappa^3\cdot\tilde{t}\in[T]$ such that $|\tran(\cR[\cF];y)|_{K_{\mathrm{det},t}\times\{t\}}|\geq\tilde{\lambda}_{\mathrm{det}}\geq\lambda_{\mathrm{det}}$. Thus \Cref{it:fderr} in \Cref{def:faultdet} holds for the gadget in \Cref{eq:mdgad}.

    We have now shown that the gadget in \Cref{eq:mdgad} is fault-detecting; we will next show that it is mending. The proof is similar as the proof of fault-detecting above. For an arbitrary input $y\in\bF_r^{[\kappa]\times[n]^u}$ and a fault $\cF$ with weight $<\lambda_{\mathrm{run}}$ in each timestep, we want to show that either there exists $x'\in\bF_r^{K_{\mathrm{in}}}$ such that $\cR[\cF](y)$ differs from $\Enc_{C(r,\epsilon,K_{\mathrm{out}},\bar{N})}\circ\bar{\cR}(x')$ at $<\lambda_{\mathrm{out}}$ dits, or else that there exists some $t\in[T]$ such that $|\tran(\cR[\cF];y)|_{E_{\mathrm{det},t}\times\{t\}}|\geq\lambda_{\mathrm{det}}$. As in the non-mending case above, we again let $\tilde{y}:=\phi^{\sqcup[n]^u}(y)$, and we construct a fault $\tilde{\cF}$ with weight $<\tilde{\lambda}_{\mathrm{run}}$ in each timestep, such that \Cref{eq:mdtranrestricts} holds. By \Cref{claim:mdla}, the gadget in \Cref{eq:mdlagad} is mending (in the sense of \Cref{def:faultdetalph}), so either there exists $x'\in\bF_r^{K_{\mathrm{in}}}$ such that $\tilde{\cR}[\tilde{\cF}](\tilde{y})$ differs from $\Enc_{C(q,u,n,k,K_{\mathrm{out}})}\circ\bar{\cR}(x')=\phi^{\sqcup [n]^u}\circ\Enc_{C(r,\epsilon,K_{\mathrm{out}},\bar{N})}\circ\bar{\cR}(x')$ at $<\tilde{\lambda}_{\mathrm{out}}$ dits, or else there exists $\tilde{t}\in[\tilde{T}]$ such that $|\tran(\tilde{\cR}[\tilde{\cF}];\tilde{y})|_{\tilde{E}_{\mathrm{det},\tilde{t}}\times\{\tilde{t}\}}|\geq\tilde{\lambda}_{\mathrm{det}}$. In the first case, by \Cref{eq:mdtranrestricts,eq:mdlamcomp}, we conclude that $\cR[\cF](y)$ differs from $\Enc_{C(r,\epsilon,K_{\mathrm{out}},\bar{N})}\circ\bar{\cR}(x')$ at $<\kappa\cdot\tilde{\lambda}_{\mathrm{out}}\leq\lambda_{\mathrm{out}}$ dits. In the latter case, by \Cref{eq:mdtranrestricts,eq:mdlamcomp}, then there exists $t=\kappa^3\cdot\tilde{t}\in[T]$ such that $|\tran(\cR[\cF];y)|_{K_{\mathrm{det},t}\times\{t\}}|\geq\tilde{\lambda}_{\mathrm{det}}\geq\lambda_{\mathrm{det}}$. Thus the gadget in \Cref{eq:mdgad} is mending, as desired.
  \end{proof}
\end{proof}

\subsection{Fault-Tolerant Gadget}
\label{sec:mainft}
In this section, we prove \Cref{it:mainft} of \Cref{thm:main}. The proof is similar as the proof of \Cref{it:mainfd} of \Cref{thm:main} in \Cref{sec:mainfd} above, except we now will apply \Cref{it:ft} (instead of \Cref{it:fd}) of \Cref{thm:fdft}. The main differences in the proof are therefore that the values of parameters change (due to the differences in parameters between \Cref{it:fd} and \Cref{it:ft} of \Cref{thm:fdft}), and that we do not prove the mending property here (so some steps are simplified).

In the proof below, we will use the following basic fact:

\begin{fact}
  \label{fact:pow32}
  For every $a\geq 0$ and every $0\leq\Delta a\leq a$, then
  \begin{equation*}
    a^{3/2}+\frac32a^{1/2}\Delta a \leq (a+\Delta a)^{3/2} \leq a^{3/2}+3a^{1/2}\Delta a.
  \end{equation*}
\end{fact}
\begin{proof}
  Letting $f(a)=a^{3/2}$, then $f'(a)=(3/2)a^{1/2}\geq 0$ and $f''(a)=(3/4)a^{-1/2}\geq 0$. Therefore $f(a+\Delta a)\geq f(a)+f'(a)\Delta a$ and $f(a+\Delta a)\leq f(a)+f'(a+\Delta a)\Delta a\leq f(a)+f'(2a)\Delta a$, from which the desired inequalities follow.
\end{proof}

\begin{proof}[Proof of \Cref{it:mainft} in \Cref{thm:main}]
  At various points throughout the proof below, we will require that $\bar{N}_0(r)$ be sufficiently large compared to $r$ so that certain bounds on $|\bar{N}|$ that hold asymptotically in fact hold for every $|\bar{N}|\geq\bar{N}_0(r)$. The final value of $\bar{N}_0(r)$ will then implicitly be given by the maximum of the values required in these various instances.
  
  Given the set $\bar{N}$, we define\footnote{Note that the variable $n$ here is distinct from $n(r,\bar{N})$ in the statement of \Cref{thm:main}.}
  \begin{align}
    \label{eq:mtnvars}
    \begin{split}
      n &= 2^{\lfloor(\log|\bar{N}|)^{2/3}+4(\log|\bar{N}|)^{1/3}\rfloor} \\
      k &= \lfloor n/32\rfloor \\
      u &= \lfloor\sqrt{\log n}\rfloor \\
      \kappa &= \lceil\log(n)/\log(r)\rceil \\
      q &= r^\kappa.
    \end{split}
  \end{align}
  We let
  \begin{align*}
    M &= M(\bar{N}) = n^u,
  \end{align*}
  so that if $|\bar{N}|\geq\bar{N}_0(r)$ is sufficiently large, then
  \begin{align}
    \label{eq:Mkulower}
    \begin{split}
      M \geq k^u
      &\geq 2^{((\log|\bar{N}|)^{2/3}+2(\log|\bar{N}|)^{1/3})\cdot(\sqrt{(\log|\bar{N}|)^{2/3}+2(\log|\bar{N}|)^{1/3}}-1)} \\
      &= 2^{((\log|\bar{N}|)^{2/3}+2(\log|\bar{N}|)^{1/3})^{3/2}-((\log|\bar{N}|)^{2/3}+2(\log|\bar{N}|)^{1/3})} \\
      &\geq 2^{(\log|\bar{N}|+3(\log|\bar{N}|)^{2/3})-3(\log|\bar{N}|)^{2/3}} \\
      &= |\bar{N}|,
    \end{split}
  \end{align}
  where the second inequality above holds by \Cref{fact:pow32}.
  Similarly, again assuming $|\bar{N}|\geq\bar{N}_0(r)$ is sufficiently large,
  \begin{align*}
    M
    &\leq 2^{((\log|\bar{N}|)^{2/3}+4(\log|\bar{N}|)^{1/3})\cdot\sqrt{(\log|\bar{N}|)^{2/3}+4(\log|\bar{N}|)^{1/3}}} \\
    &= 2^{((\log|\bar{N}|)^{2/3}+4(\log|\bar{N}|)^{1/3})^{3/2}} \\
    &\leq 2^{\log|\bar{N}|+12(\log|\bar{N}|)^{2/3}} \\
    &= |\bar{N}|\cdot 2^{12(\log|\bar{N}|)^{2/3}},
  \end{align*}
  where the third inequality above holds by \Cref{fact:pow32}. Thus $M\in[|\bar{N}|,\;2^{12(\log|\bar{N}|)^{2/3}}\cdot|\bar{N}|]$, as desired.

  Meanwhile because $k^u\geq|\bar{N}|$ by \Cref{eq:Mkulower}, we may fix an arbitrary set inclusion $\bar{N}\hookrightarrow[k]^u$, so that we may view $\bar{N}$ as a subset of $[k]^u$. Also if $|\bar{N}|\geq\bar{N}_0(r)$ is sufficiently large, then we have $u\geq 4$, $k\leq n/32$, and $q\geq n$. Thus for $K\subseteq\bar{N}'\subseteq[k]^u$, the inequalities above ensure that the code $C(q,u,n,k,K)$ with encoding map $\Enc_{C(q,u,n,k,K)}$ in \Cref{thm:fdft} is well-defined.
  By \Cref{thm:fdft}, $C(q,u,n,k,K)$ has parameters $[n^u,|K|,\geq(n-8k+1)^u]_q$.

  To prove \Cref{it:mainft} in \Cref{thm:main}, we will first construct fault-tolerant gadgets over the alphabet~$\bF_q$, before reducing the alphabet to $\bF_r$ using \Cref{lem:alphdown}. Specifically, let $\eta_{\Dec}$ be the absolute constant defined in \Cref{it:ft} in \Cref{thm:fdft}, and let
  \begin{align}
    \label{eq:mttlams}
    \begin{split}
      \tilde{\lambda}_{\mathrm{run}} &= 3\kappa^3\cdot\lambda_{\mathrm{run}} \\
      \tilde{\lambda}_{\mathrm{in}} &= \left(\frac{n-8k+1}{8}\right)^u \\
      \tilde{\lambda}_{\mathrm{out}} &= \eta_{\Dec}\cdot 2^{u(u+8)}\cdot n^3\log(q)\cdot\tilde{\lambda}_{\mathrm{run}}.
    \end{split}
  \end{align}
  For $\alpha\in\{\mathrm{in},\mathrm{out}\}$ let
  \begin{align*}
    \tilde{D}_\alpha &= (C(q,u,n,k,K_\alpha),\; \Enc_{C(q,u,n,k,K_\alpha)},\; 2^{[n]^u}|_{\geq\tilde{\lambda}_\alpha}).
  \end{align*}

  Assuming $|\bar{N}|\geq\bar{N}_0(r)$ is sufficiently large, the expression $\bar{\lambda}_{\mathrm{run}}(u,n)$ in \Cref{eq:ftbarlamrun} satisfies
  \begin{align*}
    \label{eq:mtblams}
    \begin{split}
      \bar{\lambda}_{\mathrm{run}}(u,n)
      &= 3\kappa^3\cdot\frac{M}{\eta_{\Dec}\cdot 3\kappa^3\cdot\log(q)\cdot 2^{u(u+16)}\cdot n^{16}} \\
      &\geq 3\kappa^3\cdot\frac{M}{n^{20}} \\
      &\geq 3\kappa^3\cdot\frac{M}{2^{32(\log|\bar{N}|)^{2/3}}} \\
      &= 3\kappa^3\cdot\bar{\lambda}_{\mathrm{run}}(r,\bar{N}),
    \end{split}
  \end{align*}
  where the inequalities above hold by \Cref{eq:mtnvars}, and the final equality holds by \Cref{eq:mtbarlamrun}. Then because $\tilde{\lambda}_{\mathrm{run}}=3\kappa^3\cdot\lambda_{\mathrm{run}}\leq 3\kappa^3\cdot\bar{\lambda}_{\mathrm{run}}(r,\bar{N})$, we have $\tilde{\lambda}_{\mathrm{run}}\leq\bar{\lambda}_{\mathrm{run}}(u,n)$.

  Let $\bar{\cR}_r=\bar{\cR}$, and let $\bar{\cR}_q$ denote the circuit obtained from $\bar{\cR}_r$ by replacing every $r$-ary gate in $\cG$ with the analogous $q$-ary gate. Therefore $\bar{\cR}_q$ acts on $q$-ary dits, but otherwise has the same space usage $\bar{N}$ and time usage $\bar{T}$ as $\bar{\cR}_r$. Letting $\iota:\bF_r\hookrightarrow\bF_q$ denote the inclusion, then by \Cref{lem:alphup} we have $\bar{\cR}_q\circ\iota^{\sqcup K_{\mathrm{in}}}=\iota^{\sqcup K_{\mathrm{out}}}\circ\bar{\cR}_r$. That is, when restricting inputs to the subfield $\bF_r\subseteq\bF_q$, then $\bar{\cR}_q$ has the same output as $\bar{\cR}_r$.

  Now by \Cref{it:ft} in \Cref{thm:fdft}, there exists a circuit $\tilde{\cR}$ using space $|\tilde{N}|\leq 3\eta_{\Dec}\cdot\log(q)\cdot n^{u+1}$, time $\tilde{T}\leq O(\bar{T}\cdot u^4n^2(\log n)^2(\log q))$, and gate set $\cG$ such that
  \begin{equation}
    \label{eq:mtlagad}
    \left(\tilde{\cR},\; \tilde{\cE}_{\mathrm{run}}=2^{\tilde{N}}|_{\geq\tilde{\lambda}_{\mathrm{run}}}^{\sqcup\tilde{T}},\; \tilde{D}_{\mathrm{in}},\; \tilde{D}_{\mathrm{out}}\right)
  \end{equation}
  is a fault-tolerant gadget for $\tilde{\bar{O}}(\cdot)=\bar{\cR}_q(\cdot)$.

  We now apply \Cref{lem:alphdown} to reduce the alphabet size of the physical dits of $\tilde{\cR}$ in the gadget in \Cref{eq:mdlagad} from $q$ to $r$. Specifically, fix an $\bF_r$-linear isomorphism $\phi:\bF_r^\kappa\rightarrow\bF_q$ such that that first input component maps to the subfield $\bF_r\subseteq\bF_q$, that is, $\phi|_{\bF_r\times\{0\}^{\kappa-1}}=\iota$. For $K\subseteq\bar{N}\subseteq[k]^u$, we obtain the code $C(r,K,\bar{N})$ with encoding map $\Enc_{C(r,K,\bar{N})}$ in \Cref{it:mainft} in \Cref{thm:main} by restricting $C(q,u,n,k,K)$ to messages in $\bF_r^K\subseteq\bF_q^K$ and applying $\phi^{-1}$ to all codeword dits, that is,
  \begin{align*}
    \Enc_{C(r,K,\bar{N})} &= (\phi^{-1})^{\sqcup[n]^u}\circ\Enc_{C(q,u,n,k,K)}\circ\iota^{\sqcup K} \\
    C(r,K,\bar{N}) &= \Enc_{C(r,K,\bar{N})}(\bF_r^K).
  \end{align*}
  Then $C(r,K,\bar{N})$ is a code over $r$-ary dits with length
  \begin{align*}
    n(r,\bar{N})
    &= \kappa\cdot n^u \leq (\log|\bar{N}|)\cdot M
  \end{align*}
  dimension $|K|$, and distance $d(r,\bar{N})$ at least the distance of $C(q,u,n,k,K)$, so that
  \begin{align*}
    d(r,\bar{N})
    &\geq (n-8k+1)^u \\
    &\geq M\cdot(1/2)^u \\
    &\geq \frac{M}{2^{2(\log|\bar{N}|)^{1/3}}}.
  \end{align*}
  The above inequalities hold by \Cref{eq:mtnvars} assuming that $|\bar{N}|\geq\bar{N}_0(r)$ is sufficiently large.

  Now we define a circuit $\cR$ on $r$-ary dits using space $N=[\kappa]\times\tilde{N}$ and time $T=\kappa^3\cdot\tilde{T}$ as follows: for each $q$-ary gate $G\in\cG$ acting on some set of dits $B\subseteq\tilde{N}$ in some timestep $\tilde{t}\in[\tilde{T}]$ of $\tilde{\cR}$, we let $\cR$ apply the corresponding $r$-ary circuit $\cR_G$ from \Cref{lem:alphdown} to dits $[\kappa]\times B$ in timesteps $\{\kappa^3(\tilde{t}-1)+1,\dots,\kappa^3\tilde{t}\}\subseteq[T]$.

  Assuming $|\bar{N}|\geq\bar{N}_0(r)$ is sufficiently large, the space and time usages of $\cR$ satisfy
  \begin{align*}
    |N|
    &= \kappa\cdot|\tilde{N}| \\
    &\leq 3\eta_{\Dec}\cdot\log(q)\cdot n^{u+1} \\
    &\leq 2^{2(\log|\bar{N}|)^{2/3}}\cdot M \\
    T
    &= \kappa^3\cdot\tilde{T} \\
    &\leq O(\kappa^3\cdot\bar{T}\cdot u^4n^2(\log n)^2(\log q)) \\
    &\leq \bar{T}\cdot 2^{4(\log|\bar{N}|)^{2/3}},
  \end{align*}
  where above we apply the bounds in \Cref{eq:mtnvars}.

  We will now complete the proof of \Cref{it:mainft} in \Cref{thm:main} by showing the following.
  
  \begin{claim}
    \label{claim:mtfinal}
    The data
    \begin{equation}
      \label{eq:mtgad}
      \left(\cR,\; \cE_{\mathrm{run}}=2^N|_{\geq\lambda_{\mathrm{run}}},\; D_{\mathrm{in}},\; D_{\mathrm{out}}\right)
    \end{equation}
    forms a fault-tolerant gadget for $\bar{\cR}(\cdot)=\bar{\cR}_r(\cdot)$, where $D_{\mathrm{in}},D_{\mathrm{out}}$ and $\lambda_{\mathrm{run}},\lambda_{\mathrm{in}},\lambda_{\mathrm{out}}$ are defined as in \Cref{it:mainft} in \Cref{thm:main}.
  \end{claim}

  To prove \Cref{claim:mtfinal}, we will use the following bounds relating the variables defined in \Cref{eq:mttlams} to those defined in \Cref{eq:mtlams}, where below we assume $|\bar{N}|\geq\bar{N}_0(r)$ is sufficiently large:
  \begin{align}
    \label{eq:mtlamcomp}
    \begin{split}
      \tilde{\lambda}_{\mathrm{in}}
      &\geq \left(\frac{n}{16}\right)^u = \frac{M}{2^{4u}} \geq \frac{M}{2^{8(\log|\bar{N}|)^{1/3}}} = \lambda_{\mathrm{in}} \\
      \tilde{\lambda}_{\mathrm{out}}
      &= \frac{\eta_{\Dec}\cdot 2^{u(u+8)}\cdot n^3\log(q)\cdot 3\kappa^4\cdot\lambda_{\mathrm{run}}}{\kappa} \\
      &\leq \frac{n^5\cdot\lambda_{\mathrm{run}}}{\kappa} \\
      &\leq \frac{2^{8(\log|\bar{N}|)^{2/3}}\cdot\lambda_{\mathrm{run}}}{\kappa} \\
      &= \frac{\lambda_{\mathrm{out}}}{\kappa}.
    \end{split}
  \end{align}
  
  \begin{proof}[Proof of \Cref{claim:mtfinal}]
    The proof is similar as the proof of \Cref{claim:mdfinal} above, except here we show fault-tolerance instead of fault-detection.

    First, \Cref{it:ftnoerr} in \Cref{def:faulttol} holds for the gadget in \Cref{eq:mtgad} by the fault-tolerance of the gadget in \Cref{eq:mtlagad} along with \Cref{lem:alphdown}. Specifically, in the absence of errors, then $\cR$ simply runs $\tilde{\cR}$, but expresses each $q$-ary dit using $\kappa$ $r$-ary dits according to the isomorphism $\phi$, and spends $\kappa^3$ timesteps simulating each $q$-ary gate $G$ using a circuit $\cR_G$ of $r$-ary gates. Therefore when we run $\cR$ on input $y=\Enc_{C(r,K_{\mathrm{in}},\bar{N})}(x)$ for $x\in\bF_r^{K_{\mathrm{in}}}$, we must have
    \begin{align*}
      \cR(y)
      &= (\phi^{-1})^{\sqcup[n]^u}\circ\tilde{\cR}\circ\Enc_{C(q,u,n,k,K_{\mathrm{in}})}(x) \\
      &= (\phi^{-1})^{\sqcup[n]^u}\circ\Enc_{C(q,u,n,k,K_{\mathrm{out}})}\circ\bar{\cR}_q(x) \\
      &= \Enc_{C(r,K_{\mathrm{out}},\bar{N})}\circ\bar{\cR}_r(x).
    \end{align*}

    We now similarly show that \Cref{it:fterr} in \Cref{def:faulttol} holds for the gadget in \Cref{eq:mtgad} by using the fact that this property holds for the gadget in \Cref{eq:mtlagad}. Let $x\in\bF_r^{K_{\mathrm{in}}}$, let $y$ be an input to $\cR$ that differs from $\Enc_{C(r,K_{\mathrm{in}},\bar{N})}(x)$ by an error of weight $<\lambda_{\mathrm{in}}$, and let $\cF$ be a fault with weight $<\lambda_{\mathrm{run}}$ in each timestep. We want to show that $\cR[\cF](y)$ differs from $\Enc_{C(r,K_{\mathrm{out}},\bar{N})}\circ\bar{\cR}(x)$ by an error of weight $<\lambda_{\mathrm{out}}$.

    By definition, $\tilde{y}:=\phi^{\sqcup[n]^u}(y)$ differs from $\Enc_{C(q,u,n,k,K_{\mathrm{in}})}(x)=\phi^{\sqcup [n]^u}\circ\Enc_{C(r,K_{\mathrm{in}},\bar{N})}(x)$ at $<\lambda_{\mathrm{in}}$ dits, and $\lambda_{\mathrm{in}}\leq\tilde{\lambda}_{\mathrm{in}}$ by \Cref{eq:mtlamcomp}. Furthermore, $\cR[\cF](y)$ simulates the execution of $\tilde{\cR}[\tilde{\cF}](\tilde{y})$ for an appropriate fault $\tilde{\cF}$. Specifically, for $\tilde{t}\in[\tilde{T}]$, we have
    \begin{equation}
      \label{eq:mttranrestricts}
      \tran(\cR[\cF];y)|_{N\times\{\kappa^3\cdot\tilde{t}\}} = \phi^{\sqcup\tilde{N}}(\tran(\tilde{\cR}[\tilde{\cF}];\tilde{y})|_{\tilde{N}\times\{\tilde{t}\}}),
    \end{equation}
    where $\tilde{\cF}$ is an appropriate fault on $\tilde{\cR}$ that simulates the effect of the fault $\cF$ on $\cR$. For every gate $G$ acting on some dits $B\subseteq\tilde{N}$ in some timestep $\tilde{t}\in[\tilde{T}]$ of $\tilde{\cR}$, if $\cF$ has nontrivial support on the restriction to the associated subcircuit $\cR_G$ of $\cR$, then at timestep $\tilde{t}$ we let $\tilde{\cF}$ set the value of dits $B$ to equal $(\phi^{-1})^{\sqcup 3}(\tran(\cR[\cF];y)|_{[\kappa]\times B\times\{\kappa^3\cdot\tilde{t}\}})$. If instead $\cF$ has no support on the restriction to this subcircuit $\cR_G$ of $\cR$, then $\tilde{\cF}$ has no support on dits $B$ in timestep $\tilde{t}$. Because every gate $G\in\cG$ in our gate set has $|B|\leq 3$, for each $\tilde{t}\in[\tilde{T}]$, each of the $\kappa^3$ timesteps $\{\kappa^3(\tilde{t}-1)+1,\dots,\kappa^3\tilde{t}\}$ of $\cR[\cF]$ can induce errors on up to $3\lambda_{\mathrm{run}}$ dits in $\tilde{\cR}[\tilde{\cF}]$. Therefore $\tilde{\cF}$ has weight $<\tilde{\lambda}_{\mathrm{run}}=3\kappa^3\cdot\lambda_{\mathrm{run}}$ in timestep $\tilde{t}$ (see \Cref{eq:mttlams}).

    Thus by the fault-tolerance of the gadget in \Cref{eq:mtlagad}, $\tilde{\cR}[\tilde{\cF}](\tilde{y})$ differs from
    \begin{equation*}
      \Enc_{C(q,u,n,k,K_{\mathrm{out}})}\circ\bar{\cR}(x) = \phi^{\sqcup [n]^u}\circ\Enc_{C(r,K_{\mathrm{out}},\bar{N})}\circ\bar{\cR}(x)
    \end{equation*}
    at $<\tilde{\lambda}_{\mathrm{out}}$ dits. Then by \Cref{eq:mttranrestricts,eq:mtlamcomp}, we conclude that $\cR[\cF](y)$ differs from $\Enc_{C(r,K_{\mathrm{out}},\bar{N})}\circ\bar{\cR}(x)$ at $<\kappa\cdot\tilde{\lambda}_{\mathrm{out}}\leq\lambda_{\mathrm{out}}$ dits. Thus \Cref{it:fterr} in \Cref{def:faulttol} holds for the gadget in \Cref{eq:mtgad}.
  \end{proof}
\end{proof}

\section{Probabilistically Checkable Proofs}
\label{sec:pcp}
In \Cref{thm:pcp} below, we apply \Cref{thm:main} to obtain probabilistically checkable proofs (PCPs) for circuit-satisfiability (and therefore for $\NP$) with almost-linear size and inverse-polylogarithmic soundness. In this section, we work over the binary alphabet $\bF_2$ for simplicity in notation, but our techniques and results also apply to larger alphabets. We do still use the term ``dit'' instead of ``bit'' for consistency with the rest of the paper. We begin by defining CSPs\footnote{Our definition of CSPs is sufficient for our purposes, though more general definitions are useful in other settings.}.

\begin{definition}
  A \emph{constraint satisfaction problem (CSP) instance $\cP$} with $n$ \emph{variables} and $m$ \emph{constraints} consists of a (multi)set $\cP=\{\cP_1,\dots,\cP_m\}$ of $m$ functions called \emph{constraints}, where for each $j\in[m]$, the constraint $\cP_j:\bF_2^{B_j}\rightarrow\bF_2$ is a function for some $B_j\subseteq[n]$. The instance $\cP$ has \emph{locality $\ell$} if every $|B_j|\leq\ell$, and if every $i\in[n]$ is contained in at most $\ell$ of the sets $B_j$.

  For an assignment $x\in\bF_q^{[n]}$ of elements of $[n]$ to values in $\bF_2$, a constraint $\cP_j$ is \emph{satisfied} if $\cP_j(x|_{B_j})=0$. For $\delta>0$, if $\geq(1-\delta)m$ of the constraints are satisfied, we say $x$ is a \emph{$\delta$-approximate solution} to $\cP$. A $0$-approximate solution is simply called a \emph{solution} or a \emph{satisfying assignment}; a CSP instance $\cP$ with a satisfying assignment is called \emph{satisfiable}.
\end{definition}

Below, recall that the \emph{size} of a circuit $\cR$ using a gate set $\cG$ is the number of input dits plus the total number of non-identity gates in the circuit.

\begin{theorem}[PCPs for circuit-satisfiability with inverse-polylogarithmic soundness]
  \label{thm:pcp}
  Fix an arbitrary $\epsilon>0$, and fix an arbitrary set $\cG$ of gates, each of which acts on $O(1)$ bits (i.e.~2-ary dits). Then for some $s_0(\epsilon)>0$, there exists a polynomial-time algorithm that takes as input a circuit $\cR$ using gate set $\cG$ of size $s\geq s_0(\epsilon)$ that has a single output dit, and outputs a CSP instance $\cP$ with the following properties. $\cP$ has $n\leq s^{1+\epsilon}$ variables, $m=\Theta(n)$ constraints, and locality $O(1)$. Furthermore, we have:
  \begin{enumerate}
  \item (Completeness) If there exists an input $x$ such that $\cR(x)=0$, then $\cP$ is satisfiable.
  \item (Soundness) If every input $x$ has $\cR(x)=1$, then $\cP$ has no $1/(\log m)^{O(1/\epsilon)}$-approximate solution.
  \end{enumerate}
\end{theorem}
\begin{proof} 
  We will first construct a constant-depth circuit $\bar{\cR}$ that takes as input a transcript of the execution of $\cR$, and outputs a vector of all-$0$s iff the transcript describes a fault-free execution of $\cR$ on some input $x$ with output $\cR(x)=0$. We will then apply \Cref{it:mainfd} in \Cref{thm:main} to construct a mending fault-detecting gadget $\tilde{\cR}$ for $\bar{\cR}(\cdot)$.
  Our final CSP instance $\cP$ will have variables corresponding to dits in the transcript of $\tilde{\cR}$ on some input, and constraints verifying that this transcript describes a fault-free execution of $\tilde{\cR}$ with all-$0$s output.
  The mending fault-detecting property implies that if few constraints of this CSP instance are violated, then the transcript indeed describes an execution of $\tilde{\cR}$ with a low-weight fault, and with some input $x$ that has all-$0$s output.
  Thus soundness follows; completeness holds by construction.

  We now present the details for how we construct the CSP instance $\cP$.
  Each stage of our construction will be explicit, and hence we will obtain a polynomial-time algorithm to construct~$\cP$.
  We begin by showing how to construct the circuit~$\bar{\cR}$. 

  \begin{claim}
    \label{claim:pcpbar}
    There exists a circuit $\bar{\cR}$ using space $|\bar{N}|=\Theta(s)$, time $\bar{T}=O(1)$, and gate set $\cG'=\{\gInit,\gTerm,\gX^*,\gCX^*,\gCCX^*\}$ such that $\bar{\cR}$ has an input yielding all-$0$s output iff $\cR$ has an input yielding output $0$. Furthermore, $\bar{\cR}$ can be constructed from $\cR$ in $\poly(s)$ time.
  \end{claim}
  \begin{proof}
    Let $\cR$ use space $N$ and time $T$. Let $S\subseteq N\times\{0,\dots,T\}$ be the set of all pairs $(j,t)\in N\times\{0,\dots,T\}$ such that either $t=0$ and dit $j$ is an input dit of $\cR$, or $t=T$ and $j=j_{\mathrm{out}}$ is the unique output dit of~$\cR$, or else $t\in\{1,\dots,T-1\}$ and dit $j$ is an active dit that is acted upon by a non-identity gate in timestep $t$.
    In words, $S$ contains the active dits across all timesteps of~$\cR$, excluding locations in which a dit idles, i.e.~experiences an identity gate. Furthermore, each non-identity gate $G\in\cG$ acts on $O(1)$ dits, and therefore contributes $O(1)$ elements to $S$ for every appearance of $G$ in $\cR$. Thus $|S|=\Theta(s)$.

    We define $\bar{\cR}$ to take as input a set of dits labeled by $S$. Let there be $g$ non-identity gates in $\cR$, labeled by the set $[g]$ (in an arbitrary order); we also include the gate acting on the output dit $j_{\mathrm{out}}$ in timestep $T$ as a ``non-identity gate,'' even if it acts as the identity. For $i\in[g]$, let the $i$th gate $G_i\in\cG$ act on dits $B_i\subseteq N$ in timestep $t_i\in[T]$. Let $S_i^1\subseteq B_i\times\{t_i\}$ be the set of output dits of $G_i$. Let $S_i^0\subseteq N\times[T]$ contain the input dits of $G_i$ in their most recent non-idling timesteps; that is, for each input dit $j\in B_i$, then $S_i^0$ contains the pair $(j,t)$ that lies in $S$ for the greatest possible value of $t<t_i$. We then let $S_i=S_i^0\sqcup S_i^1$. In words, $S_i\subseteq S$ contains the input and output dits of the $i$th gate in the transcript of $\cR$, where we exclude idling locations.

    We then let $\bar{\cR}$ have $g+1$ output dits labeled by $\{0,\dots,g\}$. On input $z\in\bF_2^S$, the $0$th output $\bar{\cR}(z)_0$ equals the input dit $z_{(j_{\mathrm{out}},T)}$, which corresponds to $\cR$'s output dit when viewing $z$ as a (partial) transcript for $\cR$. For $i\in[g]$, the $i$th output $\bar{\cR}(z)_i$ equals $1$ if $z|_{S_i^1}\neq G_i(z|_{S_i^0})$, and $0$ otherwise. That is, $\bar{\cR}(z)_i$ indicates if there was an error in the execution of $G_i$, again viewing $z$ as a (partial) transcript for $\cR$. Because each output dit of $\bar{\cR}$ only depends on $\max_i|S_i|\leq\max_i2|B_i|=O(1)$ input dits, and each input dit affects at most $2$ output dits, we can implement $\bar{\cR}$ with gate set $\cG'$ (or any universal gate set) using space $O(|S|)=O(s)$ and time $O(1)$. Specifically, we color the $g+1$ output dits of $\bar{\cR}$ with a constant number of colors, such that no two output dits that both depend on some shared input dit have the same color. We then loop through the colors sequentially, and for each color compute all outputs with that color in parallel.

    There is a polynomial-time algorithm to color the vertices of a graph with maximum degree $\Delta$ using $\Delta+1$ colors such that no two adjacent vertices share a color; simply sequentially assign each vertex a distinct color from its previously assigned neighbors. Thus $\bar{\cR}$ can be constructed from $\cR$ in $\poly(s)$ time.

    By construction, if $\cR$ has some input $x$ with $\cR(x)=0$, then $\bar{\cR}(\tran(\cR;x)|_S)=0^{g+1}$. Conversely, if every input $x$ for $\cR$ has output $\cR(x)=1$, then every $z\in\bF_2^S$ has $\bar{\cR}(z)\neq 0^{g+1}$. For if instead some $z\in\bF_2^S$ had $\bar{\cR}(z)=0^{g+1}$, then $z$ would provide a transcript for the execution for $\cR$ on input $x=z|_{N\times\{0\}}$ with output $\cR(x)=z_{(j_{\mathrm{out}},T)}=\bar{\cR}(z)_0=0$, a contradiction. Here the fact that $\cR(x)=z_{(j_{\mathrm{out}},T)}$ follows from the assumption that $\bar{\cR}(z)=0^{g+1}$, which implies that $z$ provides a fault-free transcript for the exeuction of $\cR$, where idling locations have been ``shortcutted'' (i.e.~removed from the transcript and from the checks that gates executed without fault). Thus $\bar{\cR}$ exhibits the desired properties, completing the proof of the claim.
  \end{proof}

  We now apply \Cref{it:mainfd} in \Cref{thm:main} to the circuit $\bar{\cR}$ using space $\bar{N}$ and time $\bar{T}$ in \Cref{claim:pcpbar} with alphabet size $r=2$. We set the $\epsilon$ parameter in \Cref{thm:main} to equal the value $\epsilon'=\epsilon/16$ for $\epsilon$ defined in \Cref{thm:pcp}, so that in \Cref{thm:main} we have
  \begin{equation}
    \label{eq:pcpM}
    M=M(\epsilon',\bar{N})\in[|\bar{N}|^{1+\epsilon'},|\bar{N}|^{1+4\epsilon'}].
  \end{equation}
  Then assuming $|\bar{N}|=\Theta(s)$ is sufficiently large (i.e.~larger than $\bar{N}_0(2,\epsilon')$ in \Cref{thm:main}, which we can ensure by setting $s_0(\epsilon)$ sufficiently large), \Cref{thm:main} provides a polynomial-time constructable mending fault-detecting gadget
  \begin{equation*}
    \left(\cR',\; \cE_{\mathrm{run}}=2^{N'}|_{\geq\lambda_{\mathrm{run}}},\; D_{\mathrm{in}},\; D_{\mathrm{out}},\; \cE_{\mathrm{det}}=\bigsqcup_{t'\in[T']}(2^{E_{\mathrm{det},t'}}|_{\geq\lambda_{\mathrm{det}}})\right)
  \end{equation*}
  for $\bar{\cR}(\cdot)$, using space $|N'|\leq(\log|\bar{N}|)^2\cdot M$, time $T'\leq\bar{T}\cdot(\log|\bar{N}|)^{8/\epsilon'}$, and gate set $\cG'$, where
  \begin{align*}
    D_{\mathrm{in}} &= (C_{\mathrm{in}},\; \Enc_{\mathrm{in}},\; \cE_{\mathrm{in}}=2^{[n(2,\epsilon',\bar{N})]}|_{\geq\lambda_{\mathrm{in}}}) \\
    D_{\mathrm{out}} &= (C_{\mathrm{out}},\; \Enc_{\mathrm{out}},\; \cE_{\mathrm{out}}=2^{[n(2,\epsilon',\bar{N})]}|_{\geq\lambda_{\mathrm{out}}})
  \end{align*}
  are decorated codes of length $n(2,\epsilon',\bar{N})$ and distance $d(2,\epsilon',\bar{N})$, such that
  \begin{align*}
    d(2,\epsilon',\bar{N}) &\geq \frac{M}{2^{32}} \\
    \lambda_{\mathrm{run}} &= \frac{M}{2^{64}\cdot(\log|\bar{N}|)^{16\eta/\epsilon'}} \\
    \lambda_{\mathrm{out}} &= (\log|\bar{N}|)^{16\eta/\epsilon'}\cdot\lambda_{\mathrm{run}} = \frac{M}{2^{64}} \\
    \lambda_{\mathrm{det}} &= \lambda_{\mathrm{run}} = \frac{M}{2^{64}\cdot(\log|\bar{N}|)^{16\eta/\epsilon'}}.
  \end{align*}
  Recall above that $\eta>1$ is the absolute constant in \Cref{lem:loctest}.

  Let $S'\subseteq N'\times\{0,\dots,T'\}$ denote the set of all active dits of $\cR'$ across all timesteps, so that transcripts of $\cR'$ lie in $\bF_2^{S'}$. We define the CSP instance $\cP$ to act on a collection $y\in\bF_2^{S'}$ of $|S'|$ variables labeled by elements of $S'$, with the following three collections of constraints:
  \begin{enumerate}
  \item (Fault constraints) For every gate $G\in\cG'$ that $\cR'$ applies to dits $B\subseteq N'$ in some timestep $t'\in[T']$ with input dits $S^0\subseteq B\times\{t'-1\}$ and output dits $S^1\subseteq B\times\{t'\}$, we impose a constraint that has value $0$ if $y|_{S^1}=G(y|_{S^0})$, and value $1$ if $y|_{S^1}\neq G(y|_{S^0})$.
  \item (Detector constraints) For every $t'\in[T']$ and every detector dit $j'\in E_{\mathrm{det},t'}$, we impose a constraint that has value $y_{(j',t')}$.
  \item (Output constraints) For every output dit $(j',T')$ of $\cR'$, we impose a constraint that has value $y_{(j',T')}$.
  \end{enumerate}

  Assuming $s_0(\epsilon)$ and therefore $s$ and $|\bar{N}|$ are sufficiently large, the number $|S'|$ of variables in $\cP$ satisfies
  \begin{align*}
    |S'|
    &\leq |N'|\cdot T' \leq M\cdot(\log|\bar{N}|)^{O(1/\epsilon)} \\
    &\leq |\bar{N}|^{1+\epsilon/4}\cdot(\log|\bar{N}|)^{O(1/\epsilon)} \\
    &\leq s^{1+\epsilon} \\
    |S'|
    &\geq n(2,\epsilon',\bar{N}) \geq d(2,\epsilon',\bar{N}) \geq \frac{M}{32}.
  \end{align*}
  By definition, each variable participates in at most two fault constraints, one detector constraint, and one output constraint. Similarly, as each gate in $\cG'$ acts on $\leq 3$ dits, each constraint depends on $\leq 6$ variables. Thus $\cP$ has
  \begin{equation}
    \label{eq:numconst}
    \Omega(M) \leq m = \Theta(|S'|) \leq M\cdot(\log|\bar{N}|)^{O(1/\epsilon)}
  \end{equation}
  constraints, and locality $O(1)$, as desired. Also by definition, we can construct $\cP$ in polynomial time given the circuit $\cR'$ and the detector sets $E_{\mathrm{det},t'}$.

  In words, $\cP$ takes as input a transcript $y$ for $\cR'$, and checks that each gate was executed without fault, that all detector dits have value $0$, and that all output dits have value $0$. This definition immediately implies completeness:

  \begin{claim}
    The completeness claim in \Cref{thm:pcp} holds.
  \end{claim}
  \begin{proof}
    If there exists an input $x$ for $\cR$ with output $\cR(x)=0$, then by \Cref{claim:pcpbar}, there exists an input $z$ to $\bar{\cR}$ with output $\bar{\cR}(z)=0^{g+1}$. Then by the definition of a fault-detecting gadget, letting $z'=\Enc_{\mathrm{in}}(z)$, we have
    \begin{equation*}
      \cR'(z') = \Enc_{\mathrm{out}}\circ\bar{\cR}(z) = \Enc_{\mathrm{out}}(0^{g+1}) = 0^{n(2,\epsilon',\bar{N})},
    \end{equation*}
    where the final equality above holds because every linear code maps the all-$0$s input to the all-$0$s output. The fault-detecting property also implies that the transcript $\tran(\cR';z')$ has value $0$ on the restriction to every detector set $E_{\mathrm{det},t'}\times\{t'\}$. Thus all constraints have value $0$ for the assignment $z'=\Enc_{\mathrm{in}}(z)$, so $\cP$ is satisfiable.
  \end{proof}

  We now conclude the proof of \Cref{thm:pcp} by proving soundness:
  \begin{claim}
    The soundness claim in \Cref{thm:pcp} holds.
  \end{claim}
  \begin{proof}
    Assume that every input $x$ for $\cR$ has output $\cR(x)=1$. Let $y\in\bF_2^{S'}$ be an arbitrary assignment of the variables of $\cP$. We will show that $\cP$ has at least $m/(\log m)^{O(1/\epsilon)}$ unsatisfied constraints on assignment $y$.
    
    Let $z'\in\bF_2^{n(2,\epsilon',\bar{N})}$ equal the restriction of $y$ to the input dits of $\cR'$ in $N'\times\{0\}$. Then we can express $y$ as the transcript $y=\tran(\cR'[\cF'];z')$ of a noisy execution of $\cR'$, where we define the fault~$\cF'$ as follows. For every gate $G\in\cG'$ that $\cR'$ applies to dits $B\subseteq N'$ in some timestep $t'\in[T']$ with input dits $S^0\subseteq B\times\{t'-1\}$ and output dits $S^1\subseteq B\times\{t'\}$, if $y|_{S^1}\neq G(y|_{S^0})$, we let $\cF'$ set the value of dits $S^1$ to be $G(y|_{S^0})$ in timestep $t'$. If instead $y|_{S^1}=G(y|_{S^0})$, then $\cF'$ has no support inside (i.e.~does not corrupt) $S^1$ in timestep $t'$.

    Then $y=\tran(\cR'[\cF'];z')$ by induction on $t'=0,\dots,T'$. Indeed, for the base case, by definition $y|_{N'\times\{0\}}=z'=\tran(\cR'[\cF'];z')|_{N'\times\{0\}}$. For the inductive step, for $t'\in[T']$, assuming $y|_{N'\times\{t'-1\}}=\tran(\cR'[\cF'];z')|_{N'\times\{t'-1\}}$, then for every dit $j'$ in which the gates at time $t'$ do not set the dit's value to $y_{(j',t')}$, the fault $\cF'$ sets the value of dit $j'$ to be $y_{(j',t')}$.

    Furthermore, by definition for every set $S^1\subseteq B\times\{t'\}\subseteq N'\times[T']$ described above on which~$\cF'$ has nontrivial support, then by definition the fault constraint associated to this set of dits is not satisfied on assignment $y$. Because $|S^1\leq|B|\leq 3$ as every gate $G\in\cG'$ acts on $\leq 3$ dits, it follows that at least $|\cF'|/3$ fault constraints are not satisfied. Thus if $|\cF'|\geq\lambda_{\mathrm{run}}$, then the number of unsatisfied constraints is at least
    \begin{align*}
      \frac{\lambda_{\mathrm{run}}}{3}
      &\geq \frac{M}{(\log|\bar{N}|)^{O(1/\epsilon)}} \geq \frac{m}{(\log m)^{O(1/\epsilon)}},
    \end{align*}
    as desired, where the second inequality above holds by \Cref{eq:pcpM,eq:numconst}.

    Therefore it only remains to consider the case where $|\cF'|<\lambda_{\mathrm{run}}$. In this case, the mending fault-detecting property ensures that either there exists $z\in\bF_2^S$ such that $\cR'[\cF'](z')$ is a $\cE_{\mathrm{out}}$-deviation of $\Enc_{\mathrm{out}}\circ\bar{\cR}(z)$, or else there exists $t'\in[T']$ such that $|\tran(\cR'[\cF'];z')|_{E_{\mathrm{det},t'}\times\{t'\}}|\geq\lambda_{\mathrm{det}}$. In the first case, by the assumption that every input $x$ to $\cR$ has $\cR(x)=1$, then \Cref{claim:pcpbar} implies that $\bar{\cR}(z)\neq 0^{g+1}$, and hence $\Enc_{\mathrm{out}}\circ\bar{\cR}(z)$ is a nonzero codeword of $C_{\mathrm{out}}$, which must have weight $\geq d(2,\epsilon',\bar{N})$. Therefore
    \begin{align*}
      |\cR'[\cF'](z')|
      &\geq d(2,\epsilon',\bar{N})-\lambda_{\mathrm{out}} \geq \frac{M}{2^{64}} \geq \frac{m}{(\log|\bar{N}|)^{O(1/\epsilon)}}.
    \end{align*}
    Because $\cR'[\cF'](z')$ equals the restriction of $y=\tran(\cR'[\cF'];z')$ to the output dits in $N'\times\{T'\}$, all of which have output constraints checking that their values are $0$, we must have at least $m/(\log|\bar{N}|)^{O(1/\epsilon)}$ unsatisfied constraints, as desired.

    The only remaining case is if there exists $t'\in[T']$ such that $|\tran(\cR'[\cF'];z')|_{E_{\mathrm{det},t'}\times\{t'\}}|\geq\lambda_{\mathrm{det}}$. Because $y=\tran(\cR'[\cF'];z')$, in this case at least $\lambda_{\mathrm{det}}=\lambda_{\mathrm{run}}\geq m/(\log|\bar{N}|)^{O(1/\epsilon)}$ detector constraints are unsatisfied.

    Thus we have shown that if every input $x$ for $\cR$ has output $\cR(x)=1$, then every assignment $y\in\bF_2^{S'}$ for $\cP$ has $\geq m/(\log m)^{O(1/\epsilon)}$ unsatisfied constraints, as desired.
  \end{proof}    
\end{proof}

\section*{Acknowledgments}
We thank Venkatesan Guruswami for helpful discussions.

A.A.~acknowledges support through the NSF Award Nos.~2238836, 2430375 and QCIS-FF: Quantum Computing \& Information Science Faculty Fellow at Harvard University (NSF 2013303).
N.P.B.~acknowledges support through the Quantum Research Pod, CIQC and the Simons Institute for the Theory of Computing.
L.G.~acknowledges support from ONR grant N00014-24-1-2491, a UC Noyce initiative award, a Google PhD Fellowship, and a National Science Foundation Graduate Research Fellowship under Grant No.~DGE 2146752.
This work was done in part while the authors were attending the Simons Institute for the Theory of Computing, with support from NSF QLCI Grant No.~2016245.

\textbf{AI use disclosure:} All proofs and writing are the work of the human authors. ChatGPT was used to check the final draft for typos and minor presentation issues.

\bibliographystyle{alpha}
\bibliography{papers,library}

\end{document}